\documentclass[a4paper,11pt]{article}
\pdfoutput=1 

\usepackage{jheppub} 
\usepackage{amsmath,amssymb,amsthm,cancel,framed}
\usepackage{graphicx}
\usepackage{subcaption}
\usepackage[percent]{overpic}
\usepackage[dvipsnames]{xcolor}
\usepackage{multirow}
\usepackage{bbm}
\usepackage{todonotes}
\usepackage{tabularx}
\usepackage{tikz}
\usetikzlibrary{calc}
\usetikzlibrary{decorations.markings}
\usetikzlibrary{arrows,positioning,cd,calc,math,decorations.pathreplacing,decorations.markings,decorations.pathmorphing,patterns
}
\tikzset{wave/.style={decorate, decoration=snake}}
\usepackage{pgfplots}
\usetikzlibrary{decorations.markings,pgfplots.fillbetween}
\definecolor{darkspringgreen}{rgb}{0.05, 0.5, 0.06}
\definecolor{MyBlue}{rgb}{0.25,0.25,0.75}
\definecolor{MyRed}{rgb}{0.75,0.25,0.25}
\colorlet{NextBlue}{MyBlue!20}
\colorlet{SecondBlue}{MyBlue!40}
 \definecolor{shadecolor}{rgb}{0.9, 0.9, 0.86}

\definecolor{purple}{rgb}{0.7,0,0.7}

\def\be{\begin{equation}}
\def\ee{\end{equation}}

\newcommand{\CS}{\mathcal{S}}
\newcommand{\CO}{\mathcal{O}}

\newcommand{\CB}{\mathcal{B}}

\newcommand{\CM}{\mathcal{M}}
\newcommand{\CX}{\mathcal{X}}
\newcommand{\CH}{\mathcal{H}}
\newcommand{\CN}{\mathcal{N}}

\newcommand{\CW}{\mathcal{W}}

\newcommand{\tSigma}{{\widetilde{\Sigma}}}

\newcommand{\fs}{{{\mathfrak{s}}}}

\newcommand{\IC}{\mathbb{C}}
\newcommand{\IP}{\mathbb{P}}

\newcommand{\IZ}{\mathbb{Z}}

\newcommand{\IF}{\mathbb{F}}

\newcommand{\IR}{\mathbb{R}}

\renewcommand{\-}{\text{-}}
\renewcommand{\(}{\left(}
\renewcommand{\)}{\right)}

\renewcommand{\ker}{\mathrm{ker}}
\newcommand{\Hom}{\mathrm{Hom}}

\newcommand{\diag}{\mathrm{diag}}

\newcommand{\arccosh}{\mathrm{arccosh\,}}
\newcommand{\A}{\mathrm{A}}
\newcommand{\fb}{\mathfrak{b}}
\newcommand{\CL}{\mathcal{L}}
\newcommand{\CR}{\mathcal{R}}
\newcommand{\CV}{\mathcal{V}}
\newcommand{\CP}{\mathcal{P}}

\newcommand{\res}{\mathrm{res}}
\newcommand{\Uij}{T}
\newcommand{\CU}{{Y}}

\renewcommand{\Im}{{\rm Im\,}}
\newcommand{\dd}{\mathrm{d}}

\newcommand{\tr}{\mathrm{tr}}

\newcolumntype{C}[1]{>{\centering\arraybackslash}m{#1}}

\newcommand{\R}{\mathcal{R}}

\newtheorem{remark}{Remark}

\newtheorem{conjecture}{Conjecture}
\newtheorem{proposition}{Proposition}
\newtheorem{theorem}{Theorem}
\newtheorem{corollary}{Corollary}
\newtheorem{definition}{Definition}
\newtheorem{assumption}{Assumption}

\title{Monodromies of Second Order $q$-difference Equations from the WKB Approximation}

\author[a]{Fabrizio  Del Monte}
\emailAdd{f.delmonte.mp@gmail.com}

\author[b,c]{Pietro Longhi}

\affiliation[a]{School of Mathematics and Statistics, University of Sheffield, Hounsfield Road, Sheffield S3 7RH, United Kingdom}

\affiliation[b]{Department of Physics and Astronomy, Uppsala University, Box 516, 751 20 Uppsala, Sweden}

\affiliation[c]{Department of Mathematics, Uppsala University, Box 480, 751 06 Uppsala, Sweden}

\affiliation[d]{Centre for Geometry and Physics, Uppsala University, Box 480, 751 06 Uppsala, Sweden}

\emailAdd{
pietro.longhi@physics.uu.se
}

\abstract{
This paper studies the space of monodromy data of second order $q$-difference equations through the  framework of WKB analysis. We compute the connection matrices associated to the Stokes phenomenon of WKB wavefunctions and develop a general framework to parameterize monodromies of $q$-difference equations. Computations of monodromies are illustrated with explicit examples, including a $q$-Mathieu equation and its degenerations. 
In all examples we show that the monodromy around the origin of $\mathbb{C}^*$ admits an expansion in terms of Voros symbols, or exponentiated quantum periods, with integer coefficients. 
Physically these monodromies correspond to expectation values of Wilson line operators in five dimensional quantum field theories with minimal supersymmetry. In the case of the $q$-Mathieu equation, we show that the trace of the monodromy can be identified with the Hamiltonian of a corresponding $q$-Painlev\'e equation.
}

\begin{document} 

\hfill UUITP-14/24

\maketitle
\flushbottom

\section{Introduction}

The study of monodromy data of differential equations is ubiquitous both in Physics and in Mathematics, with wide ranging applications across Quantum Field Theory and String Theory.
Recent developments at the intersection of String Theory, Algebraic Geometry, Integrable Systems and Conformal Field Theory have motivated an interest in the related, but less-studied, subject of $q$-difference equations.

Foundational results on (generalised) character varieties of $q$-difference equations
date back to works of Birkhoff and Carmichael in the early twentieth century \cite{Birkhoff1913,Carmichael1912}. However this subject is much less developed than its counterpart for differential equations, with interesting exceptions which include a recent approach based on   Riemann-Hilbert correspondences \cite{Ohyama2020,Joshi2023,Ramis2023}.%
\footnote{It should be noted that with such Riemann-Hilbert approach, the term ``monodromy" is a bit of a misnomer, as it is not related to multivaluedeness but rather to the connection between the singularity at zero and infinity on $\mathbb{C}^*$. With our approach, one has \textit{bona fide} monodromy and Stokes phenomenon, somewhat analogous to the case of differential equations. We leave the study of the relation between these two notions of monodromy spaces for $q$-difference equations to future work.}
In this paper we take a different approach, based on the WKB approximation, to the characterization of monodromy data of $q$-difference equations.

A motivation behind our approach is the connection between WKB analysis and cluster varieties in the context of differential equations \cite{FockGoncharovHigher,hikami2019note}. Applications of this relations have appeared in String Theory \cite{Gaiotto:2009hg}, two-dimensional Conformal Field Theory \cite{Gavrylenko:2020gjb}, Integrable Systems \cite{chekhov2017algebras,chekhov2017painleve} and Algebraic Geometry \cite{Bridgeland:2016nqw}. 
This connection has found many concrete realizations from a pure mathematical standpoint.
For example, in the context of second order differential equations it has been shown that Stokes phenomena of Borel summed Voros symbols (periods of WKB differentials) are governed by cluster transformations \cite{IwakiNakanishi1, Allegretti:2018kvc, nikolaev2023exact}. 

Analogous connections between $q$-difference equations and cluster theory have also been studied from several directions \cite{Franco:2011sz, DML2022,Bershtein2017,Bershtein2019,Mizuno2020}.
A concrete realization of this connection in the context of first order $q$-difference equations has been recently established in \cite{Alim:2021mhp,Alim:2022oll,Grassi:2022zuk}, where resurgent structures of quantum dilogarithms have been shown to govern Stokes phenomena.

In this paper we lay foundations for the program proposed in \cite[Section 3]{DML2022}, which aims to characterize monodromies of general $q$-difference equations by means of the WKB approximation. 
In this paper we focus on the setting of second order $q$-difference equations, which presents several new challenges compared to the first-order setting, from a technical standpoint. 
In the remainder of the introduction we give an overview of our approach and the main results.

\subsection*{Main results}

In Proposition \ref{prop:Involutive}, we show that any second-order $q$-difference equation ($q$DE) on $\mathbb{C}^*$ can be brought to the form
\begin{equation}\label{eq:CanForm}
    \psi(qx)+\psi(q^{-1}x)=2T(x,q)\psi(x),
\end{equation}
with the caveat that this can involve uplifting the equation to a cover of $\mathbb{C}^*$. When the singularities of the original equations are only $0,\infty$, this cover is trivial. With this in mind, we restrict throughout the paper to equations of this form. The WKB approximation is introduced by setting $q=e^{\hbar}$, and considering \eqref{eq:CanForm} as a singularly perturbed equation as $\hbar\rightarrow0$.

Theorem~\ref{thm:WKBSol} shows that the WKB approximation produces solutions as formal power series in $\hbar$, of the form 
\begin{equation}
	\psi_{\pm,n}(x)=e^{-\chi(x;\hbar)}\exp\left\{\frac{1}{\hbar}\int_{x_0}^x\log\R_{\pm,n}(x',\hbar)\frac{\dd x'}{x'} \right\}, \qquad n\in\mathbb{Z},
\end{equation}
with $\chi$ and $\log\CR$ explicitly determined as formal power series by a recursive formula. $\log\CR\frac{\dd x}{x}$ is a differential on a logarithmic cover of the WKB curve
\begin{equation}\label{eq:WKBCurveIntro}
    y+y^{-1}=2T(x,q=1)\subset\mathbb{C}^*_x\times\mathbb{C}^*_y.
\end{equation}

A consequence of this is that the WKB Stokes diagram consists of critical trajectories of the (leading order) WKB differential $\lambda_{\pm,n}=\lim_{\hbar\rightarrow 0}\log\CR(x,\hbar)\frac{\dd x}{x}$, and can be described by
\be\label{eq:Stokes-line-Im-zero}
	\Im\left[
	\hbar^{-1} \, \int^x\left(\lambda_{s_2,n_2}-\lambda_{s_1,n_1}\right)  
	\right] = 0.
\ee
The resulting Stokes diagram coincides with the geometric data underlying Exponential Networks \cite{Klemm:1996bj,Eager:2016yxd,Banerjee:2018syt,Banerjee:2019apt,Banerjee:2020moh}, a connection that played a central role in the context of first order $q$-difference equations \cite{Alim:2022oll,Alim:2021mhp,Grassi:2022zuk}. 
An advantage afforded by the relation to networks is that it can be adopted as a definition of the Stokes graph.
Stokes lines \eqref{eq:Stokes-line-Im-zero} can originate both at branch points of \eqref{eq:WKBCurveIntro} as in standard WKB theory, but also at logarithmic singularities whenever these appear at finite distance (we do not consider this case in this work).
As they evolve, Stokes lines typically intersect each other in nontrivial patterns. 
At intersections new Stokes lines can be generated. Whether this happens, and what kinds of new lines appear is established by rules motivated by physical and mathematical considerations in the context of exponential networks~\cite{Banerjee:2018syt}. 

A novel feature of Stokes lines for $q$-difference equations is that they carry an additional label  $\ell\in\mathbb{Z}$, related to the logarithmic nature of the WKB differential.
Working in a WKB basis of solutions, we find that Stokes matrices associated to Stokes lines have the following form
\begin{equation}\label{eq:Stokes-matirces-intro}
    S^{(\ell)}=\left( \begin{array}{cc}
        -\xi^\ell & i \\
        i & 0
    \end{array} \right),\qquad \xi:=\left(\frac{x}{\fb} \right)^{\frac{2\pi i}{\hbar}}.
\end{equation}
Here $\fb$ denotes the position of a branch point, which is the choice of basepoint that determines the normalization of WKB solutions.
The case $\ell=0$ coincides with the Stokes matrices familiar in the context of Voros' connection formulae \cite{Voros}. 

To compute monodromies and construct Voros symbols, it is essential to parameterize analytic continuation between solutions normalised at different branch points.
For branch points that share a common Stokes sector, we define the transport in terms of one of the following two types of matrices:
\begin{equation}\label{eq:transport-matirces-intro}
\left( \begin{array}{cc}
0 & i\CU_{\fb\fb'} \\
i\CU_{\fb\fb'} & 0
\end{array} \right), \qquad \text{ or }\qquad  
\left( \begin{array}{cc}
\CU_{\fb\fb'} & 0 \\
0 & \CU_{\fb\fb'}^{-1}
\end{array} \right).
\end{equation}

\begin{table}[h!]
\caption{$q$-difference equations, classical geometries, and monodromies.}  
\begin{center}
\begin{tabularx}{\textwidth}{|C{0.215\textwidth}|C{0.5\textwidth}|C{0.2\textwidth}|}
    \hline
    Geometry & $q$-difference equation & Monodromy \\
    \hline\hline
    \begin{tabular}{c}
        \includegraphics[width=0.15\textwidth]{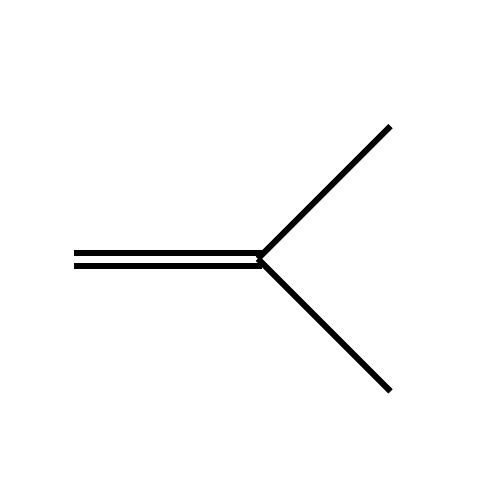} \\[2pt]
        $\IC\times (\IC^2/\IZ_2)$
    \end{tabular} &
    \begin{tabular}{c}
        $q$-Airy\\[4pt]
        $\psi(qx)+\psi(q^{-1}x)=2x\psi(x)$
    \end{tabular} &
    \begin{tabular}{c}
	\eqref{eq:TrAiry}
    \end{tabular} \\
    \hline
    \begin{tabular}{c}
        \includegraphics[width=0.15\textwidth]{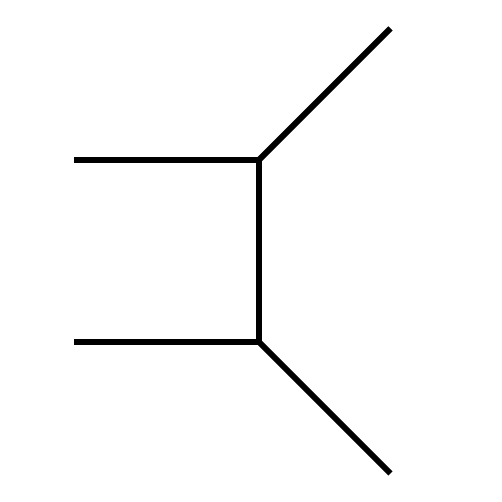} \\[4pt]
        $\CO(0)\oplus \CO(-2)$
    \end{tabular} &
    \begin{tabular}{c}
        $q$-Airy$_{\kappa}$\\[4pt]
        $\psi(qx)+\psi(q^{-1}x)=2(\kappa+x)\psi(x)$
    \end{tabular} &
    \begin{tabular}{c}
	\eqref{eq:half-geometry-monodromy}
    \end{tabular} \\
    \hline
    \begin{tabular}{c}
        \includegraphics[width=0.15\textwidth]{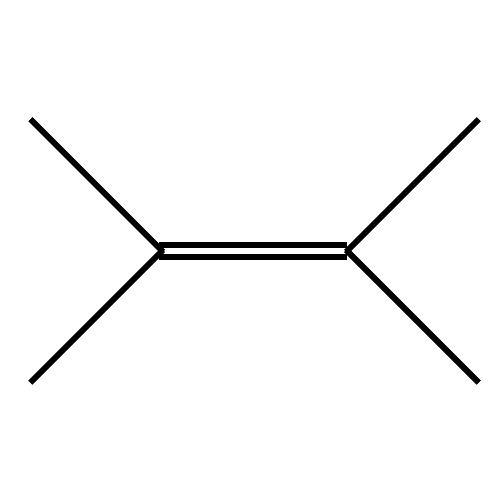} \\[4pt]
        Resolved conifold
    \end{tabular} &
    \begin{tabular}{c}
        $q$-hypergeometric (1,1)\\[4pt]
        $\psi(qx)+\psi(q^{-1}x) = 2Q^{-1}\left(x^{\frac{1}{2}}+x^{-\frac{1}{2}} \right)\psi(x)$
    \end{tabular} &
    \begin{tabular}{c}
	\eqref{eq:conifold-mondromy}, \eqref{eq:conifold-half-mondromy}
    \end{tabular} \\
    \hline
    \begin{tabular}{c}
        \includegraphics[width=0.15\textwidth]{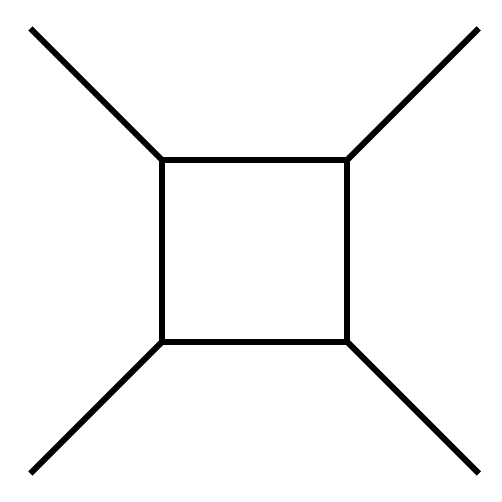} \\[4pt]
        Local $\mathbb{P}^1\times\mathbb{P}^1$
    \end{tabular} &
    \begin{tabular}{c}
        $q$-Modified Mathieu\\[4pt]
        $\psi(qx)+\psi(q^{-1}x)=\left(-\tau(x+x^{-1})+\kappa \right)\psi(x)$
    \end{tabular} &
	\eqref{eq:localF0-monodromy-final}
    \\
    \hline
\end{tabularx}
\end{center}
\label{tab:recap}
\end{table}

The Stokes matrices \eqref{eq:Stokes-matirces-intro} and the transport matrices \eqref{eq:transport-matirces-intro} provide the building blocks of a general formalism that allows to parametrize explicitly monodromy matrices of $q$-difference equations.\footnote{Analogous constructions are well-known in the context of differential equations, see for example \cite[Section 5.2]{Gavrylenko:2020gjb}.} 
In Section \ref{sec:Examples} we illustrate applications of this formalism with several examples, for which we compute the  of WKB solutions around the origin in $\IC^*$. 
The examples that we study are summarized in Table \ref{tab:recap}. The most interesting case is that of the $q$-Mathieu equation, whose monodromy has trace (in appropriate coordinates, whose meaning is spelled out in Section \ref{sec:local-F0})
\begin{equation}\label{eq:F0-monodromy-intro}
    \tr\,\mathbb{M}
	=
	\frac{1}{\sqrt{X_{{D2}_f}}}
	+ \sqrt{{X_{{D2}_f}}}
	+ \frac{\sqrt{X_{D2_f}}}{X_{D4}} 
	+ \frac{X_{D0}X_{D2_f}}{X_{{D2}_b}} \frac{X_{D4}}{\sqrt{X_{D2_f}}}
\end{equation}
Notably, this expression coincides with the Hamiltonian of the q-Painlevé equation describing the cluster integrable system associated to this WKB curve \cite{Bershtein2017} (surface type $A_7^{(1)'}$ in Sakai's classification \cite{Sakai2001}). We expect the q-Mathieu equation that we describe to be linked to the same q-Painlev\'e equation also in a different way, through the q-Painlev\'e-Heun correspondence \cite{Takemura2017,Sasaki2023}. This hints back to one of the motivations behind our approach, namely the connection to cluster varieties, whose development we leave to future work.

\paragraph{Some open questions.} 
There are several points that are raised throughout our work which deserve further attention.
One of these is the analysis of equations which feature singularities (punctures) at finite $x$ (other than $0,\,\infty$).
Another important matter is the parametrization of the connection matrix expressing the transport from $x=0$ to $x=\infty$, as this is crucial for a full characterization of the space of monodromies. 
Yet another natural question,
given the identification between our monodromy variables and cluster variables appearing in exponential networks,  is the study of the \textit{nonlinear} Stokes phenomenon for the monodromies, associated to jumps of the Stokes diagram itself.

\subsubsection*{Applications to Quantum Field Theory and String Theory}

Although our results solely rely on WKB analysis of $q$-difference equations, the motivations for this work and our interest in the subject come from the role played by $q$-difference equations in String Theory and in Geometry.
For the interested reader, here we comment on implications of our results in relation to these subjects.

Let $Y$ be a local Calabi-Yau threefold of the form
\begin{equation}\label{eq:MirrorCY3}
 \left\{uv=F(x,y)\right\}\subset \mathbb{C}^2_{u,v}\times(\mathbb{C}^*_{x,y})^2.
\end{equation}
The locus where the conic fibration $Y\rightarrow(\mathbb{C}^*_{x,y})^2$ degenerates is given by the algebraic curve $F(x,y)=0\subset(\mathbb{C}^*)^2$.
WKB curves of $q$-difference equations have precisely this form \eqref{eq:CanForm}, and can therefore be naturally associated with Calabi-Yau threefolds. 
Stokes diagrams, known as exponential networks in physics, arise in this context in the study of D3-branes wrapping special Lagrangian cycles on $Y$ \cite{Klemm:1996bj, Eager:2016yxd, Banerjee:2022oed}.
Stokes data encoded by matrices such as \eqref{eq:Stokes-matirces-intro} arises in physics in the construction of the 3d $tt^*$ connection \cite{Cecotti:2013mba}, whereby entries of the matrix are associated to counts of BPS states of 3d $\CN=2$ QFTs, known as \emph{kinky vortices} \cite{Banerjee:2018syt}.

Degenerations of the exponential network detect stable D3 branes. The spectrum of these BPS states corresponds to counts of special Lagrangian cycles in $Y$, and plays a central role in the (physical realization) of homological mirror symmetry. In fact the mirror of these are stable objects in the derived category of coherent sheaves on a mirror Calabi-Yau threefold $Y^\vee$. 
In string theory the curve $F(x,y)=0$ is naturally promoted to a $q$-difference operator due to worldsheet instanton corrections weighted by $g_s$ (with $q=e^{g_s}$) that deform the classical geometry of $Y$ \cite{Aganagic:2003qj}.

The $q$DEs that we study in this paper arise as examples of such quantum curves, and the corresponding exponential networks are studied extensively in \cite{Eager:2016yxd, Banerjee:2018syt, Banerjee:2019apt, Banerjee:2020moh, Longhi:2021qvz, Banerjee:2024smk}. The solutions to the $q$DEs can be thought as partition functions of five-dimensional gauge theories with insertion of a codimension-2 defect, or equivalently a toric brane insertion on an external leg of the toric diagram of the mirror geometry to \eqref{eq:MirrorCY3}, see \cite{Dimofte:2010tz} and references therein. 

Borel sums of Voros symbols, often called in this context \emph{quantum periods}, are then expected to parameterize the vacuum moduli space of a 5d $\CN=1$ QFT geometrically engineered by M theory on $Y^\vee\times T^2\times \IR^3$ \cite{Haghighat:2011xx, Alexandrov:2017mgi}.
This space is expected to be closely related to other types of moduli spaces that appear in the string theory literature, in particular hypermultiplet moduli spaces \cite{Alexandrov:2008gh}, the moduli space of doubly periodic monopoles \cite{Cherkis:2000cj}, and the moduli spaces of multiplicative Higgs bundles \cite{Elliott:2018yqm, Frassek:2018try}.
The quantum monodromies that we compute in this paper correspond in physics to expectation values of line operators in the 5d QFT, generalizing similar statements from the 4d setting obtained by reducing on $S^1$ \cite{Drukker:2009tz, Coman:2015lna}.
Moreover, monodromies can be expanded in the basis of quantum periods with integer coefficients, as expression \eqref{eq:F0-monodromy-intro} illustrates.
The coefficients of the expansion correspond to protected indices of \emph{Framed BPS states} \cite{Gaiotto:2010be}, and their computation by means of Stokes data is a generalization of similar constructions based on spectral networks \cite{Gaiotto:2011tf, Hollands:2013qza, Gabella:2016zxu, Gang:2017ojg}.

In the context of five-dimensional gauge theory, the WKB approximation relates to the so-called ``Nekrasov-Shatashvili'' limit of refined Topological Strings \cite{Aganagic:2011mi,Huang2013,Kashani-Poor2016}. An important observation in this regard is that, in order to obtain physically sensible results, it is essential to take into account certain additional nonperturbative corrections \cite{Kallen:2013qla,Grassi:2014zfa}. See \cite{Gu:2022fss} and references therein for some recent developments.

\subsection*{Acknowledgements}
We are grateful to Sibasish Banerjee, Tom Bridgeland, Harini Desiraju, Pavlo Gavrylenko, Alba Grassi, Kohei Iwaki, Marcos Mari\~no and Mauricio Romo for illuminating discussions. This research was partly conducted during FDM's visit to the Okinawa Institute of Science and Technology (OIST) through the Theoretical Sciences Visiting Program (TSVP), for the Thematic Program ``Exact Asymptotics: From Fluid Dynamics to Quantum Geometry''. FDM thanks the Galileo Galilei Institute for Theoretical Physics for the hospitality, during the program ``BPS Dynamics and Quantum Mathematics'', and the INFN for partial support during the completion of this work.
The work of PL is supported by the Knut and Alice Wallenberg Foundation grant KAW2020.0307 and by the VR grant to the Centre of Excellence in Geometry and Physics at Uppsala University.

\section{Exact WKB analysis of second order $q$-difference equations}\label{sec:2nd-order-qDE-WKB}

Consider a generic second order $q$-difference equation 
\begin{equation}\label{eq:2qDEgeneral}
\phi(qx)+a_1(x,q)\phi(x)+a_2(x,q)\phi(q^{-1}x)=0,
\end{equation}
where $a_i$ are rational functions of $(x,q)\in\mathbb{C}^*\times \mathbb{C}^*$ and possibly additional complex parameters.
\begin{remark}\label{rmk:$q$-periodicity}
	Due to linearity, any $q$-periodic function $\theta(qx)=\theta(x)$ gives rise to a new solution $\phi'(x) = \theta(x)\phi(x)$ to \eqref{eq:2qDEgeneral}. 
	It is therefore natural to consider the space of solutions to linear $q$-difference equations modulo $q$-periodic functions.
	Let $\mathcal{M}_q$ denote the field of $q$-periodic functions in $x$.
	We denote by $\CH$ the space of solutions to \eqref{eq:2qDEgeneral} as a vector space over $\mathcal{M}_q$.
\end{remark} 
Our first step will be to transform \eqref{eq:2qDEgeneral} into a more convenient form:
\begin{proposition}\label{prop:Involutive}
	Let $a_1(x),\,a_2(x)$ be rational functions in $x$. Through an explicit change of variables $\phi\rightarrow\psi$, the general second order $q$DE \eqref{eq:2qDEgeneral} can always be transformed in its \textit{involutive form}
	\begin{equation}
		\psi(qx)+\psi(q^{-1}x)=2T(x,q)\psi(x)
	\end{equation}
\end{proposition}
\begin{proof}
The $q$-difference equation \eqref{eq:2qDEgeneral} can be written as a $2\times2$ linear system
\begin{equation}
	\Phi(qx)=A(x)\Phi(x),\qquad \Phi(x):=\left( \begin{array}{cc}
		\phi_1(x) & \phi_2(x) \\
		\phi_1(qx) & \phi_2(qx)
	\end{array} \right) \qquad A(x)=\left( \begin{array}{cc}
			0 & 1 \\
			-a_2(qx) & -a_1(qx)
		\end{array} \right)
\end{equation}
where $\phi_1,\,\phi_2$ are linearly independent solutions of \eqref{eq:2qDEgeneral} over the field of $q$-periodic functions $\mathcal{M}_q$. Set
\begin{equation}
	\Phi(x)=f(x)\Psi(x) 
\end{equation}
where $f(x)$ is a scalar factor satisfying
\begin{equation}
	f(qx)=\sqrt{\det A(x)}f(x)=\sqrt{a_2(qx)}f(x).
\end{equation}
Since $a_2$ is a rational function, this results in the $q$-difference equation for $f$
\begin{equation}
	f(qx)=a_0\prod_i(x-x_i)^{k_i},\qquad k_i\in\mathbb{Z},
\end{equation}
which is easily seen to have solution
\begin{equation}
	f(x)=x^{\frac{\log a_0}{\log q}} \prod_i f_i(x), \qquad f_i=\begin{cases}
	e^{\frac{k_i}{2}\log x\left(\frac{\log x}{\log q}-1 \right)}	& x_i=0, \\
	e^{k_i\log x_i\frac{\log x}{\log q}}(x_k/x;q)_\infty^{k_i},	& x_i\ne0.
	\end{cases}
\end{equation}
Here $(x;q)_\infty:=\prod_{k=1}^\infty(1-xq^k)$ is the infinite q-Pochhammer symbol, satisfying
\begin{equation}
    (qx;q)_{\infty}=(1-x)^{-1}(x;q)_{\infty}.
\end{equation}
\end{proof}
We will henceforth restrict to equations of this form.
\begin{definition}
	We define the involutive form of a second order $q$-difference equation to be
	\begin{equation}\label{eq:2qDEcan}
		\psi(qx)+\psi(q^{-1}x)=2T(x,q)\psi(x).
	\end{equation}
All information retained from \eqref{eq:2qDEgeneral} is encoded in theterm $T(x,q)$, that we will call the potential of the $q$DE, or its trace function.
\end{definition}
\subsection{The associated linear system}
We will sometimes find it useful to write the second order $q$-difference equation \eqref{eq:2qDEcan} as a first order $q$-difference linear system
\begin{equation}\label{eq:qLinSys}
\Psi(qx)=A(x)\Psi(x),
\end{equation}
where $A(x)$ is a $2\times 2$ matrix,
\begin{equation}
A(x)=\left( \begin{array}{cc}
0 & 1 \\
-1 & 2T(qx)
\end{array} \right),
\end{equation} 
$\Psi$ is the fundamental matrix of solutions
\be\label{eq:matrix-solutions}
	\Psi(x):=\left(\begin{array}{cc}\psi_1(x) & \psi_2(x)\\ \psi_1(qx) & \psi_2(qx)\end{array} \right),
\ee
with $\psi_1,\,\psi_2$ linearly independent solutions to \eqref{eq:2qDEgeneral} over $\mathcal{M}_q$. 

A $q$-gauge transformation acts on \eqref{eq:qLinSys} from the left, by replacing $\Psi, A$ with
\be
	\Psi'(x) = g(x)\cdot \Psi(x)\,,
	\qquad
	A'(x) = g(qx) \cdot A(x)\cdot g(x)^{-1}\,.
\ee
When the gauge transformation is $q$-periodic, i.e. $g(qx)=g(x)$ is valued in $GL(2,\mathcal{M}_q)$, it specializes to a global change of basis over the field of $q$-periodic functions (the analogue of constant gauge transformations in the setting of differential equations).

\subsection{WKB ansatz}
We will treat \eqref{eq:2qDEcan} as a singularly perturbed equation in the parameter $\hbar:=\log q\rightarrow 0$, and it will be 
assumed that $T(x,q)$ has a regular expansion around $\hbar=0$
\begin{equation}
T(x,\hbar)=\sum_{k=0}^{\infty}T_k(x)\hbar^k.
\end{equation}
To solve \eqref{eq:2qDEcan} we consider the WKB ansatz 
\begin{align}\label{eq:WKBAns}
\psi(x,\hbar)=\exp\left(\int_{x_0}^x S(x,\hbar)\frac{\dd x}{x} \right), && S(x,\hbar)=\sum_{k=-1}^{\infty}S_k(\log q)^k=\sum_{k=-1}^{\infty}S_k\hbar^k\,,
\end{align}
where $x_0$ encodes a choice of normalization.
The leading order term obeys
\begin{equation}\label{eq:WKBCurveS}
y+y^{-1}=2T_0(x),\qquad y:=e^{S_{-1}(x)}\,.
\end{equation}
This equation defines the \textit{WKB curve} $\Sigma$ as an algebraic curve in $\mathbb{C}^*\times \mathbb{C}^*$. 
In terms of the coefficients of \eqref{eq:2qDEgeneral}, we have 
\begin{equation}
T_0(x)= \lim_{q\to 1}-\frac{a_1(x,q)}{2\sqrt{a_2(x,q)}}.
\end{equation}
Due to the logarithm in the definition of $S_{-1}=\log y$, 
WKB solutions are labeled both by a discrete index $s=\pm$ for solutions of the algebraic curve in $y$
\be
	y_\pm:=T_0\pm \sqrt{T_0^2-1}
\ee
and by a logarithmic index $n$ that defines $S_{-1}$. 
In the limit $\hbar\to 0$ the WKB solutions have the following asymptotics \cite{DML2022}
\begin{align}\label{eq:psipmnleading}
	\psi_{\pm,n}(x)\mathop{\sim}_{\hbar\rightarrow0}\exp\left\{\frac{1}{\hbar}\int^x_{x_0}\lambda_{\pm, n} \right\}
\end{align}
where
\be\label{eq:lambda-pm-n}
	\lambda_{\pm,n}:= S_{-1}(x) \frac{\dd x}{x} = \left(\log y_\pm +2\pi i n \right)\frac{\dd x}{x}=\left(\pm\log y_++2\pi i n \right)\frac{\dd x}{x} \,.
\ee
Due to the presence of the logarithm function, $\lambda_{\pm,n}$ are not differentials on the curve $\Sigma$, but rather on its infinite $\mathbb{Z}$-cover induced by $\log y\rightarrow y$.

As observed in remark \ref{rmk:$q$-periodicity}, while $\psi_{\pm,n}$ provide infinitely many linearly independent solutions over $\mathbb{C}$, only two among these solutions are independent over the field $\mathcal{M}_q$ of $q$-periodic functions. We will see later (see Section \ref{Sec:AllOrdWKB}) that any two choices of the logarithmic index are related by multiplication by simple $q$-periodic functions, and that
the dependence on the logarithmic index $n$ enters only through the semiclassical behaviour in \eqref{eq:psipmnleading}.

\subsection{General structure of the involutive WKB curve}\label{Sec:CritPtsCurve}

There are three types of points on the WKB curve that play a distinguished role in our analysis.
\paragraph{Square-root branch points.}
For WKB curves in involutive form \eqref{eq:WKBCurveS} these are defined by the condition $y = y^{-1}$. The location of these points is therefore determined by 
\be
	T_0(x) = \pm 1.
\ee
We shall always assume that branch points are isolated, which means that locally $y\sim \pm (1 + c_0 \sqrt{x-x_0} +\dots)$ and therefore
\be\label{eq:T0-branch-point}
	T_0 
	= \frac{1}{2}(y+y^{-1}) 
	= \pm \left(1 + \frac{c_0^2}{2}  (x-x_0) + \dots\right)
\ee
At square-root branch points the two sheets $y^{\pm 1}(x)$ coincide, and get exchanged when going around a loop around the branch point.
Consequently, upon crossing a square-root branch cut, the two families of solutions are also exchanged:
\be
	\psi_{\pm,n} \to \psi_{\mp,n}\,.
\ee

\paragraph{Logarithmic branch points.}
Another set of distinguished points are those where one of the sheets $y(x)$ runs to infinity (correspondingly, at the same value of $x$ the other sheet must go to zero). 
These correspond to singularities $x=x_*$ of the equations (they are points where the determinant of the q-connection $A(x)$ vanishes). In this section we will consider $x_*\ne0,\infty$. In terms of $T_0(x)$, they are characterized by the degree of divergence $k$
\be\label{eq:log-bp-T0}
	T_0(x) \sim (x-x_*)^k + \dots\qquad (k<0)
\ee
Near $x=x_*$ we fix a labeling of sheets such that $y_+$ is the one that becomes large, behaving as $y_+(x) \approx 2 T_0(x)$ while $y_-(x)\approx (2 T_0(x))^{-1}$ goes to zero.
If follows that on $\mathbb{C}^\star$, the multivalued functions $\log y_\pm(x)$ undergo additive monodromy around a logarithmic singularity of $T_0$ 
\be
	\log y_\pm(e^{2\pi i } w) = \log y_{\pm}(w) \pm 2\pi i \, k
	\qquad
	(w=x-x_*).
\ee

\begin{figure}[h!]
\begin{center}
\includegraphics[width=0.4\textwidth]{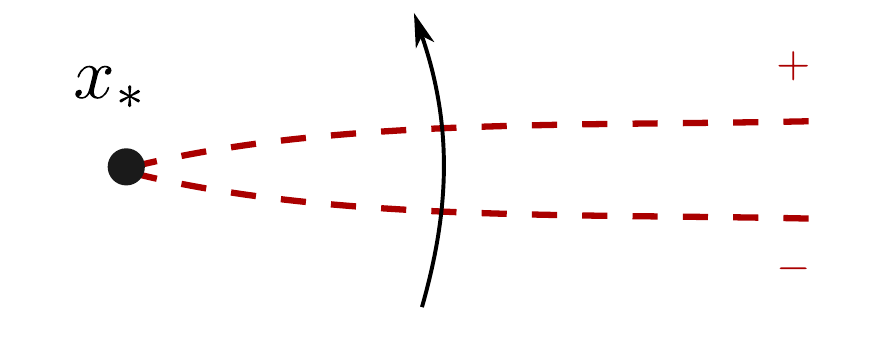}
\caption{Local behavior of logarithmic cuts near a logarithmic singularity on $\IC^*_x$. The singularity lifts to two logarithmic punctures on $\Sigma$ with $y_\pm$ running to infinity and to zero, respectively. The logarithmic cut lives on $\Sigma$ and connects the two punctures on different sheets $\pm$ (the cut therefore turns around a branch point elsewhere to switch between sheets).}
\label{fig:log-branching}
\end{center}
\end{figure}

In fact, since the monodromy preserves the sheet of the double cover $\Sigma\rightarrow\mathbb{C}^*$, this type of branching occurs on $\Sigma$ itself. It is natural to choose logarithmic branch cuts  on $\Sigma$ that connect the pole of $y_\pm$ and the zero of  $y_\mp$ at the same value of $x$. Logarithmic cuts are defined on $\Sigma$, and to run from one sheet to another they must circle around one of the square-root ramification points.
Upon crossing a logarithmic branch cut in the counter-clockwise direction (as seen on the $x$-plane) the solutions are exchanged by a shift in the logarithmic index
\be\label{eq:log-monodromy}
	\psi_{\pm,n} \to \psi_{\pm,n \pm k} \qquad \text{(counter-clockwise crossing)}\,.
\ee

\paragraph{Punctures of $\IC^*_x$.}
Finally, since we are working with $x\in \IC^*$ there are punctures at $x=0,\infty$, appearing due to the term $\frac{\dd x}{x}$ in the differential \eqref{eq:lambda-pm-n}. From the expression
\begin{equation}
    \lambda_{\pm,n}=\left(\log y_\pm+2\pi i n \right)\frac{\dd x}{x}
\end{equation}
we see that the origin and infinity are punctures of different types depending on the behaviour of $y$:
\begin{enumerate}
    \item If $\lim_{x\rightarrow 0}y_\pm=0,\,\infty$ (respectively $\lim_{x\rightarrow \infty}y_\pm=0,\,\infty$), then the origin (respectively $\infty$) is a logarithmic puncture, and the same considerations of the previous section apply.
    \item If $\lim_{x\rightarrow 0,\infty}=y_{0,\infty}\notin\{\pm1,\,0,\,\infty\}$, then we have a regular puncture, signaled by a simple pole with residue 
    \begin{equation}
        \res_{x=0}\lambda_{\pm,\,n}=\pm\log y_{0}+2\pi i n,\quad\res_{x=\infty}\lambda_{\pm,\,n}=\pm\log y_{\infty}-2\pi i n.
    \end{equation}
    \item If $\lim_{x\rightarrow0,\,\infty} y_+=\lim_{x\rightarrow0,\,\infty} y_- =a\in\{1,\,-1\}$, upon expanding $\log y$ we can distinguish two subcases:
    \begin{enumerate}
        \item If
        \be
\begin{split}
	\lambda_{\pm, n} &=
	(2\pi i \,n \mp    b x^{k}+ \dots) \frac{dx}{x}
	\qquad (a=1)
	\\
	\lambda_{\pm, n} &=
	(2\pi i \,(n+1/2) \pm  b x^{k}+ \dots) \frac{dx}{x}
	\qquad (a=-1)\\
 k&\in\mathbb{Z},
\end{split}
\ee
then we have either a regular or apparent singularity, depending on the values of $a,\, n$.
\item \label{item:square-root-puncture} If
        \be
\begin{split}
	\lambda_{\pm, n} &=
	(2\pi i \,n \mp    b x^{k}+ \dots) \frac{dx}{x}
	\qquad (a=1)
	\\
	\lambda_{\pm, n} &=
	(2\pi i \,(n+1/2) \pm  b x^{k}+ \dots) \frac{dx}{x}
	\qquad (a=-1)\\
 k&\in\mathbb{Z}+1/2,
\end{split}
\ee
Then we have a regular puncture that collided with a branch point. The monodromy around the origin has to take into account also for the change of sheets, inducing $\psi_\pm\mapsto \psi_\mp$.
\end{enumerate}
\end{enumerate}

\subsection{The $q$-Riccati equation}
The WKB ansatz \eqref{eq:WKBAns} is formulated in terms of the quantum 1-form $S(x;\hbar)\dd x$, however no explicit all-order WKB solution is known for $S$. In this section, we push further the ideas of \cite{DML2022}, and formulate WKB approximation in terms of another differential, for which we provide the all-order WKB solution.

\begin{definition}
    Consider the expression
\be
	\R(x,q):=\frac{\psi(qx)}{\psi(x)},
\ee
where $\psi$ satisfies \eqref{eq:2qDEcan}. We call the equation satisfied by $\R,$
\begin{equation} \label{eq:q-Riccati}
\R(x,q)\R(q^{-1}x,q)-2T(x,q)\R(q^{-1}x,q)+1=0
\end{equation}
the q-Riccati equation.
\end{definition}

This equation shares fundamental features of the Riccati differential equation associated to Schr\"odinger operators. In particular, it has  the advantage over \eqref{eq:2qDEcan} of being first-order in the $q$-shift, but this comes at the price of being nonlinear and inhomogeneous. Similarly to the differential case, one has the following
\begin{proposition}\label{thm:qR-sol}
Let $\R$ be given by the formal power series
\begin{equation}
\R(x,q)=\sum_{k=0}^{\infty}\R_k(x)(\log q)^k=\sum_{k=0}^{\infty}\R_k(x)\hbar^k\,.
\end{equation}
Them the q-Riccati equation \eqref{eq:q-Riccati} admits the following recursive formal solution:
\begin{equation}\label{eq:qR-sol}
\begin{split}
\R_{n,\pm}(x) & =\pm\frac{ 1}{2\sqrt{T_0^2-1}}\sum_{m=1}^{n-1}\sum_{l=0}^m\frac{1}{l!}\R_{m-l,\pm}\partial^l_{\log x}\left(\R_{n-m,\pm}-2T_{n-m}\right)\\
&\pm\frac{1}{2\sqrt{T_0^2-1}}\sum_{l=1}^n\R_{n-l,\pm}\left(\partial^l_{\log x}\left(\R_{0,\pm}-2T_0\right)-2T_l\right) ,
\end{split}
\end{equation}
with the leading order $n=0$ recovering the WKB curve,
\begin{equation}\label{eq:ClassCurve}
\R_0(x)+\frac{1}{\R_0(x)}=2 T_0(x),\qquad \R_{0,\pm}=y_\pm=T_0\pm\sqrt{T_0^2-1} .
\end{equation}
\end{proposition}
\begin{proof}
    The proof is obtained by plugging the formal series in the q-Riccati equation, setting $q=e^{\hbar}$, $\hbar\rightarrow 0$ (see Section 3.5 of \cite{DML2022})
\end{proof}
This recursion produces two solutions to the $q$-Riccati equation, specified by a choice of square-root branch for $\sqrt{T_0^2-1}$.  
The first two terms in this expansion are
\begin{align}\label{eq:Ricc0}
\R_{0,\pm}&=T_0\pm\sqrt{T_0^2-1}=e^{\hbar  S_{-1}},\\
\frac{\R_{1,\pm}}{\R_{0,\pm}}
&=\mp\frac{T_1}{\sqrt{T_0^2-1}} \mp \frac{1}{2} \partial_{\log x} 
\log\left(
\frac{T_0 + \sqrt{T_0^2-1}}{T_0 - \sqrt{T_0^2-1}}
\right)
\pm \frac{1}{4}\partial_{\log x}\log(T_0^2-1),
\\
& = \mp \frac{T_1}{\sqrt{T_0^2-1}} \mp \frac{1}{2} \partial_{\log x} 
\log\left(
\frac{y_+}{y_-}
\right)
\pm \frac{1}{2}\partial_{\log x}\log(y_+-y_-)\,.
\end{align}

\subsection{All order WKB asymptotic series}\label{Sec:AllOrdWKB}

The formal solution to the $q$-Riccati equation \eqref{eq:qR-sol} allows to compute the formal WKB solutions to the original $q$-difference equation \eqref{eq:2qDEcan} to all orders. This is given by the following
\begin{proposition}
Let
\begin{equation}
	S(x;\hbar)=\sum_{n=-1}^{\infty}S_n\hbar^n
\end{equation}	
be the formal power series defining the WKB solution \eqref{eq:WKBAns}. This is determined by the coefficients $D_n$ in \eqref{eq:LogRExp} through the recursion
\begin{equation}
	S_{n-1}=D_n+\sum_{k=1}^n\frac{1}{(k+1)!}\partial^k_{\log x}S_{n-k-1},
\end{equation}
where
\begin{equation}\label{eq:LogRExp}
	D_n=\frac{(-1)^{n-1}}{n}\det\left( \begin{array}{ccccc}
		\frac{\CR_1}{\CR_0} & 1 & 0 & \cdots & 0 \\
		2\frac{\CR_2}{\CR_0} & \frac{\CR_1}{\CR_0} & 1 & \cdots & 0 \\
		3\frac{\CR_3}{\CR_0} & \frac{\CR_2}{\CR_0} & \frac{\CR_1}{\CR_0} & \ddots & \vdots \\
		\vdots & \vdots & \vdots & \ddots &  1 \\
		n\frac{\CR_n}{\CR_0} & \frac{\CR_{n-1}}{\CR_0} & \cdots & \cdots & \frac{\CR_1}{\CR_0}
	\end{array} \right),
\end{equation}
are the terms in the formal power series for $\log \CR$, 
\begin{equation}
\log\CR(x,\hbar)=\log\CR_0+\sum_{n=1}^{\infty}D_n\hbar^n,
\end{equation}
and the $\CR_n$'s are determined by \eqref{eq:qR-sol}.
\end{proposition}
\begin{proof}
To prove this statement, it is necessary to spell out the relation between the formal power series $S(x,\hbar)$ and $\R(x,\hbar)$, which reads
\begin{equation}\label{eq:logRS}
\begin{split}
\log\R(x,\hbar) & =\int_x^{qx} S(x,\hbar)\frac{\dd x}{x}=\hbar S(x,\hbar)+\sum_{k=1}^{\infty}\frac{\hbar^{k+1}}{(k+1)!}\partial_{\log x}^{k} \left(S(x,\hbar) \right)\\
& =\hbar S(x,\hbar)+ \hbar x\partial_x\chi(x)\,,
\end{split} 
\end{equation}
where
\begin{equation}
\chi(x):=\sum_{k=1}^{\infty}\frac{\hbar^{k+1}}{(k+1)!}\partial_{\log x}^{k}S(x,\hbar).
\end{equation}

On the other hand, equation \eqref{eq:qR-sol} provides an all-order recursive formula for the $\CR_n$'s. We can use it to obtain the formal power series of the logarithm
\begin{equation}\label{eq:LogRDDef}
\log\CR=\log\left(\sum_n \CR_n\hbar^n\right)=\log\CR_0+\hbar\frac{\CR_1}{\CR_0}+\dots=:\log \CR_0+ \sum_{n=1}^{\infty} D_n\hbar^n .
\end{equation}
Expression \eqref{eq:LogRExp} for the coefficients $D_n$ of this formal power is a well-known result from the theory of symmetric functions \cite{Littlewood1970Book,Cadogan1971,Weisstein}.
\end{proof}
Since $S(x,\hbar)$ and $\frac{1}{\hbar}\log R(x,\hbar)$ differ only by a total derivative, and since the wavefunction only depends on the primitive of $S(x,\hbar)$, it follows that solutions to the $q$-difference equation \eqref{eq:2qDEcan} can be expressed directly in terms of the quantity $\log\R$:
\begin{theorem}\label{thm:WKBSol}
The formal WKB solution \eqref{eq:WKBAns} can be rewritten in the following form:
\begin{equation}\label{eq:psi-R}
	\psi(x)=e^{-\chi(x;\hbar)}\exp\left\{\frac{1}{\hbar}\int_{x_0}^x\log\R(x',\hbar)\frac{\dd x'}{x'} \right\}, 
\end{equation}
where $\chi(x,\hbar)$ is a single-valued function on the WKB curve $\Sigma$. 
\end{theorem}
\begin{proof}
Expression \eqref{eq:psi-R} is simply a consequence of \eqref{eq:logRS}. We have to show that $\chi$ is single-valued on $\Sigma$. From the solution \eqref{eq:qR-sol} to the Riccati equation we can see that the functions $\CR_n(x)$ involve at most square-root multivaluedeness and are single-valued functions on $\Sigma$. The formal $\hbar$-expansion of $\log\CR(x,\hbar)$ has then logarithmic branching at leading order and nowhere else, as can be seen from \eqref{eq:LogRDDef} and \eqref{eq:LogRExp}. Furthermore, since $\log\R_0=\hbar S_{-1}=\log y$ (see equations \eqref{eq:lambda-pm-n} and \eqref{eq:ClassCurve}), from equation \eqref{eq:logRS} it follows that neither the higher-order corrections to $S(x,\hbar)$, nor the function $x\partial_x\chi$, have multi-valuedeness on $\Sigma$. In fact, since $\chi$ is constructed from total derivatives of $S$, it is also itself single-valued.
\end{proof}
In the proof we used the fact that the logarithmic ambiguity is only present at leading order in $\hbar$. This fact has an important consequence:
\begin{corollary}\label{thm:LogCorollary}
The relation between WKB wave-functions with different $\mathbb{Z}$-labels is
\begin{equation}\label{eq:log-index-$q$-periodicity}
	\psi_{\pm, m}(x) = \psi_{\pm, n}(x) 
	\left(\frac{x}{x_0}\right)^{ \frac{2\pi i }{\hbar} (m-n) }\,,
\end{equation}
i.e. $\psi_{\pm,m}$ and $\psi_{\pm,n}$ with $n\ne m$ differ by $q$-periodic factors and are linearly dependent over $\mathcal{M}_q$.
\end{corollary}
Henceforth we will define the $q$-periodic function associated to the choice of basepoint $x_0$ as
\be\label{eq:xi-def}
	\xi  :=\left(\frac{x}{x_0}\right)^{ \frac{2\pi i }{\hbar}  }
\ee
Finally, from \eqref{eq:Ricc0} we can write explicitly the leading $\hbar$-behavior of the WKB wavefunctions:
\begin{equation}\label{eq:psi-leading-order}
\begin{split}
\psi_{\pm,n}(x,\hbar) & \mathop{\sim}_{\hbar\rightarrow0}\frac{1}{(T_0^2-1)^{\frac{1}{4}}}  
\left(T_0 \mp \sqrt{T_0^2-1}\right)
\exp\left\{\pm\frac{1}{\hbar}\int^x_{x_0}\left(\arccosh T_0(x')+2\pi i n\right)\frac{\dd x'}{x'}\right\}\\
&\times\exp\left\{\mp \int^x\frac{T_1(x')}{\sqrt{T_0(x')^2-1}}\frac{\dd x'}{x'} \right\}(1+O(\hbar))\,.
\end{split}
\end{equation}

\section{The WKB (linear) Stokes phenomenon  for second order $q$DEs}\label{sec:stokes-data}

Formal series solutions of $q$-difference equations can be resummed into meromorphic functions with poles by means of a $q$-analogue of Borel summation involving the $q$-Borel and $q$-Laplace transforms. 
In the confluent limit, i.e. taking  $q\to 1$, these notions are generally expected to reduce to the usual notion of Borel summation for series in $\hbar = \log q$. 
Relevant definitions and a discussion of the confluence limit can be found in \cite{2009arXiv0903.0853R, lastra2016class, RZ02, Tahara2015, 2015arXiv150102994D, AIF_2009__59_1_347_0} and references therein.

In this paper, instead of working at finite $q$, we have adopted the WKB limit where $\hbar=\log q\rightarrow0$, and have chosen to study formal solutions in WKB form \eqref{eq:WKBAns}, or equivalently \eqref{eq:psi-R}. These feature an exponential leading term, and a formal power series in $\hbar$ \footnote{
This limit can be motivated physically by recalling that $q$-difference equations arising from quantization of mirror curves have solutions that correspond to open topological string partition functions \cite{Aganagic:2003qj}.
From this perspective $\hbar$ is identified with the string coupling and taking $\hbar\to 0$ therefore corresponds to studying the genus expansion. 
On the other hand, working at finite $q$ is more natural from the viewpoint of M-theory \cite{Gopakumar:1998ii, Gopakumar:1998jq}. Although this is a very interesting perspective, we will not pursue it in this paper.
}:
\begin{equation}
\psi(x,\hbar) = e^{\frac{1}{\hbar} w(x)} \, \sum_{k\geq 0} f_k(x) \hbar^k .
\end{equation}

Starting from this section, we will make the folllowing reasonable assumption:
\begin{assumption}\label{assumption}
There exist piecewise analytic solutions to \eqref{eq:2qDEcan} having \eqref{eq:psi-R} as asymptotic expansion.
\end{assumption}
A possible way of producing analytic functions from the WKB asymptotic series is Borel summation $\CB$, defined as the composition of a Borel transform and a Laplace transform, which produces analytic functions of $\hbar$ defined over open patches in $\IC^*_x$%
\footnote{More precisely, this is the case if the series is Gevrey one, see e.g. \cite{costin2008asymptotics}.} 
\be\label{eq:borel-sum}
	\varphi_{\pm,n}(x) := \CB[\psi_{\pm,n}(x)].
\ee
We will not enter into the details of the Borel sum operation $\CB$, as they will not be used in what follows. It is well known that the Borel summability of asymptotic WKB series depends on $x$, because
the Borel transform of asymptotic series often features poles whose position depends on $x$.
When the poles move across the integration contour for the Laplace transform (which we take to be $\IR_{>0}$) the Borel sum fails to converge.
The collection of points on $\IC^*_x$ where this phenomenon happens is a collection of lines on $\IC^*_x$ called the Stokes graph, and its complement defines a a system of patches $\bigcup_\alpha R_\alpha$ on $\IC^*_x$. 
Borel summation defines locally analytic functions $\varphi_{\pm,n}^{R}$ within each patch, but its result is discontinuous from one patch to another. 
The discontinuity is described by an analogue of Voros' connection formulae~\cite{Voros}.

\subsection{The Stokes graph}\label{sec:Stokes-graph}

To locate the position of Stokes lines we take a shortcut, and instead of studying Borel summability of $\psi_{\pm,n}$ we adopt the general principle that Stokes phenomena between two solutions 
of a differential equation appear when one is maximally dominant over the other~\cite{dingle1973asymptotic,bender2013advanced}, i.e. where 
$	\left|\frac{\psi_{s_1,n_1}(x)}{\psi_{s_2,n_2}(x)}\right| $ 
is maximized. This is a standard definition of Stokes phenomenon in the realm of asymptotic analysis \cite{bender2013advanced,Olde1995}. In our case, its validity relies on our original assumption \ref{assumption}.
An analogous perspective has been recently applied for the study the \textit{nonlinear} Stokes phenomenon of q-Painlev\'e I equation \cite{joshi2019nonlinear}.

At small $\hbar$ the condition of relative dominance is determined by the exponential factor in $\psi$, and reads
\be\label{eq:Stokes-line-Im-zero}
	\Im\left[
	\hbar^{-1} \, \int^x (\log y_{s_2}(x')-\log y_{s_1}(x') + 2\pi i (n_2-n_1))\,\frac{\dd x'}{x'} 
	\right] = 0,\quad s_i=\pm,\quad n_i\in\mathbb{Z}.
\ee
This definition coincides with that of trajectories of an ``exponential network'' \cite{Klemm:1996bj, Gaiotto:2012rg, Eager:2016yxd, Banerjee:2018syt}. For certain first-order $q$-difference equations it was shown in \cite{Grassi:2022zuk, Alim:2022oll} that this definition indeed coincides with Stokes lines for Borel sums of asymptotic series solutions.

From \eqref{eq:Stokes-line-Im-zero}, we see that Stokes lines for exact WKB analysis of second order $q$-difference equations fall into four classes, labeled by 
$(\pm\mp,n) $ and $ (\pm\pm,n)$ with $n=n_1-n_2 \in \IZ$:
\begin{enumerate}
    \item Stokes lines of types $(+-,n)$ can be generated at square-root branch points, or at mutual intersections of other Stokes lines. On a Stokes line of type $(+-,n)$ the solutions $\psi_{-,k+n}$ are maximally subdominant with respect to $\psi_{+,k}$.
    \item Stokes lines of types $(-+,n)$ can be generated at square-root branch points, or at mutual intersections of other Stokes lines. On a Stokes line of type $(-+,n)$ the solutions $\psi_{+,k+n}$ are maximally subdominant with respect to $\psi_{-,k}$.
\end{enumerate}

As will be discussed below, near a branch point there are always three Stokes lines of types $(\pm\mp,0)$.

\begin{enumerate}\setcounter{enumi}{2}
    \item Lines of type $(++,n)$ can be generated at logarithmic branch points (if these lie at finite $x$) or at intersections of $(+-,m)$ and  $(-+,n-m)$ lines. On these Stokes lines, the solution $\psi_{\pm,k+n}$ is maximally subdominant with respect to $\psi_{\pm,k}$.
    \item Lines of type $(--,n)$ can be generated at logarithmic branch points (if these lie at finite $x$) or at intersections of $(+-,m)$ and  $(-+,n-m)$ lines. On these Stokes lines, the solution $\psi_{\pm,k+n}$ is maximally subdominant with respect to $\psi_{\pm,k}$.
\end{enumerate}

As will be discussed below, near a logarithmic branch point there are always two families of Stokes lines of types $(\pm\pm,n)$. For details see \cite{Eager:2016yxd, Banerjee:2018syt, Grassi:2022zuk, Alim:2022oll}.

Any two Stokes lines can intersect each other as long as their types are not the same or `opposite' (i.e. $(ij,n)$ and $(ji,-n)$).
There are three possible types of intersections \cite{Banerjee:2018syt}:
\begin{itemize}
\item When lines of type $(+-,n)$ and $(-+,m)$ intersect they generate infinitely many new trajectories. These are organized in three different types $(+-,n+k(m+n))$, $(-+,m+k(m+n))$ and $(\pm\pm,(k+1)(m+n))$ for $k\in \IZ_{\geq 0}$.
\item When lines of type $(+-,n)$ and $(+-,m)$ intersect no new lines are generated.
\item When lines of type $(+-,n)$ and $(\pm\pm,m)$ intersect they generate one new Stokes line of type $(+-,k(n+m))$ with $k\geq 1$ a positive integer. 
\end{itemize}

We will now study in detail the Stokes phenomenon of $q$DEs of the form \eqref{eq:2qDEcan} near the different types of critical points classified in Section \ref{Sec:CritPtsCurve}. 
The local descriptions thus obtained are the ``building blocks'' to be patched together for a global description of the Stokes phenomenon and monodromy of the WKB solutions $\psi_{\pm,n}$.

\subsection{Stokes behaviour near branch points}\label{Sec:AiryStokes}
As discussed earlier, the critical points of the $q$DE \eqref{eq:2qDEcan} are branch points on $\IC^*_x$, defined by the equation $T_0(x_\pm) = \pm1$. We will discuss solutions normalised at $x_\pm$ respectively. The solution behaves similarly near square-root branch points $x_\pm$, but there are some differences. We discuss each in turn.
Near a branch point where $T_0(x) = 1$ the leading order behaviour of the WKB wavefunction is given by 
\eqref{eq:psi-leading-order} where $T_0$ is given in \eqref{eq:T0-branch-point}
\be
\begin{split}
	\psi_{s,n}(x,\hbar) 
	& \mathop{\sim}_{\stackrel{\hbar\rightarrow0}{x\to x_+}} 
	\frac{1}{(x-x_+)^{1/4}}
	\exp\left\{ \frac{1}{\hbar}\int_{x_+}^x\left( s\, c_0 (x-x_+)^{1/2} +2\pi i n\right)\frac{\dd x'}{x'}\right\}
	(1+O(\hbar))
	\\
	& \sim 
	\frac{1}{(x-x_+)^{1/4}}
	\left( \frac{x}{x_+} \right)^{\frac{2\pi i n}{\hbar}}
	\exp\left\{\frac{2}{3}\frac{s}{\hbar}  c_0 (x-x_+)^{3/2}\right\}
	(1+O(\hbar))\,,
	\\
\end{split}
\ee
where $x_+$ denotes both the position of the branch point and the choice of normalization basepoint for the WKB exponent integral.
The factor $( {x}/{x_+})^{\frac{2\pi i n}{\hbar}}$ was discussed in \eqref{eq:log-index-$q$-periodicity}, where it was observed that it is $q$-periodic and can therefore be regarded as an overall normalization `constant' when working with the vector space of solutions $\CH$ over the field $\mathcal{M}_q$.
The remaining factors coincide with the WKB asymptotics of solutions to the Airy equation, whose Stokes behaviour is well known.
There are three Stokes lines emerging from the branch point that divide its neighbourhood into three sectors. 
It is clear from the similarity with Stokes phenomena of Airy's equation that $\psi_{+,n}$ and $\psi_{-,m}$ can only mix if $n=m$, since the dependence on the logarithmic index appears as an overall $q$-constant factor which does not affect Borel summation of asymptotic series in $\hbar$. We now describe the Stokes behavior in more detail.

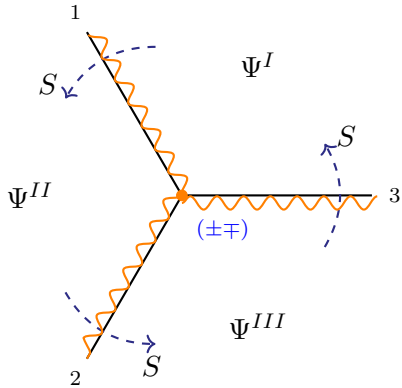
\begin{figure}[h!]
\begin{center}
\begin{tikzpicture}
\draw[thick](0,0)--+(0:2.5) (0,0)--+(120:2.5) (0,0)--+(-120:2.5);
\draw[thick,MyBlue!70!black,dashed,->](-20:2) to[bend right] (20:2);
\draw[thick,MyBlue!70!black,dashed,->] (100:2) to[bend right] (140:2);
\draw[thick,MyBlue!70!black,dashed,->] (-140:2) to[bend right] (-100:2);
\draw[thick,orange,wave](0,-0.1)--+(0:2.5) (0,0)--+(120:2.5) (0,0)--+(-120:2.5);
\draw[fill,orange](0,0)circle[radius=0.7mm];
\node at(20:2.3){$S$};
\node at(140:2.3){$S$};
\node at(260:2.3){$S$};
\node at(60:2){$\Psi^{I}$};
\node at(180:2){$\Psi^{II}$};
\node at(300:2){$\Psi^{III}$};
\node at(0:2.8){\scriptsize$3$};
\node at(120:2.8){\scriptsize$1$};
\node at(240:2.8){\scriptsize$2$};
\node at(320:0.7){{\scriptsize\color{blue}$(\pm\mp)$}};
\end{tikzpicture}
\caption[]{Stokes phenomenon of the solution near a branch point.  \label{Fig:AiryStokes}}
\end{center}
\end{figure}

If $\arg\hbar$ is generic, Stokes lines evolve away from the branch point and towards punctures. This allows to define a canonical basis of solutions within each sector, as follows.
Labeling the Stokes lines by $1,2,3$, we denote by $s_1,\,s_2,\,s_3$ the unique solution that decays exponentially fast along the respective line.
More precisely, the solution is unique up to an overall multiplication by a $q$-periodic factor $\xi^n$ and we shall work with the choice $n=0$ of solutions that have zero logarithmic index near the branch point.
For later convenience we also introduce the following notation 
\begin{equation}
	\fs_i(x):=\left(\begin{array}{c}
		s_i(x) \\
		s_i(qx)
	\end{array} \right).
\end{equation}
We label Stokes sectors by roman numerals as shown in Figure \eqref{Fig:AiryStokes}, and in each sector we define a basis of solutions given by the two vanishing sections along the delimiting Stokes lines.
We collect the basis in a matrix of solutions $\Psi^\star$ associated to each sector, with $\star=I,\,II,\,III$. 
The first column of $\Psi^\star$ is proportional to the $\fs _i$ that
decays along the Stokes ray that delimits the Stokes region on the left, while the second column is proportional to the $\fs _i$ that decays along the other Stokes ray. 
The three bases of solutions are related by a universal matrix $S$ due to the cyclic symmetry
\be\label{eq:Stokes-jump-branch-point}
	\Psi^{\star+I}=\Psi^{\star}S.
\ee

Following a standard argument from exact WKB analysis,
requiring that the monodromy around the branch point is trivial, i.e. $S^3=\mathbb{I}_2$, determines
\begin{equation}\label{eq:S-matrix}
	S=\left(\begin{array}{cc}
		-1 & i \\
		i & 0
	\end{array} \right)=
	\left( \begin{array}{cc}
		0 & i \\
		i & 0
	\end{array} \right)  
	\left(\begin{array}{cc}
		1 & 0 \\
		i & 1
	\end{array} \right),
\end{equation}
The three solutions therefore have the following mutual normalizations
\begin{equation}\label{eq:AiryPsis}
	\Psi^I=\left(\fs_1\vert i\fs_3 \right),\quad \Psi^{II}=\left(\fs_2\vert i\fs_1 \right),\quad \Psi^{III}=\left(\fs_3\vert i\fs_2 \right),
\end{equation}
and the three vanishing solutions must obey the linear relation
\begin{equation}
	s_1+s_2+s_3=0\,.
\end{equation}

The matrix $S$ acts on the space of solutions $\CH$ as a vector space over $\mathbb{F}_{q}$. When working over $\mathbb{C}$ the basis of solutions becomes infinite dimensional $\Psi_{m}$ for $m\in \IZ$, and the Stokes matrix becomes 
$S_{m,n} = S\otimes \mathbb{I}_{mn}$.

\subsubsection{Fixing signs: two types of branch point}
In the above discussion, the choice of basis in each sector is defined in relation to the behavior of WKB solutions along each of the Stokes lines.
For many purposes, it will be important to know how a basis of solutions near a branch point is related to the basis of solutions near another one. 
Clearly, columns of the matrices \eqref{eq:AiryPsis} correspond to the WKB basis, i.e. to solutions with an asymptotic expansion (in $\hbar$) that agrees with the WKB solutions in each sector.
The only remaining ambiguity is whether it is the left (resp. right) column that corresponds to $\psi_{+,n}$ or $\psi_{-,n}$.

To assign signs globally it is necessary to fix a choice of branch cuts for the square-root branch point, since $\psi_{\pm,n}$ are only well-defined after such a global choice of trivialization.
A canonical choice is to consider three square-root branch branch cuts emanating from the branch point, and running along each of the Stokes lines slightly displaced counterclockwise,\footnote{Away from the branch point Stokes lines may intersect. When this happens we join the branch cuts at the intersection. Since there are only two sheets this is always possible. Thus branch cuts always run between branch points, or between a branch point and a puncture. Only Stokes lines generated at branch points (a.k.a. `primary' lines) carry branch cuts next to them, and only up the their first intersection with another primary Stokes line.} see Figure~\ref{Fig:AiryStokes}.

A Stokes line will be labeled by $(ij,\Delta n)$ if the solution $\psi_{j,n+\Delta n}$ is maximally subdominant over $\psi_{i,n}$. 
Stokes lines that emanate from a branch point $x_+$ always have $\Delta n=0$.
Moreover, with our choice of trivialization all three Stokes lines are actually of the same type: they are either all $(+-,0)$ or all $(-+,0)$ for a given branch point.
Therefore there are two types of branch points, that we label by $(\pm\mp)$ as shown in Figure \ref{Fig:AiryStokes}.
The signs correspond to those of the WKB solutions in each column of the matrices \eqref{eq:AiryPsis} 
\be\label{eq:BasisChoiceBPT}
\begin{split}
	\text{type }(+-):\qquad 
	\Psi^{\star}(x) & \sim
	\left(\begin{array}{cc}
		\psi_{+,n}(x) & \psi_{-,n}(x) \\
		\psi_{+,n}(qx) & \psi_{-,n}(qx)
	\end{array} \right)
	\\
	\text{type }(-+):\qquad 
	\Psi^{\star}(x) & \sim
	\left(\begin{array}{cc}
		\psi_{-,n}(x) & \psi_{+,n}(x) \\
		\psi_{-,n}(qx) & \psi_{+,n}(qx)
	\end{array} \right) 
\end{split}
\ee
Here $\sim$ denotes the fact that each matrix corresponds to the (WKB) asymptotic series expansion of the actual solutions $\Psi^\star$, and $\star=I,II,III$ denotes the Stokes sector.
Given the chosen position of branch cuts, the right column of the matrix is the (asymptotic expansion) of the solution that decays fastest along the clockwise Stokes line. Therefore for a branch point of type $(+-)$ all three Stokes lines must be of type $(+-,0)$ so that $\psi_{-,n}$ is maximally dominated by $\psi_{+,n}$. Viceversa for a branch point of type $(-+)$ all three Stokes lines must be of type $(-+,0)$.
Note that we have included the possibility of having solutions with arbitrary logarithmic index $n$, since this corresponds to a change of basis $\Psi\to \Psi \cdot \xi^n \mathbb{I}_{2\times 2}$, which leaves the matrix $S$ invariant according to \eqref{eq:Stokes-jump-branch-point}.

\subsubsection{The case $T_0(x)=-1$}
Next we discuss the behavior of Borel summation of WKB solutions near a branch point $x_-$ defined by $T_0(x_-) = -1$.
In this case the leading order behaviour of the WKB wavefunction is 
\be
\begin{split}
	\psi_{s,n}(x,\hbar) 
	& \mathop{\sim}_{\stackrel{\hbar\rightarrow0}{x\to x_-}} 
	\frac{1}{(x-x_-)^{1/4}}
	\exp\left\{ \frac{1}{\hbar}\int_{x_-}^x\left( -s \, c_0 (x-x_-)^{1/2} +2\pi i \left(n+\frac{1}{2}\right)\right)\frac{\dd x'}{x'}\right\}
	(1+O(\hbar))
	\\
	& \sim
	\frac{1}{(x-x_-)^{1/4}}
	\left( \frac{x}{x_-} \right)^{\frac{2\pi i (n+1/2) }{\hbar}}
	\exp\left\{-\frac{s}{\hbar} \frac{2}{3} c_0 (x-x_-)^{3/2}\right\}
	(1+O(\hbar)).
	\\
\end{split}
\ee
As before, the local behavior of the WKB exponential factor reproduces that of WKB solutions to the Airy equation, and there are three Stokes lines emanating from the branch point. There are two differences with the case of $x_+$: the signature of the exponent is opposite ($-s$ instead of $s$) and the factor $\left( \frac{x}{x_-} \right)^{\frac{2\pi i (n+1/2) }{\hbar}}$ is not $q$-periodic, but picks up a sign under $x\to qx$. Aside from these differences, the local behavior of solutions is identical to the previous case. 
Namely there are three Stokes sectors and in each of them we can define a canonical choice of basis for $\CH$ defined by the solutions that decay along the limiting Stokes lines. The jumps of this basis are once again described by $S$.

Recall that all Stokes lines from a given branch point are of the same type, but Stokes lines from two distinct branch points may have opposite signature. Nevertheless, with the choice of basis we have just introduced, the matrices describing Stokes jumps is the same in both cases. What changes is the identification of $\psi_{s,n}$ with the basis elements associated  (within each Stokes sector) to each of the Stokes lines.

\subsection{Logarithmic cuts}\label{sec:log-cuts}

Logarithmic branch cuts are defined on $\Sigma$, and connect points $T_0^{-1}(\infty)\in\mathbb{C}^*$, which are zeroes and infinities of $y$. 
As shown in Corollary \ref{thm:LogCorollary}, logarithmic branching relates linearly independent solutions only when working over $\IC$, while it produces ``rescaling'' of solutions by $q$-periodic functions when working over $\mathcal{M}_q$.

Due to the symmetry $y_-= y^{-1}_+$ of the involutive form, every logarithmic branch point on the base $\mathbb{C}^*$ has two pre-images on $\Sigma$. We can take the cut to start from the point $T_0^{-1}(\infty)$ on the $+$, sheet and to go back to the same point of the base after having turned around one of the square-root branch points $T_0^{-1}(\pm1)$, as in Figure \ref{Fig:LogCut}.
We choose to arrange logarithmic cuts in this way, so that their projection to $\IC^*_x$ resembles a `double-line' segment running from a logarithmic branch point to a square-root branch point. 

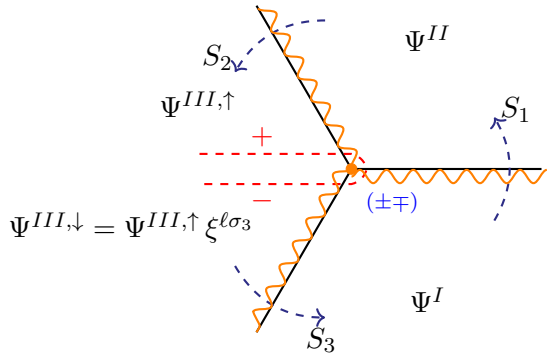
\begin{figure}[h]
\begin{center}
    \begin{tikzpicture}
\draw[thick](0,0)--+(0:2.5) (0,0)--+(120:2.5) (0,0)--+(-120:2.5);
\draw[thick,MyBlue!70!black,dashed,->](-20:2) to[bend right] (20:2);
\draw[thick,MyBlue!70!black,dashed,->] (100:2) to[bend right] (140:2);
\draw[thick,MyBlue!70!black,dashed,->] (-140:2) to[bend right] (-100:2);
\draw[thick,dashed,red] (-2,0.2) to (0,0.2) to[bend left] (0.2,0) to[bend left] (0,-0.2) to (-2,-0.2);
\draw[thick,orange,wave](0,-0.1)--+(0:2.5) (0,0)--+(120:2.5) (0,0)--+(-120:2.5);
\draw[fill,orange](0,0)circle[radius=0.7mm];
\node at(20:2.3){$S_1$};
\node at(140:2.3){$S_2$};
\node at(260:2.3){$S_3$};
\node at(60:2){$\Psi^{II}$};
\node at(157:2.2){$\Psi^{III,\uparrow}$};
\node at(160:1.25){{\color{red}$+$}};
\node at(200:1.25){{\color{red}$-$}};
\node at(195:3.0){$\Psi^{III,\downarrow} = \Psi^{III,\uparrow}\, \xi^{\ell \sigma_3}$};
\node at(300:2){$\Psi^{I}$};
\node at(320:0.7){{\scriptsize\color{blue}$(\pm\mp)$}};
\end{tikzpicture}
\caption[]{Stokes phenomenon of the solution near a branch point encircled by a log-cut.  \label{Fig:LogCut}}
\end{center}
\end{figure}

The two halves of the logarithmic cut run in this way over the two distinct sheets of $\Sigma$. It is always possible to fix the choice of logarithmic cuts in such a way that it runs on sheets $\pm$ as shown in Figure \ref{Fig:LogCut}, and we will fix this convention from now on.
Moreover, when crossing the projection of the logarithmic cut on $\IC^*_x$, we are crossing the cut in one direction on one sheet, but in the opposite direction on the other sheet (due to the opposite co-orientation of the two strands of the red-line in Figure \ref{Fig:LogCut}).
Recall from \eqref{eq:log-monodromy} that WKB solutions pick up factors 
\be\label{eq:log-cut-ccw-puncture}
	\psi_{\pm,n} \to \psi_{\pm,n} \xi^{\pm k} \,,
\ee
when going counter-clockwise around logarithmic singularity, where we used \eqref{eq:log-index-$q$-periodicity} to express the shift in the logarithmic index.
Turning counter-clockwise around the branch point corresponds to turning clockwise around the puncture, which gives therefore opposite signs in the exponent of $\xi$. 
The shifts of the basis of solutions pick up upon crossing the (projection to $\IC^*_x$ of the) logarithmic cut counter-clockwise around the branch point at $x_0$ depend on the sign of the basis, and we have two distinct cases labeled $(\pm\mp)$
\begin{align}
	\text{type }(+-):\qquad \label{eq:log-monodromy-cases-1}
	\Psi^{III,\uparrow}(x) & =
	\left(\begin{array}{cc}
		\phi_{+,n}(x) & \phi_{-,m}(x) \\
		\phi_{+,n}(qx) & \phi_{-,m}(qx)
	\end{array} \right)
	\mapsto 
	\Psi^{III,\downarrow}=  
	\Psi^{III,\uparrow}  \xi^{-k\sigma_3}
	\\
	\text{type }(-+):\qquad \label{eq:log-monodromy-cases-2}
	\Psi^{III,\uparrow}(x) & =
	\left(\begin{array}{cc}
		\phi_{-,m}(x) & \phi_{+,n}(x) \\
		\phi_{-,m}(qx) & \phi_{+,n}(qx)
	\end{array} \right) 
	\mapsto 
	\Psi^{III,\downarrow}(x) = \Psi^{III,\uparrow}  \xi^{k\sigma_3}
\end{align}
The signature $(\pm\mp)$ refers to the basis vectors in any region around the branch point. 
It will be important to distinguish between the two cases, and we keep track of this by labeling the branch point by $(\pm\mp)$ as shown in Figure~\ref{Fig:LogCut-types}.
It is important to recall that this result is based on the convention that $y_+$ is the sheet that diverges at the logarithmic puncture while $y_-$ is the one that goes to zero, and that $k<0$ is the degree of divergence, see \eqref{eq:log-bp-T0}.

The Stokes matrices also change type, because the Stokes lines cross the logarithmic cut, explicitly breaking the $\IZ_3$ symmetry of the previous case, see Figure \ref{Fig:LogCut}.
The condition of trivial monodromy is replaced by 
\begin{equation}\label{eq:log-sqrt-monodromy-constraint}
	S_3 S_1S_2 =\xi^{-\ell \sigma_3},\qquad \xi=\left(
	\frac{x}{x_0}
	\right)^{\frac{2\pi i}{\hbar}},\qquad \ell\in\mathbb{Z},
\end{equation}
where $\ell=k$ for \eqref{eq:log-monodromy-cases-1} and $\ell=-k$ for \eqref{eq:log-monodromy-cases-2}, and
\be
	S_j
	=
	\left(\begin{array}{cc}
		i\,\alpha_j & i \\
		i & 0
	\end{array} \right)
	=
	\left(\begin{array}{cc}
		0 & i \\
		i & 0
	\end{array} \right)
	\left(\begin{array}{cc}
		1 &  0 \\
		\alpha_j & 1
	\end{array} \right)	.
\ee

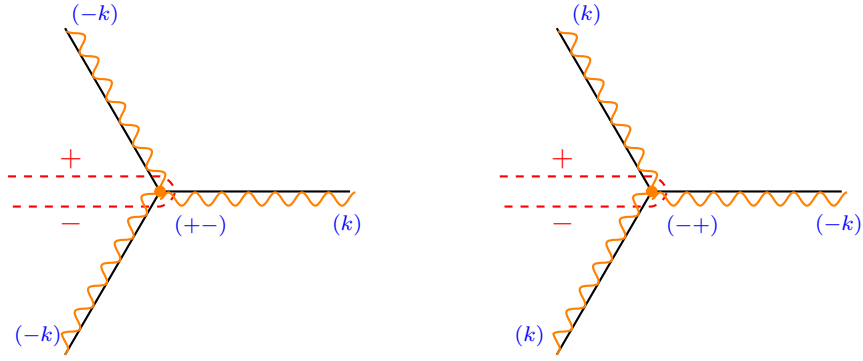
\begin{figure}[h!]
\begin{center}
\begin{tikzpicture}
\draw[thick](0,0)--+(0:2.5) (0,0)--+(120:2.5) (0,0)--+(-120:2.5);
\draw[thick,dashed,red] (-2,0.2) to (0,0.2) to[bend left] (0.2,0) to[bend left] (0,-0.2) to (-2,-0.2);
\node at(160:1.25){{\color{red}$+$}};
\node at(200:1.25){{\color{red}$-$}};
\draw[thick,orange,wave](0,-0.1)--+(0:2.5) (0,0)--+(120:2.5) (0,0)--+(-120:2.5);
\draw[fill,orange](0,0)circle[radius=0.7mm];
\node at(110:2.5){{\scriptsize\color{blue}$(-k)$}};
\node at(230:2.5){{\scriptsize\color{blue}$(-k)$}};
\node at(350:2.5){{\scriptsize\color{blue}$(k)$}};
\node at(320:0.7){{\scriptsize\color{blue}$(+-)$}};
\end{tikzpicture}
\qquad\qquad
\begin{tikzpicture}
\draw[thick](0,0)--+(0:2.5) (0,0)--+(120:2.5) (0,0)--+(-120:2.5);
\draw[thick,dashed,red] (-2,0.2) to (0,0.2) to[bend left] (0.2,0) to[bend left] (0,-0.2) to (-2,-0.2);
\node at(160:1.25){{\color{red}$+$}};
\node at(200:1.25){{\color{red}$-$}};
\draw[thick,orange,wave](0,-0.1)--+(0:2.5) (0,0)--+(120:2.5) (0,0)--+(-120:2.5);
\draw[fill,orange](0,0)circle[radius=0.7mm];
\node at(110:2.5){{\scriptsize\color{blue}$(k)$}};
\node at(230:2.5){{\scriptsize\color{blue}$(k)$}};
\node at(350:2.5){{\scriptsize\color{blue}$(-k)$}};
\node at(320:0.7){{\scriptsize\color{blue}$(-+)$}};
\end{tikzpicture}
\caption{Logarithmic labels for the Stokes data associated to each of the lines emanating from the two types of branch points encircled by logarithmic cuts.}
\label{Fig:LogCut-types}
\end{center}
\end{figure}

The Stokes constants are again determined uniquely by the equation, so that
\be\label{eq:Stokes-matrices-log}
	S_2 = S_3 = S^{(-\ell)}\,,
	\qquad
	S_1 = S^{(\ell)}
	\,,
\ee
where $S^{(\ell)}$ is  defined by
\be\label{eq:StokesLog}
	S^{(\ell)} = \left(\begin{array}{cc}
		-\xi^{\ell} & i \\
		i & 0
	\end{array} \right) \,.
\ee
This is a generalization of $S$ from \eqref{eq:S-matrix}, which corresponds to the special case $\ell=0$.
These matrices describe a change of basis for $\CH$, the space of solutions over $\mathcal{M}_q$
\be
	\Psi^{I} = \Psi^{III,\downarrow}\, S_3\,,
	\qquad
	\Psi^{II} = \Psi^{I}\, S_1\,,
	\qquad
	\Psi^{III,\uparrow} = \Psi^{II}\, S_2\, , \qquad \Psi^{III,\downarrow}=\Psi^{III,\uparrow}\xi^{\ell\sigma_3}
\ee 
For convenience we denote the logarithmic shifts included in the Stokes data of each Stokes line by a label $(\Delta n)$ as in Figure \ref{Fig:LogCut-types}. The two cases shown correspond to having a branch point of type $(+-)$ or $(-+)$ and lead to the two cases $\ell=k$ and $\ell=-k$ respectively.
The case $k=0$ corresponds to the absence of a logarithmic cut and recovers the branch point Stokes data from Figure \ref{Fig:AiryStokes}.
The Stokes matrices feature the coefficient $\xi$ which depends both on $x$, the point where the WKB solution is evaluated, and on $x_0$, branch point at which the solution is normalised.

To obtain the representation of the Stokes matrix \eqref{eq:StokesLog} on the space of solutions over $\IC$ it is sufficient to observe that the factor $\xi$ shifts the logarithmic index $\psi_{s,n} \to \psi_{s,n+1}$. The resulting infinite matrices are
\begin{align}
	(\mathsf{S}_2)_{nm}=(\mathsf{S}_3)_{nm}
	=i\sigma_1 \delta_{n,m}-\frac{1}{2} (1+\sigma_3) \delta_{m,n-\ell}, && (\mathsf{S}_2)_{nm}=i\sigma_1\delta_{n,m}-\frac{1}{2} (1+\sigma_3) \delta_{m,n+\ell},
\end{align}
where $\sigma_1,\,\sigma_3$ are Pauli matrices.

\subsection{WKB solutions near logarithmic branch points}

Near a logarithmic singularity, i.e. a point $x\notin \{0,\infty\}$ where $T_0(x)$ diverges, the leading order behaviour of the WKB wavefunction is given by 
\eqref{eq:psi-leading-order} where $T_0$ is given in \eqref{eq:log-bp-T0}
\be
\begin{split}
	\psi_{s,n}(x,\hbar) 
	& \mathop{\sim}_{\stackrel{\hbar\rightarrow0}{x\to x_*}} 
	(x-x_*)^{-k(\frac{1}{2}+s)}
	\exp\left\{\frac{1}{\hbar}\int_{x_0}^x\left( s\, k \log  (x-x_*)  +2\pi i n\right)\frac{\dd x'}{x'}\right\}
	(1+O(\hbar))
	\\
\end{split}
\ee
Here $x_0$ encodes the choice of normalization for the WKB integral, $x_*$ is the position of the logarithmic branch point, and $k<0$ is a negative integer.
Around $x_*$ the solutions have monodromy \eqref{eq:log-monodromy}, which exchanges $\psi_{s,n}$ with $\psi_{s, n+sk}$.
To determine the presence of Stokes lines, we again resort to the principle that such lines should  mark the locus of maximal dominance of one solution over another.

There are potentially four families of Stokes lines to consider, with labels $(\pm\mp,n)$ and $(\pm\pm,n)$. Since logarithmic branching preserves the index $\pm$, these Stokes lines will only relate solutions with the same sign and different logarithmic index. As we will now see, this means that the Stokes matrices  must be diagonal.

Lines of type $(++,n)$ (the opposite signature $(--,n)$ defines the same line) are simply spirals on $\IC^*_x$, and one may take Stokes lines at $x_*$ to be the two halves of the spiral going respectively towards $x=0$ and towards $x=\infty$
\be
	x(t) = x_* \, \exp \left( \frac{t \hbar}{2\pi i \, n}  \right) \qquad (t\geq 0).
\ee
Their shape is generically that of a spiral, except when $\hbar\in \mathbb{R}$ these degenerate into circles, and when $\hbar\in i\mathbb{R}$ when they become straight lines. The degeneration at $\hbar\in \mathbb{R}$ is expected on the grounds that this is the phase  where the Stokes graph (exponential network) always features a degeneration due to the appearance of a special type of saddles called D0 brane saddles in the physics literature \cite{Banerjee:2018syt}.

\begin{figure}[h!]
	\begin{center}
		\includegraphics[width=0.55\textwidth]{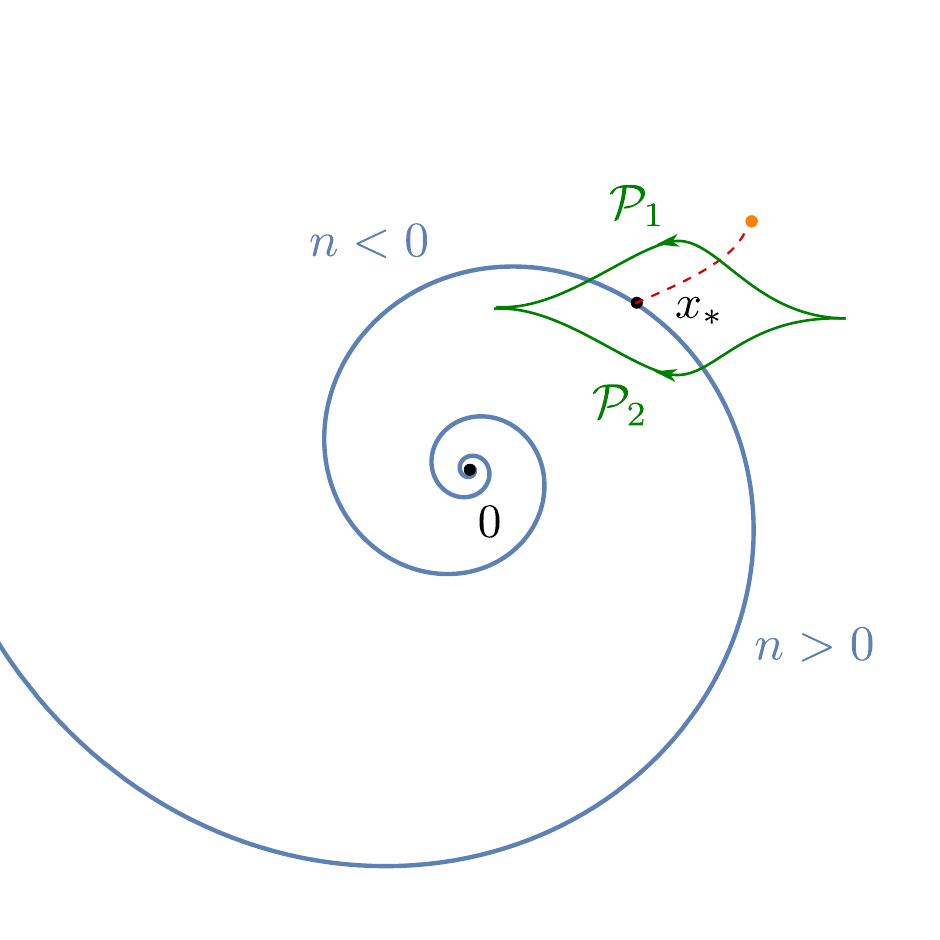}
		\caption{
		Two families of Stokes lines emanating from a logarithmic branch point. The logarithmic cut, whose projection to $\IC^*_x$ is indicated by a dashed line, turns around a square-root branch point (indicated in orange).}
		\label{fig:log-bp-stokes}
	\end{center}
\end{figure}

The monodromy of solutions across logarithmic cuts, going counterclockwise around $x_*$, is given by a power of $\xi$ as in \eqref{eq:log-cut-ccw-puncture}.
The monodromy factor is $q$-periodic, therefore it amounts to a simple rescaling of the two basis solutions when working over $\mathcal{M}_q$, but to an actual change of basis on the (infinite dimensional) space of solutions.
In matrix notation, the solutions change as follows\footnote{Comparing with \eqref{eq:log-monodromy-cases-1} and \eqref{eq:log-monodromy-cases-2}, the signs appear opposite because those expression describe counterclockwise monodromy around the branch point, not around $x_*$.}
\be
	\Psi  \to \Psi\, \xi^{\ell\sigma_3}\,,
\ee
where $\ell=\pm k$ depending on whether $\Psi=\left(\psi_{\pm},\psi_\mp\right)$, and $k<0$ is the degree of divergence for $T_0$ at $x_*$.

To determine the Stokes constants associated with Stokes lines emanating from $x_*$ we proceed exactly in the same way as for square-root branch points, i.e. by requiring that the overall monodromy around $x_*$ should be trivial. 

This leads to a new kind of Stokes line, which is radically different from the ones usually encountered in exact WKB analysis of differential equations. The corresponding Stokes matrices should not be expected to take the usual upper-triangular form of Voros' connection formula. 
Stokes lines of type $(\pm\pm,n)$ were included in the nonabelianization construction for exponential networks of \cite{Banerjee:2018syt}, where their Stokes data was determined for the case when such lines originated from intersections of other lines, but not directly from punctures.
In \cite{Grassi:2022zuk, Alim:2022oll} Stokes lines from logarithmic branch points were shown to play an important role in the description of resurgent structures of the quantum dilogarithm. 
In \cite{gupta:2024} the Stokes data for $(\pm\pm,n)$ Stokes lines originating from logarithmic branch points is determined via an extension of nonabelianization. What we present here is a translation of the results of \cite{gupta:2024} in the language of $q$WKB analysis.

The Stokes matrices $L_\pm$ associated to the lines with $n>0 $ and $n<0$ respectively are
\be\label{eq:log-matrix-ansatz}
	L_\pm  = \left(\begin{array}{cc}
		e^{\sum_{j\geq 1} \mu_j^{(\pm)} \xi^{\pm jk}} & 0 \\
		0 & e^{-\sum_{j\geq 1} \mu_j^{(\pm)} \xi^{\pm jk}}
	\end{array}\right),\qquad \mu_j^{(\pm)}\in\mathbb{Q}
\ee
where $\mu_j$ are the Stokes coefficients to be determined.
This kind of ansatz can be motivated by studying the algebra of Stokes matrices of types $(+-,n)$ and $(-+,m)$ as done in \cite{Banerjee:2018syt}.
Alternatively, one may simply observe that \eqref{eq:log-matrix-ansatz} is the most general ansatz compatible with the requirements that
$L_\pm$ are diagonal, that $\det L_\pm=1$ and that $L_\pm$ is a power series in $\xi^{\pm k}$.
Therefore an equivalent expression is simply as powers series $e^{\sum_{j\geq 1} \mu_j^{(\pm)} \xi^{\pm jk}} = 1 + \sum_{j\geq 1} \alpha_j^{(\pm)} \xi^{\pm jk}$. 

The flatness constraint can be read off from Figure \ref{fig:log-bp-stokes}, by demanding that transport along ${\mathcal{P}_1},{\mathcal{P}_2}$ agree
\be
	L_-
	\cdot
	\left(\begin{array}{cc}
		\xi^k & 0 \\
		0 & \xi^{-k} 
	\end{array}\right)
	=
	(L_+)^{-1}
\ee
This condition is solved by $\alpha^{(-)}_{j} = \delta_{j,1}$ and by $\alpha^{(+)}_{j} = (-1)^j$, so that
\be
	L_-  = \left(\begin{array}{cc}
		1+\xi^{-k} & 0 \\
		0 & \frac{1}{1+\xi^{-k}}
	\end{array}\right)\,,
	\qquad
	L_+  = \left(\begin{array}{cc}
		\frac{1}{1+\xi^k} & 0 \\
		0 & 1+\xi^k
	\end{array}\right)\,.
\ee

\section{Voros Symbols and WKB coordinates for $q$-difference equations}

In this section, we generalise the notion of shear/Fock-Goncharov coordinates for monodromies differential equations, as developed in \cite{Gaiotto:2009hg}, to the case of $q$-difference equations, and obtain for some of them expressions as ratios of q-Wronskian determinants. At the end of the section, it is explained how the shear coordinates reproduce quantum periods, given by periods of WKB differentials on $\Sigma$.

\subsection{Quantum periods}

Recalling from \eqref{eq:lambda-pm-n} that $\lambda_{\pm,n}$ is multivalued on $\IC_x^*\times \IC_y^*$, 
we consider the logarithmic covering $\IC^*_x\times\IC_{\log y}$. 
This defines a complex curve $\tSigma$ that is the branched $\IZ$-covering of $\Sigma$ with ramification at logarithmic branch points and logarithmic singularities. 
The projection map is denoted by
$\tilde\pi:\tSigma\to\Sigma$.

By construction $\lambda_{\pm,n}$ is single valued on $\tSigma$, and its periods define the BPS central charge homomorphism
$Z \in \Hom(H_1(\tSigma,\IZ),\IC)$.
However this map has a large kernel: there are infinitely many cycles with the same period, all projecting down to the same cycle in $\Sigma$. \footnote{From the point of view of the five-dimensional field theory, all these cycles should be associated to the same physical charge.}

There are two types of generators for $\ker Z$.
\begin{enumerate}
\item
Any closed cycle on $\tSigma$ comes in a $\IZ$-family with the same pushforward $\tilde\pi_*$ to $\Sigma$. 
The evaluation of $\lambda$ along different representatives differs by the logarithmic index $n$ in \eqref{eq:lambda-pm-n} and this difference cancels out in the computation of periods along closed contours, by definition of $\tSigma$ as the $\mathbb{Z}$-covering for $\log y$ branching. 
\item
The second type of generator for $\ker Z$ originates from the $\IZ_2$ symmetry of the curve in involutive form $y_+ = y_-^{-1}$. 
This symmetry acts on closed paths on $\Sigma$ by sending each point of a path from one sheet to the opposite one, effectively replacing $\lambda_{\pm,n}\to \lambda_{\mp,n}$. 
If $y_\pm(x)$ tend to finite values at $x\to 0$ (respectively $x\to\infty$) then $\Sigma$ has punctures where $\lambda$ has a simple pole, and therefore $Z$ takes a finite value when evaluated on a cycles $c_\pm$ surrounding the two lifts of $x=0$ (respectively $x=\infty$). Due to the $\IZ_2$ symmetry we have that any lift of $c_++c_- $ to $H_1(\tSigma,\IZ)$ belongs to $\ker Z$.
\end{enumerate}

\begin{definition}\label{def:physical}
    We define the physical charge lattice as the set of cycles along which the period $\lambda$ is well-defined and unique. Assuming that moduli of $\Sigma$ are generic, this is the quotient
\be\label{eq:physical-charges-def}
	\Gamma := H_1(\tSigma,\IZ)  / \ker Z\,.
\ee
\end{definition}
Generic closed paths on $\Sigma$ do not admit closed lifts to $\tSigma$.
If such a lift exists, then (the equivalence class of) its homology class $\gamma \in \Gamma$ corresponds to a physical charge (this should be contrasted with the usual second-order differential case, where every closed cycle in $H_1(\Sigma,\mathbb{Z})$ can be thought as representative of an equivalence class in the physical charge lattice of anti-invariant cycles $H_1^{-}(\Sigma,\mathbb{Z})$). 

\begin{definition}
    For a closed physical cycle $\gamma\in H_1(\tSigma,\mathbb{Z}) $, we define classical and WKB periods as follows. 
    
    Classical periods, or central charges, are the period integrals
    \begin{equation}
        Z_\gamma^{(\pm,n)}:=\oint_\gamma \lambda_{\pm,n}=\oint_\gamma\left(\log y_\pm+2\pi i n \right)\frac{\dd x}{x} 
    \end{equation}
    of the Seiberg-Witten differentials $\lambda_{\pm,\,n}$. 
    
WKB quantum periods are formal power series which are $\hbar$-deformations of classical ones, corresponding to integrals of the WKB differential 
\be\label{eq:q-lambda}
	V_\gamma^{(\pm,n)}:=\frac{1}{\hbar}\oint_\gamma\lambda_{\pm,\hbar},\quad  \lambda_{\pm,n,\hbar} := \left(\log\R_\pm(x,\hbar)+2\pi i n\right)\frac{\dd x}{x}=\lambda_{\pm,n}+O(\hbar) .
\ee
\end{definition}

The WKB periods \eqref{eq:q-lambda} define an $\hbar$-deformation of the BPS central charge homomorphism
\be
	Z_\hbar \in \Hom(H_1(\tSigma,\IZ),\IC[[\hbar]])\,.
\ee
It is natural to ask whether this descends to a well-defined map on the quotient \eqref{eq:physical-charges-def} that defines the physical charge lattice, or whether it is only well defined on $H_1(\tSigma)$.

\begin{conjecture}\label{prop:kernel-inclusion}
There is an inclusion of kernels 
\be\label{eq:kerZ-inclusion}
	\ker Z \subseteq \ker Z_\hbar\,.
\ee
\end{conjecture}

We find that this conjecture holds in all examples that we consider. In fact, we can argue that this should always hold, as follows.
Recall that $\ker Z$ has two types of generators. The first type corresponds to taking $c_n - c_m$, the difference of two cycles with the same pushforward to $\Sigma$, i.e. $\tilde\pi_* (c_{n})=  \tilde\pi_* (c_{m}) \in H_1(\Sigma,\IZ)$.
The integral along these cycles only differs in the logarithmic shift in $\lambda$, recall \eqref{eq:lambda-pm-n}.
As observed in the previous section, the wavefunction has a simple dependence on the logarithmic branch \eqref{eq:log-index-$q$-periodicity}. In turn, this implies that $\lambda_\hbar$ has exactly the same dependence as $\lambda$ on the logarithmic index.
This proves the conjecture for the generators of $\ker Z$ of the first type.
For the second type of generators observe that $c_\pm$ correspond ti cycles with trivial intersection pairing (also called flavor cycles), since they are contractible near a puncture of $\Sigma$. Central charges of flavor cycles are not expected to receive any quantum corrections in $\hbar$, therefore if $Z_\hbar(c_\pm) = Z(c_\pm)$, and if $c_++c_-\in \ker Z$ it also belongs to $\ker Z_\hbar$.

 A consequence of \eqref{eq:kerZ-inclusion} is that the quantum BPS central charge $Z_\hbar$ is well-defined on the physical charge lattice $\Gamma$.
 \begin{definition}\label{def:Voros}
Let $\gamma\in\Gamma$. We define the Voros symbol $\CU_\gamma$ as the formal series defined by the quantum period of $\gamma$ up to an overall normalization\footnote{This definition is in analogous to that of Voros symbols in the context of differential equations, see \cite{IwakiNakanishi1}.}
\be\label{eq:Voros-gen-def}
	\CU_\gamma:=e^{V_\gamma}=e^{\frac{Z_\gamma}{\hbar}}(1+O(\hbar))  \,.
\ee
The quantum period $V_\gamma$ along an element $\gamma$ of the physical charge lattice $\Gamma$ will be called Voros coefficient.
 \end{definition}
 Note that while the integral of the quantum period is taken over a representative in $H_1(\tSigma,\IZ)$ of the equivalence class $\gamma\in\Gamma$,
different choices are related by a change in the logarithmic branch, and thanks to \eqref{eq:kerZ-inclusion} this does not affect $ V_\gamma$. We will now introduce a set of coordinates for monodromies of $q$-difference equations, closely related to Voros symbols.

\subsection{A review of Fock-Goncharov coordinates from exact WKB analysis}

Before proceeding further, we briefly review how Voros symbols in the setting of differential equations may be expressed in terms of certain cross-ratios of Wronskians. Our exposition will be brief, more details can be found in \cite{FockGoncharovHigher, Gaiotto:2009hg, Gaiotto:2012db, Allegretti:2018kvc}.\footnote{For conventions see \cite{IwakiNakanishi1, Longhi:SISSAtalk}.}

\subsubsection{Exact WKB bases of solutions}

In the context of Schr\"odinger equations of the form
\begin{equation}\label{eq:Schrodinger}
	\left(-\hbar^2\frac{\partial^2}{\partial x^2}+Q(x;\hbar) \right)\psi=0,\qquad Q(x;\hbar)=Q_0(x)+O(\hbar),
\end{equation}
the WKB method produces asymptotic series solutions 
\be\nonumber
	\psi_\pm^{\fb} = \frac{1}{\sqrt{S_{odd}}} \exp\(\pm \int_{\fb}^x S_{odd}(x',\hbar) \, dx'\)
\ee
where $\fb$ is a branch point used for normalization and $S_{odd}$ is an asymptotic series obtained by solving the Riccati equation order by order in $\hbar$. In general this equation will be defined by an oper on a Riemann surface, here for simplicity we will assume that this is a differential equation on the complex plane.

It is important to realize that $\psi_\pm^{\fb}$ are defined everywhere away from a system of branch cuts. The Stokes graph $\CW$ is defined by
$
	\Im \, \hbar^{-1}\int_{\fb}^x \sqrt{Q_0(x')}\, dx' = 0
$
for lines emanating from $\fb$. A line is said to be of type $\pm$ (resp. $\mp$) if the solution $\psi^{\fb}_-$ (resp. $\psi^{\fb}_+$) decays exponentially along that line going away from~$\fb$ (see discussion in Section \ref{sec:Stokes-graph}). 
The Stokes graph together with branch cuts divide $\mathbb{C}$ into regions $\{{R}_i\}$, and inside each region Borel summation produces locally analytic functions
$
	\varphi^{{R},\fb}_{\pm} := \CS_{{R}}[\psi^{\fb}_{\pm}]
$
when $x\in{R}$.
The functions defined in this way within a region ${R}$ can be extended globally to all of $\mathbb{C}\setminus\{{\rm cuts}\}$, by analytic continuation. These functions are related to each other by constant linear transformations. \footnote{It should be noted that monodromy around a branch point with a single branch cut is $S_0S_{+}S_- S_+ = S_0S_-S_+ S_- = -\mathbb{I}_2$. In the previous section we had instead set up the Stokes matrices so that total monodromy around square-root branch points was equal to the identity. An analogous parametrisation in the differential case was used in \cite{Gavrylenko:2020gjb}.  The sign difference is due to the fact that there we used three square-root branch cuts around each branch point, and to the fact that $S_0^2 = -\mathbb{I}_2$. \label{foot:sign-bp}}
\begin{itemize}
\item Across a Stokes line of type $\pm$ from $\fb$ the basis of solutions changes as follows
\be\label{eq:stokes-jump}
	\( \varphi^{{R},\fb}_{+}, \varphi^{{R},\fb}_{-} \) = \( \varphi^{{R}',\fb}_{+}, \varphi^{{R}',\fb}_{-} \) \cdot S_\pm 
\ee
where ${R}$ lies clockwise of the Stokes line and ${R}'$ lies counter clockwise. The Stokes matrices are
\be\nonumber
	S_+ = \(\begin{array}{cc}
		1 & 0 \\
		i & 1
	\end{array}\)
	\qquad
	S_- = \(\begin{array}{cc}
		1 & i \\
		0 & 1
	\end{array}\)
\ee
At a line of type $+$ it is $\psi_+$ that jumps while $\psi_-$ stays constant, and the opposite at Stokes lines of type $-$.
\item Crossing a branch cut exchanges $\varphi_{\pm}^{\fb,{R}}$, and adds a phase due to the normalization factor. Crossing the cut clockwise near a branch point where $S_{odd}\sim\hbar^{-1} \sqrt{x}$ gives
\be\label{eq:cut-jump}
	\( \varphi^{{R},\fb}_{+}, \varphi^{{R},\fb}_{-} \) = \( \varphi^{{R}',\fb}_{+}, \varphi^{{R}',\fb}_{-} \) \cdot S_0
\ee
where ${R}$ lies clockwise of the branch cut and ${R}'$ lies counter clockwise, and 
\be\nonumber
	S_0 = \(\begin{array}{cc}
		0 & i \\
		i & 0
	\end{array}\)
\ee
\item 
Certain regions connect more than one branch point. If $R$ is such a region, introducing the Borel summation of a Voros symbol (with a slight abuse of notation, we will refer to both as Voros symbol in what follows)
\begin{equation}
    \CU_{\fb\fb'}:=\CS_R(\Uij_{\fb\fb'}),\qquad \Uij_{\fb\fb'}:=
\exp\int_{\fb}^{\fb'}S_{odd}(x,\hbar)\dd x,
\end{equation}
the corresponding solutions are related by
\be\label{eq:voros-jump}
	\( \varphi^{{R},\fb}_{+}, \varphi^{{R},\fb}_{-} \) = \( \varphi^{{R},\fb'}_{+}, \varphi^{{R},\fb'}_{-} \) \cdot \Uij_{\fb,\fb'} .
\ee
where
\be\nonumber
	\Uij_{\fb,\fb'} = \(\begin{array}{cc}
		\CU_{\fb\fb'} & 0 \\
		0 & \CU_{\fb\fb'}^{-1}
	\end{array}\).
\ee

\end{itemize}

\subsubsection{Rays in the space of solutions and Fock-Goncharov coordinates}

Let $\CV$ denote the 2-dimensional $\mathbb{C}$-vector space of solutions of the second order ODE \eqref{eq:Schrodinger}, and $\varphi_{\pm}^{{R},\fb}$ a basis of solutions normalised at the branch point $\fb$. 
Transformations \eqref{eq:stokes-jump}, \eqref{eq:cut-jump} and \eqref{eq:voros-jump} express the change of basis among pairs defined by Borel summation within each region and for each choice of branch point.
If the surface has punctures, these come in two types: regular and irregular, with irregular punctures featuring Stokes sectors. After a real blowup an irregular puncture is replaced by a disk with marked points on the boundary.

Let $p$ be any regular puncture, or any marked point of an irregular one.
To $p$ we can canonically associate a line $\CL\subset \CV$ generated by the solution that vanishes at that puncture. If Stokes lines of type $+-$ end on $p$ the line is generated by $\varphi_{-}^{{R},\fb}$, and viceversa
\be
\CL_p :=
\left\{\begin{array}{l}
	\IC\, \varphi_{-}^{{R},\fb} \qquad \text{if $+-$ Stokes lines end on $p$}\\
	\IC\, \varphi_{+}^{{R},\fb} \qquad \text{if $-+$ Stokes lines end on $p$}\\
\end{array}\right.
\ee
Any two regions ${R}$ and ${R}'$ whose boundary contains $p$ define the same line.

The Fock-Goncharov coordinate system is defined by a triangulation. Given a triangulation, any internal edge separates a pair of adjacent triangles. Let $p_i$ with $i=1,2,3,4$ be the four punctures at the corners of the corresponding quadrilateral, labeled counterclockwise. Let $s_i$ be a choice of generator for $\CL_{p_i}$, and define $\fs_i:=\left(\begin{array}{c} s_i  \\ s_i' 
\end{array} \right)$. The Fock-Goncharov coordinate associated to the edge is 
\be\label{eq:FG}
	\CU_E = -\frac{(\fs_{1}\wedge \fs_2)(\fs_3\wedge \fs_4)}{(\fs_2\wedge \fs_3)(\fs_4\wedge \fs_1)}.
\ee
We will now recall the relation between the Fock-Goncharov coordinate $\CU_E$ and the Voros symbol $\CU_\gamma$ from \eqref{eq:Voros-gen-def}.

\subsubsection{Fock-Goncharov coordinates from WKB approximation for Schr\"odinger equations}\label{sec:FGCrossRatioDiff}

The WKB approximation introduces a distinguished triangulation, defined by choosing edges to be generic leaves of the foliation induced by the quadratic differential $Q_0$ that defines the Schr\"odinger potential in \eqref{eq:Schrodinger}, see \cite{Gaiotto:2009hg}. Because of this, for Schr\"odinger equations we can equivalently view the Fock-Goncharov coordinates \eqref{eq:FG} as defined by the WKB Stokes graph, instead of a triangulation. Such a reformulation will be crucial when discussing $q$-difference equations, as in that case the Stokes graph is not dual to a triangulation.

\begin{figure}[h]
\begin{center}
 \begin{tikzpicture}
\draw[fill] (-2.0,0) node (v1) {} circle[radius=0.7mm];
\draw[fill](1.5,0) node (v2) {} circle[radius=0.7mm];
\draw[thick] (v1.center) to [out=60, in=-50] (-3,3);
\draw[thick] (v1.center) to  (-3.5,0);
\draw[thick] (v1.center) to [out=-60, in=150] (2.4,-3);
\draw[thick] (v2.center) to [out=120, in=-30] (-2.4,3);
\draw[thick] (v2.center) to  (3,0);
\draw[thick] (v2.center) to [out=-120, in=125] (2.9,-3);
\draw[red,double] (2.7,-3)--(-2.7,3);

\draw[thick, green!70!black,dashed, <-](v2) to[bend right=15] (v1);
\node at (0,-0.5){$\Uij_{\fb\fb'}$};;
\node at (3.5,0){\small\color{blue}$(\mp,\pm)$};
\node at (-4.0,0){\small\color{blue}$(\pm,\mp)$};
\node at (-1.4,-0.1){\tiny$(\fb, I)$};
\node at (-2.6,0.3){\tiny$(\fb,II)$};
\node at (-2.6,-0.3){\tiny$(\fb,III)$};
\node at (1,-0.2){\tiny$(\fb', I)$};
\node at (2.1,0.3){\tiny$(\fb',III)$};
\node at (2.1,-0.3){\tiny$(\fb',II)$};

\node at (-2,3){\tiny 2};
\node at (3,0.5){\tiny 1};
\node at (2,-3){\tiny 4};
\node at (-3.5,-0.5){\tiny 3};

\end{tikzpicture}
\caption{Neighboring turning points. \label{fig:FGSchrod}}
\end{center}
\end{figure}
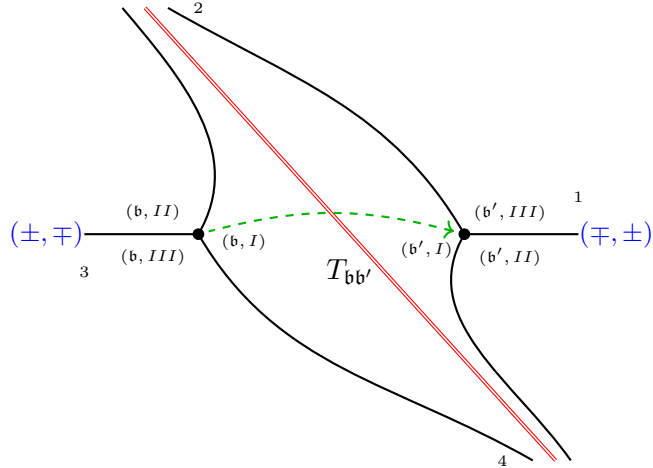

Consider two branch points $\fb,\,\fb'$ sharing a region defined by the WKB Stokes graph, and a path connecting the two branch points. We choose the branch cuts close to the Stokes lines as in Section \ref{Sec:AiryStokes}. In this way, each region can be labeled by the pair of a branch point and a roman numeral, as in Figure \ref{fig:FGSchrod}. 

Introduce the Voros symbol associated to the parallel transport between branch points:
\begin{equation}
    \varphi^{(\fb',I)}_\pm=\varphi^{(\fb,I)}_\pm\cdot(\CU_{\fb\fb'})^{\pm1},
\end{equation}
which in matrix form reads
\begin{equation}\label{eq:VorosMatrix}
    \Psi_\fb^I(x)\Uij_{\fb\fb'}=\Psi_{\fb'}^I(x)\,,\qquad \Uij_{\fb\fb'}=\left( \begin{array}{cc}
       0  & i\CU_{\fb\fb'} \\
       i\CU_{\fb\fb'}^{-1}  & 0
    \end{array} \right).
\end{equation}
We can introduce the four vanishing solutions as
\be
	s_1 = \varphi_\mp^{(\fb',III)}=\left(\Psi^{(\fb',III)} \right)_{i2},
	\qquad
	s_2 = \varphi_{\pm}^{(\fb',III)}=\left(\Psi^{(\fb',III)} \right)_{i1} , \nonumber
\ee
\be
	s_3 = \varphi_{\pm}^{(\fb,III)}=\left(\Psi^{(\fb,III)} \right)_{i2} ,
	\qquad
	s_4 = \varphi_{\mp}^{(\fb,III)}=\left(\Psi^{(\fb,III)} \right)_{i1} .
\ee

We will now verify that the Fock-Goncharov coordinate defined by the cross-ratio \eqref{eq:FG} of these four solutions coincides with the square of the Voros symbol $\CU_{\fb\fb'}$.
First note that
\begin{equation}
    \fs_1\wedge\fs_2=-\det\Psi^{(\fb',III)}=-1, \qquad \fs_3\wedge\fs_4=-\det\Psi^{(\fb,III)}=-1\,.
\end{equation}
Using the explicit form of the matrices $S$ and $\Uij_{\fb\fb'}$ in \eqref{eq:S-matrix} and \eqref{eq:VorosMatrix}, we also find
\begin{equation}
    s_3=\left(\Psi^{(\fb,III)} \right)_{i2}=\left(\Psi^{(\fb',I)}S\Uij_{\fb\fb'}^{-1}S^2 \right)_{i2}=-\CU_{\fb\fb'}\fs_1-i\left(\CU_{\fb\fb'}+\CU_{\fb\fb'}^{-1} \right)\fs_2,
\end{equation}
\begin{equation}
    \fs_4=\left(\Psi^{(\fb',I)}S\Uij_{\fb\fb'}^{-1}S^2 \right)_{i1}=-i\CU_{\fb\fb'}\fs_1+\CU_{\fb\fb'}\fs_2
\end{equation}
so that
\begin{equation}
    \fs_2\wedge\fs_3=\CU_{\fb\fb'}\fs_1\wedge\fs_2=-\CU_{\fb\fb'},\qquad \fs_1\wedge\fs_4=\CU_{\fb\fb'}\fs_1\wedge\fs_2=-\CU_{\fb\fb'}.
\end{equation}
The Fock-Goncharov edge coordinate is then
\begin{equation}
	\CU_E = \frac{(\fs_{1}\wedge \fs_2)(\fs_3\wedge \fs_4)}{(\fs_2\wedge \fs_3)(\fs_4\wedge \fs_1)}=\frac{1}{\CU_{\fb\fb'}^2 }.
\end{equation}

The four lines $ s_i\, \IC$ that appear in the definition \eqref{eq:FG} of Fock-Goncharov coordinates have then been expressed directly in terms of the WKB basis defined in the two regions $(\fb,III)$ and $(\fb',III)$ adjacent to the common one, labeled by the numeral $I$. 

\begin{remark} \label{rmk:AntiInv}
Recall that the lattice of physical charges here is the anti-invariant homology $H^-(\Sigma,\mathbb{Z})$. We can think of an anti-invariant cycle $\gamma$ on $\Sigma$ being constructed by two copies of an open segment connecting the two branch points on different sheets. This explains the square and tells us that
$\CU_E=\CU_\gamma$.

Furthermore, the fact that the Voros symbol is squared resolves the sign ambiguity in the normalization of wavefunctions. 
\end{remark}

\subsection{WKB coordinates for $q$-difference equations}

We now generalise the above construction of Fock-Goncharov coordinates in terms of Wronskians to the setting of $q$-difference equations.

\subsubsection{Arc decomposition of homology cycles}\label{sec:arc-splitting}
We begin by constructing a map that associates to any physical charge $\gamma\in \Gamma$ a system of arcs on $\IC^*_x$ labeled by pairs of sheets of $\tSigma$. This generalises the simple observation in Remark \ref{rmk:AntiInv}, that the anti-invariant cycle $\gamma$ around two branch points of the WKB curve for a second-order ODE can be written as the sum of two open segments connecting the two branch points on the base. 

Start by fixing a choice of representative for $\gamma$ in $H_1(\tSigma,\mathbb{Z})$.
This can be pushed forward via the projection maps $\tilde\pi:\tSigma\to\Sigma$ and $\pi:\Sigma\to\IC^*_x$ down to a closed path on $\IC^*_x$.
The resulting cycle on $\IC^*_x$ will generally cross branch cuts both of logarithmic and square-root types.
Dividing the path into components cut out by these intersections we assign a label $(s,n)$ to each component that keeps track of the origin of that path on the $s$-th sheet of $\Sigma$ and on the $n$-th lift of that sheet to $\tSigma$.
The resulting collection of arcs with labels retains all information about the original cycle on $\tSigma$.
Clearly, the WKB quantum period \eqref{eq:q-lambda} can be expressed as a sum of integrals of the WKB differential $\lambda_{s,n,\hbar}$ along these labeled arcs.

\subsubsection{Fock-Goncharov coordinates for $q$-difference equations}\label{sec:qFG}

We now turn to the task of defining Voros symbols of open paths for $q$-difference equations.
As explained in section \ref{sec:2nd-order-qDE-WKB}, WKB solutions of second-order $q$-difference equations are labeled by two indices $(s,n)$ where $s=\pm$ and $n\in \IZ$. 
For any path connecting two branch points, we define the associated Voros symbol as the ratio of solutions with equal logarithmic index
\begin{equation}\label{eq:qVorosDef}
    \CU_{\fb\fb'}^{\pm1}:=\frac{\varphi_{s,0}^{(\fb',R)}}{\varphi_{s',0}^{(\fb,R)}}
\end{equation}
in a region $R$ connecting two quadratic branch points $\fb,\,\fb'$. The ratio of solutions with logarithmic index different than zero will be instead
\begin{equation}
    \frac{\varphi_{s,n}^{(\fb',R)}}{\varphi_{s',n}^{(\fb,R)}}=\CU_{\fb\fb'}^{\pm1}\left(\frac{\fb}{\fb'}\right)^{\frac{2\pi in}{\hbar}}.
\end{equation}
The signs $s,s'$ can either be the same or opposite, depending on the signatures of $\fb, \fb'$. In what follows, we will consider only matrices representing analytic continuation of matrices with zero logarithmic index.

\begin{enumerate}
\item When the lines emanating from $\fb$ are of type $(\pm\mp,k)$ and the lines emanating from $\fb'$ are of type $(\mp\pm,k')$, the situation is identical as in the differential case (see Figure \ref{fig:qFG1}). The matrix expressing analytic continuation of solutions normalised at different branch points is then
\begin{equation}\label{eq:qVorosMatrix1}
    \Psi^R_\fb(x) \Uij_{\fb\fb'}=\Psi^R_{\fb'}(x),\qquad \Uij_{\fb\fb'}=\left(\begin{array}{cc}
    0     & i\CU_{\fb\fb'} \\
    i\CU_{\fb\fb'}^{-1}     & 0
    \end{array} \right).
\end{equation}
In this case the signs in \eqref{eq:qVorosDef} are identical $s=s'$.

\item In the case of $q$-difference equations, it is possible that the lines emanating from $\fb$ and $\fb'$ are of type $(\pm\mp,k)$ and $(\pm\mp,k')$ respectively. In this case, we have instead
\begin{equation}\label{eq:qVorosMatrix2}
    \Psi^R_\fb(x) \tilde \Uij_{\fb\fb'}=\Psi^R_{\fb'}(x),\qquad \tilde \Uij_{\fb\fb'}=\left(\begin{array}{cc}
    \CU_{\fb\fb'} & 0 \\
    0 & \CU_{\fb\fb'}^{-1}
    \end{array} \right).
\end{equation}
In this case the signs in \eqref{eq:qVorosDef} are opposite $s=-s'$.

\item Finally, due to the possibility of intersections between Stokes lines, it can happen that there is a range of phases of $\hbar$ for which two branch points have a common region, and another range where they don't, without any saddle in between. In this case, we define the matrix of analytic continuation between the two branch points as in cases 1 and 2, that incorporates analytic continuation across possibly infinite secondary Stokes lines. We will discuss this possibility more extensively in our first example, the q-Airy equation of Section \ref{sec:q-Airy}.
\end{enumerate}
\begin{figure}[h]
\begin{center}
\begin{subfigure}{.5\textwidth}
\centering
 \begin{tikzpicture}[scale=0.85]
\draw[fill] (-2.0,0) node (v1) {} circle[radius=0.7mm];
\draw[fill](1.5,0) node (v2) {} circle[radius=0.7mm];
\draw[thick] (v1.center) to [out=60, in=-50] (-3,3);
\draw[thick] (v1.center) to  (-3.5,0);
\draw[thick] (v1.center) to [out=-60, in=150] (2.4,-3);
\draw[thick] (v2.center) to [out=120, in=-30] (-2.4,3);
\draw[thick] (v2.center) to  (3,0);
\draw[thick] (v2.center) to [out=-120, in=125] (2.9,-3);
\draw[red,double] (2.7,-3)--(-2.7,3);

\draw[thick, green!70!black,dashed, <-](v2) to[bend right=15] (v1);
\node at (0,-0.5){$\Uij_{\fb\fb'}$};
\node at (3.9,0){\small\color{blue}$(\pm,\mp,k')$};
\node at (-4.3,0){\small\color{blue}$(\mp,\pm,k)$};
\node at (-1.4,-0.1){\tiny$(\fb, I)$};
\node at (-2.6,0.3){\tiny$(\fb,II)$};
\node at (-2.6,-0.3){\tiny$(\fb,III)$};
\node at (1,-0.2){\tiny$(\fb', I)$};
\node at (2.1,0.3){\tiny$(\fb',III)$};
\node at (2.1,-0.3){\tiny$(\fb',II)$};

\node at (-2,3){\tiny 2};
\node at (3,0.5){\tiny 1};
\node at (2,-3){\tiny 4};
\node at (-3.5,-0.5){\tiny 3};

\end{tikzpicture}
\caption{Case 1}\label{fig:qFG1}
\end{subfigure}\hfill
\begin{subfigure}{.4\textwidth}
\centering
\begin{tikzpicture}[scale=0.85]
\draw[fill] (-2.0,0) node (v1) {} circle[radius=0.7mm];
\draw[fill](1.5,0) node (v2) {} circle[radius=0.7mm];
\draw[thick] (v1.center) to [out=60, in=-50] (-3,3);
\draw[thick] (v1.center) to  (-3.5,0);
\draw[thick] (v1.center) to [out=-60, in=150] (2.4,-3);
\draw[thick] (v2.center) to [out=120, in=-30] (-2.4,3);
\draw[thick] (v2.center) to  (3,0);
\draw[thick] (v2.center) to [out=-120, in=125] (2.9,-3);
\draw[red,double] (2.7,-3)--(-2.7,3);

\draw[thick, green!70!black,dashed, <-](v2) to[bend right=15] (v1);
\node at (0,-0.5){$\tilde\Uij_{\fb\fb'}$};
\node at (3.9,0){\small\color{blue}$(\pm,\mp,k')$};
\node at (-4.3,0){\small\color{blue}$(\pm,\mp,k)$};
\node at (-1.4,-0.1){\tiny$(\fb, I)$};
\node at (-2.6,0.3){\tiny$(\fb,II)$};
\node at (-2.6,-0.3){\tiny$(\fb,III)$};
\node at (1,-0.2){\tiny$(\fb', I)$};
\node at (2.1,0.3){\tiny$(\fb',III)$};
\node at (2.1,-0.3){\tiny$(\fb',II)$};

\node at (-2,3){\tiny 2};
\node at (3,0.5){\tiny 1};
\node at (2,-3){\tiny 4};
\node at (-3.5,-0.5){\tiny 3};

\end{tikzpicture}
\caption{}\label{fig:qFG2}
\end{subfigure}
\end{center}
\caption{}
\label{eq:qFG-both}
\end{figure}

We will now study the relation between the Voros symbol $\CU_{\fb\fb'}$ in \eqref{eq:qVorosDef} and the cross-ratio of q-Wronskians $ \frac{(\fs_1\wedge\fs_2)(\fs_3\wedge\fs_4)}{(\fs_2\wedge\fs_3)(\fs_1\wedge\fs_4)}$ (recall that $\fs_i=\left(\begin{array}{c}
     s_i(x)  \\
     s_i(qx)
\end{array}\right)$). We distinguish two cases, corresponding to transport between branch points with signatures as shown in Figure \ref{eq:qFG-both}.

\subsubsection*{Case 1:} 
If branch points have opposite signature, as in Figure \ref{fig:qFG1}, introduce the bases of solutions as in the differential case
\be
	s_1 = \varphi_{\mp,n_1}^{(\fb',III)}=\left(\Psi^{(\fb',III)} \right)_{i2},
	\qquad
	s_2 = \varphi_{\pm,n_2}^{(\fb',III)}=\left(\Psi^{(\fb',III)} \right)_{i1} , \nonumber
\ee
\be
	s_3 = \varphi_{\pm,n_3}^{(\fb,III)}=\left(\Psi^{(\fb,III)} \right)_{i2} ,
	\qquad
	s_4 = \varphi_{\mp,n_4}^{(\fb,III)}=\left(\Psi^{(\fb,III)} \right)_{i1} .
\ee
We included the possibility of having bases with arbitrary logarithmic index.

The bases are related by Stokes and transport matrices as follows
\begin{equation}
    s_3=\left(\Psi^{(\fb,III)} \right)_{12}=\left(\Psi^{(\fb',I)}S^{(\ell')}\Uij_{\fb\fb'}^{-1}(S^{(\ell)})^{-1} \right)_{12}=-is_2\left(\CU_{\fb\fb'}\xi_\fb^{\ell}+\xi_{\fb'}^{\ell'}\CU_{\fb\fb'}^{-1} \right)-s_1\CU_{\fb\fb'}, 
\end{equation}
\begin{equation}
s_4=\left(\Psi^{(\fb,III)} \right)_{12}=\left(\Psi^{(\fb',I)}S^{(\ell')}\tilde\Uij_{\fb\fb'}^{-1}(S^{(\ell)})^{-1} \right)_{11}=-is_1\CU_{\fb\fb'}+s_2\CU_{\fb\fb'}\xi_{\fb}^{\ell}
\end{equation}
where $\xi_\fb$ and $\xi_{\fb'}$ refer to the logarithmic shift parameter $\xi$ in \eqref{eq:xi-def} with $x_0$ chosen at the position of $\fb$ or $\fb'$ respectively.
Therefore the $q$-Wronskian recovers the transport coefficients multiplied by appropriate logarithmic shifts
\begin{equation}
    \frac{(\fs_1\wedge\fs_2)(\fs_3\wedge\fs_4)}{(\fs_2\wedge\fs_3)(\fs_1\wedge\fs_4)}=\CU_{\fb\fb'}^{-2}\,\xi_{\fb}^{\ell}\,\xi_{\fb'}^{-\ell'}.
\end{equation}

\subsubsection*{Case 2:}

If branch points have same signature, as in Figure \ref{fig:qFG2}, introduce the four vanishing solutions 
\be
	s_1 = \varphi_{\pm,n_1}^{(\fb',III)}=\left(\Psi^{(\fb',III)} \right)_{i2},
	\qquad
	s_2 = \varphi_{\mp,n_2}^{(\fb',III)}=\left(\Psi^{(\fb',III)} \right)_{i1} , \nonumber
\ee
\be
	s_3 = \varphi_{\pm,n_3}^{(\fb,III)}=\left(\Psi^{(\fb,III)} \right)_{i2} ,
	\qquad
	s_4 = \varphi_{\mp,n_4}^{(\fb,III)}=\left(\Psi^{(\fb,III)} \right)_{i1} .
\ee
Here we have included the possibility of having bases with arbitrary logarithmic index.
Using the explicit form of the matrices $S^{(\ell)}$ and $\tilde\Uij_{\fb\fb'}$ in \eqref{eq:StokesLog} and \eqref{eq:qVorosMatrix2}, we find
\begin{equation}
    s_3=\left(\Psi^{(\fb,III)} \right)_{12}=\left(\Psi^{(\fb',I)}S^{\ell}\tilde\Uij_{\fb\fb'}^{-1}(S^{(\ell')})^{-1} \right)_{12}=is_2\left(\frac{\xi^{\ell'}}{\CU_{\fb\fb'}}-\CU_{\fb\fb'}\xi^\ell \right)+s_1\CU_{\fb\fb'}^{-1},
\end{equation}
\begin{equation}
    \fs_4=\left(\Psi^{(\fb,III)} \right)_{12}=\left(\Psi^{(\fb',I)}S^{\ell}\tilde\Uij_{\fb\fb'}^{-1}(S^{(\ell')})^{-1} \right)_{11}=\CU_{\fb\fb'}\fs_2,
\end{equation}
so that 
\begin{equation}
    \fs_3\wedge \fs_4=\fs_1\wedge\fs_2,\quad \fs_2\wedge\fs_3=\CU_{\fb\fb'}^{-1}\fs_2\wedge\fs_1,\quad \fs_1\wedge\fs_4=\CU_{\fb\fb'}\fs_1\wedge\fs_2.
\end{equation}
We then conclude that the $q$-Wronskian is does not capture the Voros symbol, since
\begin{equation}\label{eq:qWronskiansCase2}
	\frac{(\fs_{1}\wedge \fs_2)(\fs_3\wedge \fs_4)}{(\fs_2\wedge \fs_3)(\fs_1\wedge \fs_4)}=1.
\end{equation}

As we noted earlier, in this case the path between branch points does not correspond to a saddle. This results in the cross-ratio of solutions being trivial. Nontheless, the Voros symbol $\CU_{\fb\fb'}$ provides a well-defined coordinate for the monodromies of the $q$-difference equation.

\subsection{Quantum periods and monodromy representation of $q$-difference equations}\label{sec:qMonodromy}

By concatenating matrices of type \eqref{eq:S-matrix}, \eqref{eq:StokesLog}, \eqref{eq:qVorosMatrix1}, \eqref{eq:qVorosMatrix2}, it is possible to explicitly compute monodromies arising from analytic continuation of WKB solutions. Monodromy matrices of $q$-difference equations in general are not constant matrices: they can have nontrivial $x$-dependence, but must be q-periodic. In our present setting, by construction the monodromy matrix will be a Laurent polynomial in the variable $x^{\frac{2\pi i}{\hbar}}$:
\begin{equation}\label{eq:GenericMonodromy}
    M(x)=\sum_k x^{\frac{2\pi ik}{\hbar}} M_k, \qquad M(qx)=M(x).
\end{equation}
For certain applications, for example if one wants to define an $\CX$-cluster variety structure on the space of monodromies, along the lines of what happens in differential equations \cite{FockGoncharovHigher,hikami2019note}, it can become important to extract $x$-independent information out of the monodromy. This is necessary if one wants to introduce an analogue of the infinite-dimensional Riemann-Hilbert Problem of \cite{Gaiotto:2008cd,Bridgeland:2016nqw} for monodromies of the $q$DE.

We do this by working $\mathbb{Z}$-equivariantly with respect to the logarithmic cover $\tSigma\rightarrow\Sigma$, which translates to keeping only the term of order zero in $x^{\frac{2\pi i}{\hbar}}$ when taking traces of \eqref{eq:GenericMonodromy}. This can be thought of as follows: one can think of the monodromy \eqref{eq:GenericMonodromy} in terms of the infinite-dimensional basis of WKB solutions over $\mathbb{C}$. Then it will act as an operator, that can be represented through the infinte matrix
\begin{equation}\label{eq:GenericInfinite}
    \mathbb{M}_{m,n}:=\sum_k M_k\delta_{m,n+k}.
\end{equation}
Taking the term of order zero in $x^{\frac{2\pi i}{\hbar}}$ in the trace of \eqref{eq:GenericMonodromy} corresponds then to a regularised trace of the indinite-dimensional matrix \eqref{eq:GenericInfinite}. We will now consider several examples, keeping in mind both these perspectives. The distinction between the two notions of monodromies becomes relevant when studying analytic continuation around logarithmic singularities, see Section \ref{sec:local-F0}.

\section{Examples}\label{sec:Examples}

\subsection{The $q$-Airy equation}\label{sec:q-Airy}

Consider the following $q$-difference equation 
\begin{equation}\label{eq:q-Airy}
	\psi(q^2x)-2x\psi(qx)+\psi(x)=0\,,
\end{equation}
A q-series solution to \eqref{eq:q-Airy} is 
\be
	\psi(x) 
	= e^{-\frac{i\pi}{2}\frac{\log x}{\log q}}\sum_{n\geq 0} \frac{q^{\frac{1}{2} n(n-1)}(-2ix)^n}{(q;q)_n(-q;q)_n}=e^{-\frac{i\pi}{2}\frac{\log x}{\log q}}\mathrm{Ai}_q(-2ix),
\ee
closely related to the definition of a $q$-Airy function $\mathrm{Ai}_q$ from \cite{Ismail2005,Morita2011}. This is not the type of solution that one encounters in the WKB analysis, as we will instead have asymptotic series in $\hbar=\log q$, that we choose to be normalised at the branch points $\Sigma$.

The WKB curve is 
\be	\label{eq:q-Airy-classical}
	y^2-2 xy+1=0 \,.
\ee
The curve has turning points at $x=\pm1$, a regular puncture at the origin $x=0$, and a logarithmic singularity at $x=\infty$. 
The two sheets are $y_\pm = x\pm\sqrt{x^2-1}$ and as $x\to\infty$, $y_+$ blows up while $y_-$ goes to zero. The logarithmic cut shown in dashed red in Figure \ref{fig:q-Airy-phase-I} connects the two logarithmic punctures at infinity.
At large $x$ we have the monodromy $\log y_\pm(x e^{2\pi i}) \approx \pm \log (x e^{2\pi i})= \log y_\pm(x) \pm  2\pi i $. This means that crossing the cut counterclockwise takes $\lambda_{\pm,n}\to\lambda_{\pm ,n\pm 1}$. In the notation of Section \ref{sec:log-cuts}, the puncture at infinity therefore has $k=-1$.
At $x=0$ the two sheets asymptote the constant values $y_\pm =  \pm i $. 

The local monodromy around the origin can be easily computed as follows: in \cite{Birkhoff1913} it is proven that for a $q$-difference equation $\Psi(qx)=A(x)\Psi(x)$, with
\begin{equation}
    A(x)=\sum_{k=0}^NA_kx^k,\qquad A_k\in GL(2,\mathbb{C}),\qquad\det A_0,\,\det A_N\ne 0,
\end{equation}
a solution normalised at the origin can be written as
\begin{equation}
    \Psi_0(x)= x^{\frac{\log\vec{\mu}}{\log q}}G(x),
\end{equation}
with $G(x)$ analytic at the origin, and $\vec{\mu}:=\diag(\mu_1,\,\mu_2)$, where $e^{\mu_1},\,e^{\mu_2}$ are the eigenvalues of $A_0$. In our present case, we have
\begin{equation}
	A(x    )=\left(\begin{array}{cc}
		0 & 1 \\ 
		-1 & 2qx
	\end{array} \right)\qquad A_0=\left(\begin{array}{cc}
	0 & 1 \\
	-1 & 0	
	\end{array} \right),\qquad \log\vec\mu=\frac{i\pi}{2}\sigma_3.
\end{equation}
An arbitrarily normalised solution of \eqref{eq:q-Airy}, will behave as follows around the origin:
\begin{equation}
    \Psi(x)=\Psi_0(x)C(x),\qquad C(x)\in SL(2,\CM_q), 
\end{equation}
The monodromy around the origin will then be given by
\begin{equation}\label{eq:MqAiry}
    M_0(x)=\Psi(x)^{-1}\Psi(e^{2\pi i}x)=C(x)^{-1}e^{-\frac{\pi^2}{\hbar}\sigma_3}C(e^{2\pi i}x).
\end{equation}

We can already compare with the periods of the leading order WKB differential \eqref{eq:psipmnleading}:
\be\label{eq:q-Airy-residue}
 \frac{1}{\hbar} \oint_{x=0} (\log y_\pm + 2\pi i n) \, d\log x = \mp  \frac{\pi^2  }{\hbar} - \frac{4\pi^2}{\hbar} n \,.
\ee
The first term is simply the monodromy exponent from Birkhoff's analysis. The second term comes from the monodromy of the connection matrix $C(x)$. Let us make this observation more precise. First let us note that both terms in \eqref{eq:q-Airy-residue} receive no $\hbar$-corrections: the first because it is the local monodromy exponent coming from the q-connection, the second because it is coming from periodicity of $C(x)$, which in turn originates from analytically continuing factors of $x^{\frac{2\pi i}{\hbar}}$.

With our current formalism, we can do more: we will compute the whole monodromy matrix  for solutions normalised at the branch points.

\begin{figure}[h!]
\begin{center}
\includegraphics[width=0.6\textwidth]{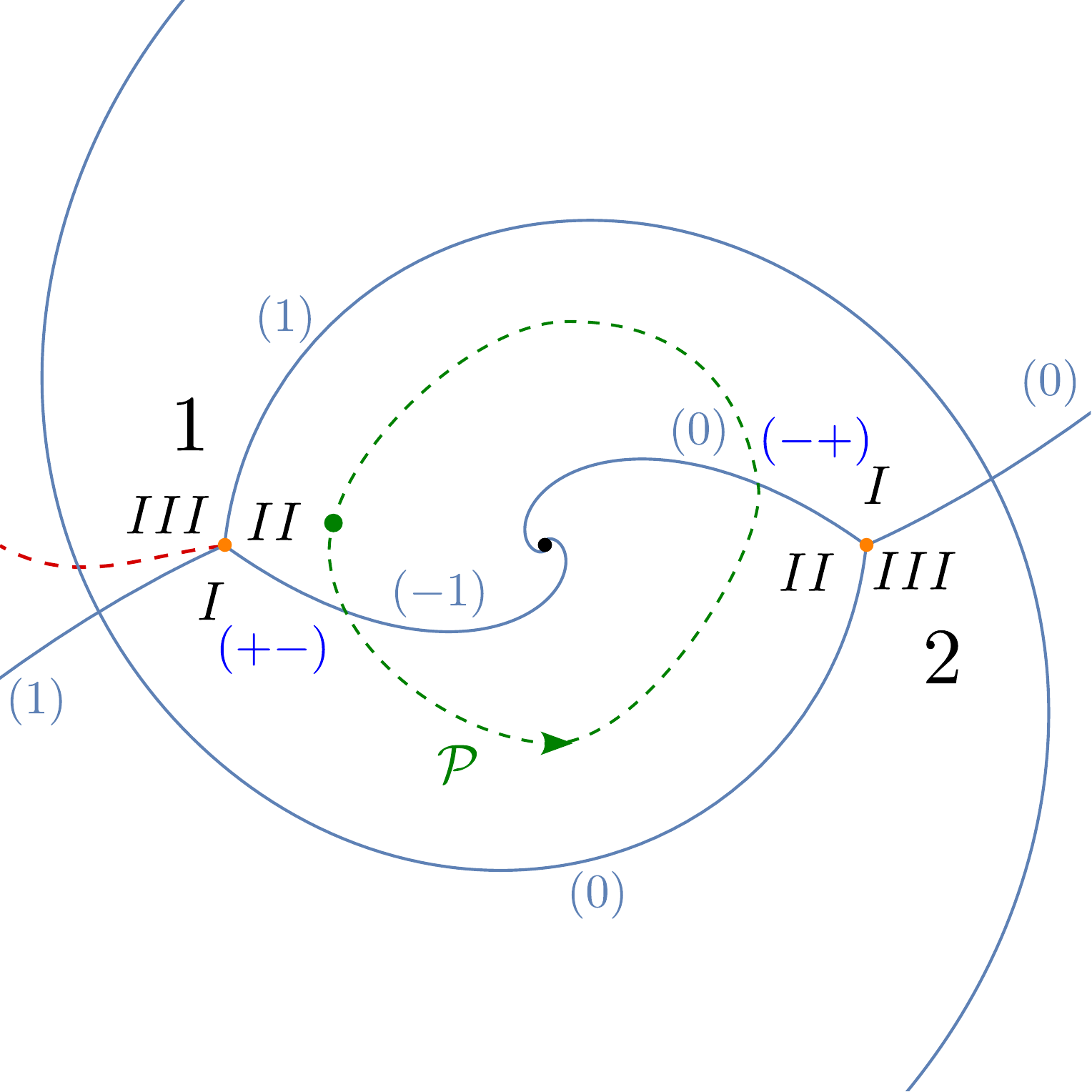}
\caption{
The Stokes graph of \eqref{eq:q-Airy} for $\vartheta =\arg\hbar= \frac{\pi}{5}$. Only primary Stokes lines (those that emanate directly from branch points) are shown. The label $(\ell)$ near each trajectory denotes the type of Stokes matrix $S^{(\ell)}$ associated to each Stokes line, as discussed in Section \ref{sec:stokes-data}. 
}
\label{fig:q-Airy-phase-I}
\end{center}
\end{figure}

\subsubsection{Monodromy}

To compute the monodromy around the origin, we start with a convenient choice of phase $\vartheta:=\arg\hbar$. Later we will discuss what happens for arbitrary phases.
The Stokes graph and the corresponding system of charts are given in Figure \ref{fig:q-Airy-phase-I}.
We label branch points by $a\in\{1,2\}$.
Following Sections \ref{Sec:AiryStokes} and \ref{sec:log-cuts}, we fix bases of solutions normalized at the same branch point and encoded by a matrix $\Psi_a^{K}$, with $K\in\{I,II,III\}$ labelling the three Stokes regions  in counterclockwise order.
The three bases of solutions around each branch point are related as discussed in Section \ref{sec:stokes-data}.\footnote{Recall that the left column of $\Psi_{a}^{K}$ corresponds to the solution that decays exponentially fast along the Stokes line separating regions $(K,K-1)$, and the right column corresponds to the solution that decays along the Stokes line separating regions $(K,K+1)$ (cyclic ordering is understood).}
Near $x=1$ we have 
\begin{align}
	\Psi_2^{K+I}=\Psi_{2}^{K}S^{(0)}\qquad (K\sim K+{III})\,,
\end{align}
while near $x=-1$ the three solutions are related as follows 
\be
    \Psi_1^{I} = \Psi_1^{III,\downarrow}\, S^{(1)}\,,
    \qquad
	\Psi_1^{II} = \Psi_1^{I}\, S^{(-1)}\,,
	\qquad
	\Psi_1^{III,\uparrow} = \Psi_1^{II}\, S^{(1)}\,,
	\qquad
	\Psi_1^{III,\downarrow} = \Psi_1^{III,\uparrow}\,\xi^{-\sigma_3}\,.
\ee
where $S^{(\ell)}$ are given by \eqref{eq:StokesLog}.
Region $I$ near $x=1$ and region $II$ near $x=-1$ coincide, as they are not separated by any Stokes line, and likewise for region $II$ near $x=1$ with region $I$ near $x=-1$.

Since the two branch points have opposite signature, the transport matrices between regions $(1,I)\to (2,II)$ and $(2,I)\to (1,II)$, namely. 
\be
	\Psi_2^{II} = \Psi_1^I\cdot \Uij_{\text{bottom}}\,,\qquad
	\Psi_1^{II} = \Psi_2^I\cdot \Uij_{\text{top}}\,,
\ee 
are off-diagonal. We will choose the basis of solutions in such a way that
\begin{equation}\label{eq:BasisqAiry}
    \Psi_1^{II}=\left(\begin{array}{cc}
        \varphi_{+,0}(x) & \varphi_{-,0}(x) \\
        \varphi_{+,0}(qx) & \varphi_{-,0}(qx)
    \end{array} \right),
\end{equation}
Due to the $\IZ_2$ symmetry $x\to -x, y\to -y$ the transports are actually identical
\be
	\Uij_{\text{top}} =\Uij_{\text{bottom}}
	= \left(\begin{array}{cc}
	0 & i \CU \\
	i \CU^{-1} & 0  \\
	\end{array}\right)\,,
\ee
depending on the logarithmic index $n$ in our WKB basis \eqref{eq:BasisqAiry}. 
The monodromy around the loop shown in Figure \ref{fig:q-Airy-phase-I} can be computed by matrix multiplication of transport and Stokes matrices

\begin{equation}
    M_\CP(x)=\left[S^{(-1)}(\xi)\right]^{-1} \Uij_{\text{bottom}} \left[S^{(0)}\right]^{-1} \Uij_{\text{top}}=\left(
\begin{array}{cc}
 -\frac{1}{\CU^2} & \qquad 0 \\
 i \left(\frac{1}{\xi  \CU^2}+1\right) & \qquad-\CU^2 \\
\end{array}
\right),
\end{equation}
and
\begin{shaded}
\be\label{eq:TrAiry}
\begin{split}
	{\rm tr}\,M_\CP
 = -\frac{1}{\CU^2} - \CU^2.
\end{split}
\ee
\end{shaded}
Comparing \eqref{eq:TrAiry} with \eqref{eq:MqAiry} and \eqref{eq:q-Airy-residue}, we obtain a direct map between our parameter $\CX$ and the only parameter $q$ entering in the q-Airy equation \eqref{eq:q-Airy}:
\begin{equation}
    \CU^2=-e^{-\frac{\pi^2}{\hbar}}.
\end{equation}

\subsubsection{Comparison with String Theory and D-brane charges}
As shown in Table \ref{tab:recap}, the q-Airy equation is the quantum mirror curve to the geometry $\mathbb{C}\times\left(\mathbb{C}^2/\mathbb{Z}_2 \right)$, whose only BPS state is the D0-brane (skyscraper sheaf), associated to the saddle at $\vartheta=0$. In the context of geometric engineering, the q-Airy equation \eqref{eq:q-Airy} represents the moduli space of a 3d defect theory on a toric brane probing an orbifold singularity $\IC\times (\IC^2/\IZ_2)$ resolved by $B$-field flux.

In type IIA string theory, the quantum period of a $D0$ brane charge can be written directly in terms of the central charge of a D0-brane and $\hbar$ (see e.g. \cite{DML2022}),
\begin{equation}
    X_{D0}=e^{-\frac{4\pi^2}{\hbar}}\,.
\end{equation}
Taking this into account, we can identify the Voros symbol $\CU$ with
\be
	\CU^2 =e^{-\frac{\pi^2}{\hbar}}= X_{\frac{1}{4} D0}\,.
\ee
The trace of the monodromy, which is the holonomy of a line operator in the geometrically engineered (trivial) 5d theory, is then
\be
	\tr M_\CP=X_{\overline{D0}}^{\frac{1}{4}} + X_{D0}^{\frac{1}{4}}   \,.
\ee
This is the character of an $SU(2)$ flavor symmetry associated with this geometry \cite{Dimofte:2010tz, Banerjee:2019apt}.

\subsubsection{Stokes graph and monodromy at other phases}

It follows from results of \cite{Banerjee:2019apt} 
that there is only one critical phase for the Stokes graph, namely $\arg \hbar=0$ in \eqref{eq:Stokes-line-Im-zero}.
The degenerate Stokes graph is shown in Figure \ref{fig:q-airy-saddles}.
However there are additional distinguished values of $\vartheta := \arg \hbar$, that are not saddles of the logarithmic differential, when it comes to defining a system of patches on $\IC^*_x$ in the context of exact WKB analysis.
In this example, there is one special value of $\vartheta = \vartheta_* \neq 0$ for which a primary Stokes line connects the two branch points directly.
This is not a saddle, because the second branch point is not the endpoint of a critical trajectory, but the Stokes line is only ``passing through''. 
 However, at $\vartheta_*$ the system of charts changes significantly.
In fact, for $\vartheta_*<\theta<\pi-\vartheta_*$ the system of charts defined by the Stokes graph features a pair of patches that connects the branch points at $x=\pm1$ while for other phases $0\leq \vartheta<\vartheta_*$ and $\pi-\vartheta_*<\vartheta\leq \pi$ this is not the case anymore, see Figure \ref{fig:q-Airy-phase-II}. The monodromy computation for all phases $\vartheta_*<\theta<\pi-\vartheta_*$ is identical to the one we just performed, as the topology of the Stokes diagram does not change.

\begin{figure}[h!]
\begin{center}
\includegraphics[width=0.5\textwidth]{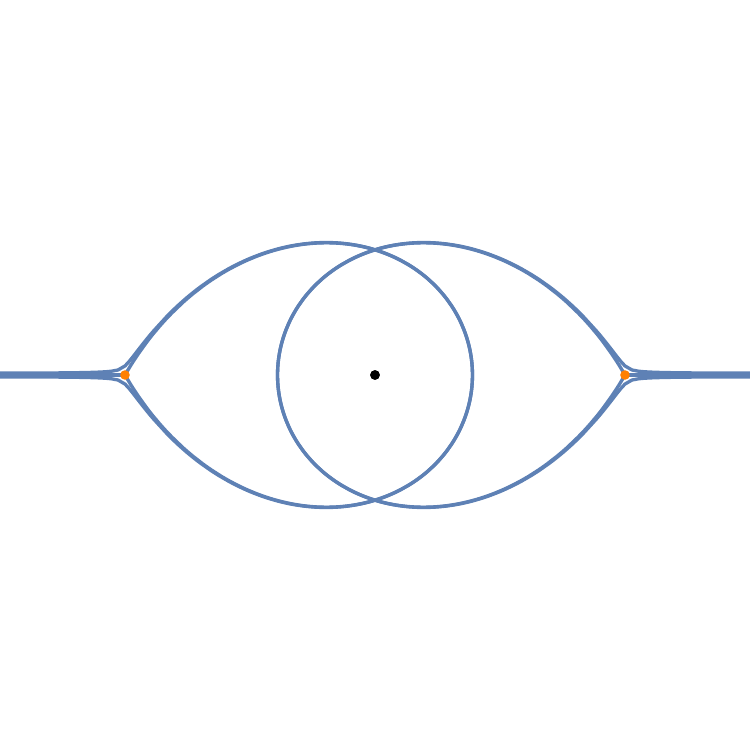}
\caption{Saddles formed by primary trajectories of the Stokes graph at $\vartheta=0$.}
\label{fig:q-airy-saddles}
\end{center}
\end{figure}

\subsubsection*{Phases $0\leq \vartheta<\vartheta_*$ and $\pi-\vartheta_*<\vartheta\leq \pi$}

In this case, there is no Stokes sector connecting the two branch points. Instead, there are infinitely many different configurations possible in this range, with a varying number $k$ of Stokes regions separating the two branch points. An example is shown in Figure \ref{fig:q-Airy-phase-II}. This is an example of Case 3 in Section \ref{sec:qFG}.

\begin{figure}[h!]
\begin{center}
\includegraphics[width=0.45\textwidth]{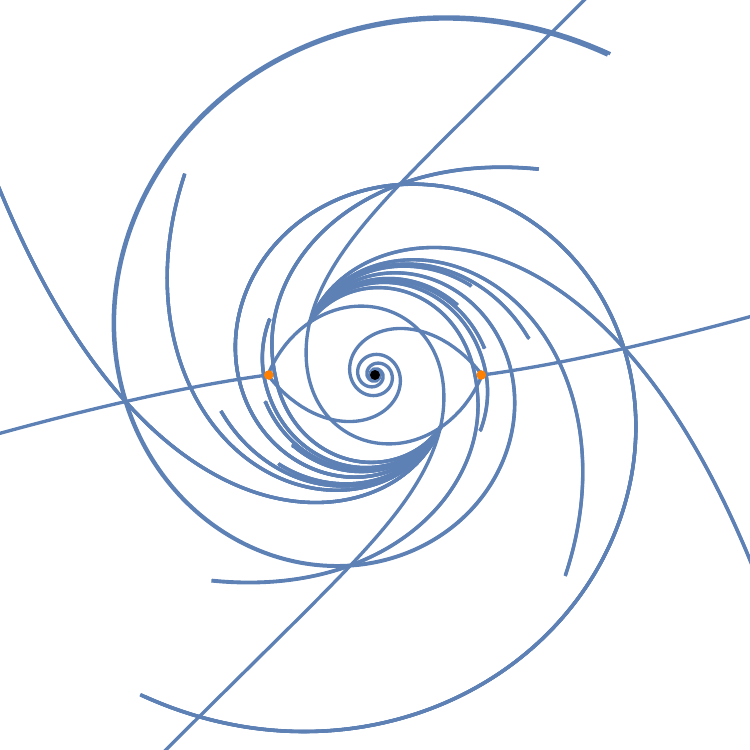}
\caption{The Stokes graph of the orbifold curve of $\CO(-2)\oplus\CO(0)$ for $0< \vartheta<\vartheta_*$. 
}
\label{fig:q-Airy-phase-II}
\end{center}
\end{figure}

To define variables $\CU_I, \CU_{II}$ encoding the change of normalisation between solutions at the two different branch points, we start by defining local solutions in the regions neighbouring with the two branch points as before. The analytic continuation of solutions will be decomposed into a product of Stokes jumps $S^{(i,\uparrow)}$ or $S^{(j,\downarrow)}$, $i=1,\dots,k_\uparrow$, $j=1,\dots,k_\downarrow$ along the internal Stokes lines in the regions above and below the puncture. However, by definition of our WKB bases near the two branch points, it is still true that the transition matrices between the initial and final region must be
	\begin{align}
	\Uij_{\text{top}}=\prod_{i=1}^{k_\text{top}}S^{(i,\text{top})}=\Psi_{1}^{II}(\Psi_{-1}^{I})^{-1}:=\left(\begin{array}{cc}
	0 & i \CU_n \\
	i \CU_n^{-1} & 0  \\
	\end{array}\right), \\ 
	\Uij_{\text{bottom}}=\prod_{i=1}^{k_\text{bottom}}S^{(i,\text{bottom})}=\Psi_{1}^{I}(\Psi_{-1}^{II})^{-1}:=\left(\begin{array}{cc}
	0 & i \CU_n \\
	i \CU_n^{-1} & 0  \\
	\end{array}\right).
\end{align} 
respectively above and below the puncture. This is necessary to match solutions decaying at zero an at infinity on the two sides. The monodromy is computed in the same way as before, leading to the same answer. This was originally expected, since the Stokes diagram has no jumps, so there is no wall-crossing in terms of our monodromy coordinates.

\subsubsection{Ramanujan function and the quantum mirror curve of $\mathbb{C}^3$}\label{sec:q-Ramanujan}

In the literature, one can find another q-discretisation of the Airy function, called the Ramanujan function \cite{Ismail2005,Morita2011}
\begin{equation}\label{eq:qRamf} 
\A_q(x):=\sum_{n=0}^{\infty}\frac{q^{n^2}}{(q;q)_n}(-x)^n\,.
\end{equation}
It solves the following equation
\begin{equation}\label{eq:qRameq}
	qx\psi(q^2x)-\psi(qx)+\psi(x)=0,
\end{equation}
whose WKB curve is the mirror curve of $\IC^3$ in quadratic framing \cite[eq. (6.2)]{Ekholm:2019lmb}
\be
	1 - y + x y^2 =0\,.
\ee

From the WKB/geometric point of view, these are naturally related. We can recast \eqref{eq:qRameq} in involutive form by rescaling
$\psi(x)\mapsto f(x)\psi(x)$ by a function with the property $f(qx)=x^{-\frac{1}{2}}f(x)$, such as
\begin{align}
	f(x)=e^{-\frac{\log x}{4\log q}\left(\frac{\log x}{\log q}-1 \right)}\,.
\end{align}
This leads to 
\begin{equation}
	\psi(qx)+\psi(q^{-1}x)=2x^{-\frac{1}{2}}\psi(x),
\end{equation}
i.e.
\begin{align}
T(x)=\frac{1}{\sqrt{x}}, && \lambda_n=\left[\log\left(\frac{1+\sqrt{1-x}}{\sqrt{x}} \right)+2\pi in \right]\frac{\dd x}{x}.
\end{align}
This is the same as the previous equation after redefining $x\mapsto x^{-2}$, illustrating that the change to the involutive form of the $q$DE can involve a covering map. Apart from this, the construction is entirely analogous to the case of the q-Airy equation.

\subsection{A deformation of the q-Airy equation}\label{sec:half-geometry}

We next consider a deformation of the $q$-Airy equation studied in Section \ref{sec:q-Airy} 
\begin{equation}\label{eq:q-Bessel}
\psi(qx)+\psi(q^{-1}x)=2(\kappa+x)\psi(x),\quad \Psi(qx)=A(x)\Psi(x),\quad A(x)=\left( \begin{array}{cc}
    0 & 1 \\
    -1 & 2(\kappa+x)
\end{array} \right) .
\end{equation}
We will call this the q-Airy$_\kappa$ equation. This equation adds generic monodromy around the origin to the q-Airy equation: as before, we can write $A(x)=A_0+xA_1$, and now
\begin{equation}
    \mathrm{Spec}\,A_0=\left\{\kappa+\sqrt{\kappa^2-1},\,\kappa-\sqrt{\kappa^2-1} \right\}:= \left\{e^\mu,\,e^{-\mu} \right\}.
\end{equation}
The monodromy around the origin of an arbitrarily normalised solution will then be of the form 
\begin{equation}\label{eq:AiryKappaM}
    M_0(x)=C(x)e^{\frac{2\pi i}{\hbar}\mu\,\sigma_3}C(e^{2\pi i}x)^{-1}.
\end{equation}
The WKB curve is
\begin{equation}\label{eq:curve-half-geometry}
	y+y^{-1}=2(\kappa+x).
\end{equation}
There is a turning point at $x=-\kappa\pm1$, a regular puncture at the origin, and a logarithmic singularity at $x=\infty$. 
The period of the WKB differential \eqref{eq:psipmnleading} around the origin is again purely semiclassical, and given by 
\be\label{eq:half-geom-residue}
	\frac{1}{\hbar} \oint_{x=0} (\log y_\pm + 2\pi i n) \, d\log x = \pm 2\pi i \frac{\mu }{\hbar} - \frac{4\pi^2}{\hbar} n \,.
\ee

Similarly to the case of q-Airy, the WKB Stokes diagram has a saddle at $\vartheta=0$, and two different behaviour for $\vartheta_*<\vartheta<\pi-\vartheta_*$ and $0\le\vartheta<\vartheta_*$ or $\pi-\vartheta_*<\vartheta\le\pi$, see Figure \eqref{fig:qBesselStokes}. Since the treatment is very similar to what we did in the q-Airy case, we will only do the computation of the monodromy for $\vartheta_*<\vartheta<\pi-\vartheta_*$, where the Stokes diagram takes the form of Figure~\ref{fig:qBesselPhase1}. 

\begin{figure}[ht!]
    \centering
    \begin{subfigure}[b]{0.45\textwidth}
        \includegraphics[width=\textwidth]{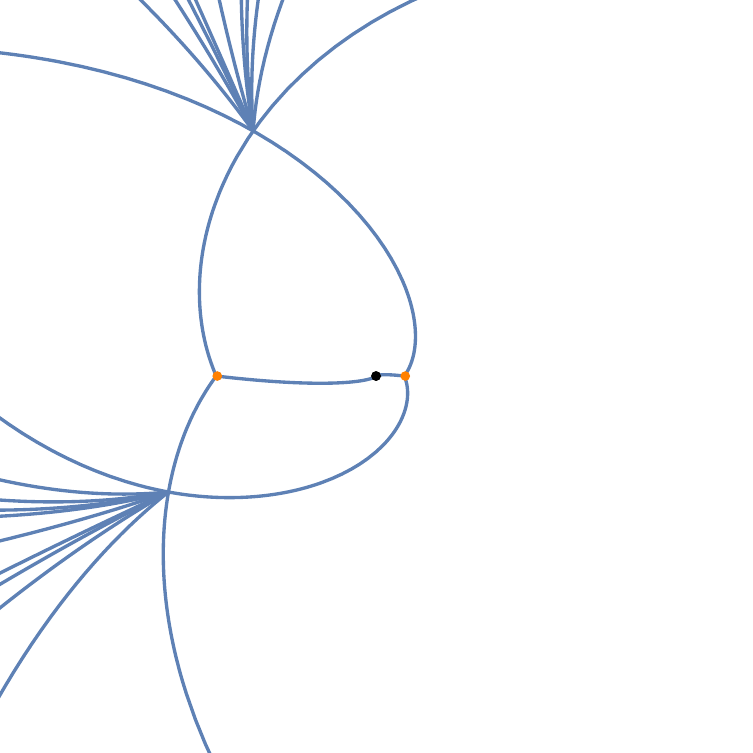}
        \caption{Stokes diagram for the q-Airy$_\kappa$ equation at $\vartheta=4\pi/9>\vartheta_*$.}
        \label{fig:qBesselPhase1}
    \end{subfigure}
    \hfill 
    \begin{subfigure}[b]{0.45\textwidth}
        \includegraphics[width=\textwidth]{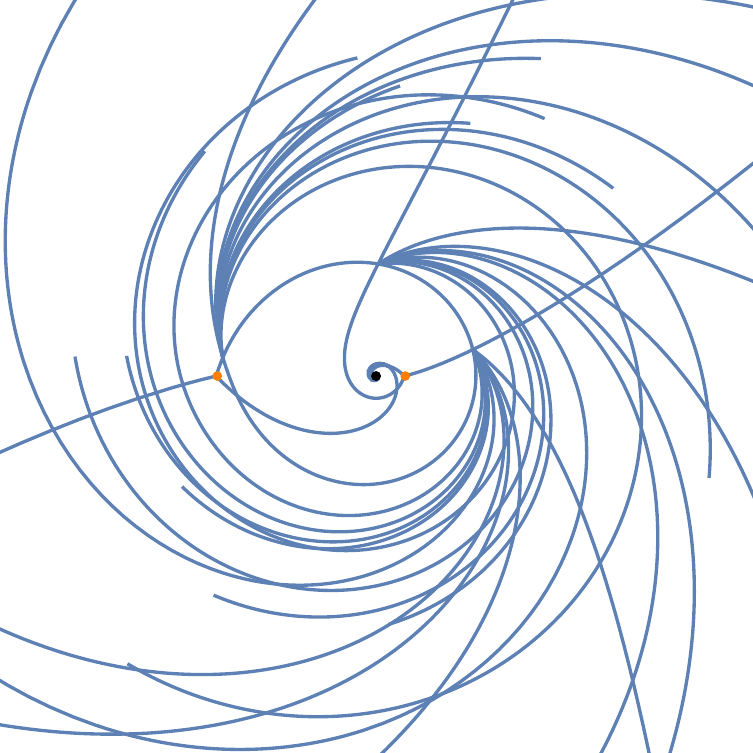}
        \caption{Stokes diagram for the q-Airy$_\kappa$ equation at $\vartheta=\pi/9<\vartheta_*$.}
        \label{fig:qBesselPhase2}
    \end{subfigure}
    \caption{Two types of Stokes diagrams for the q-Bessel equation.}
    \label{fig:qBesselStokes}
\end{figure}

\begin{figure}[h!]
\begin{center}
\includegraphics[width=0.75\textwidth]{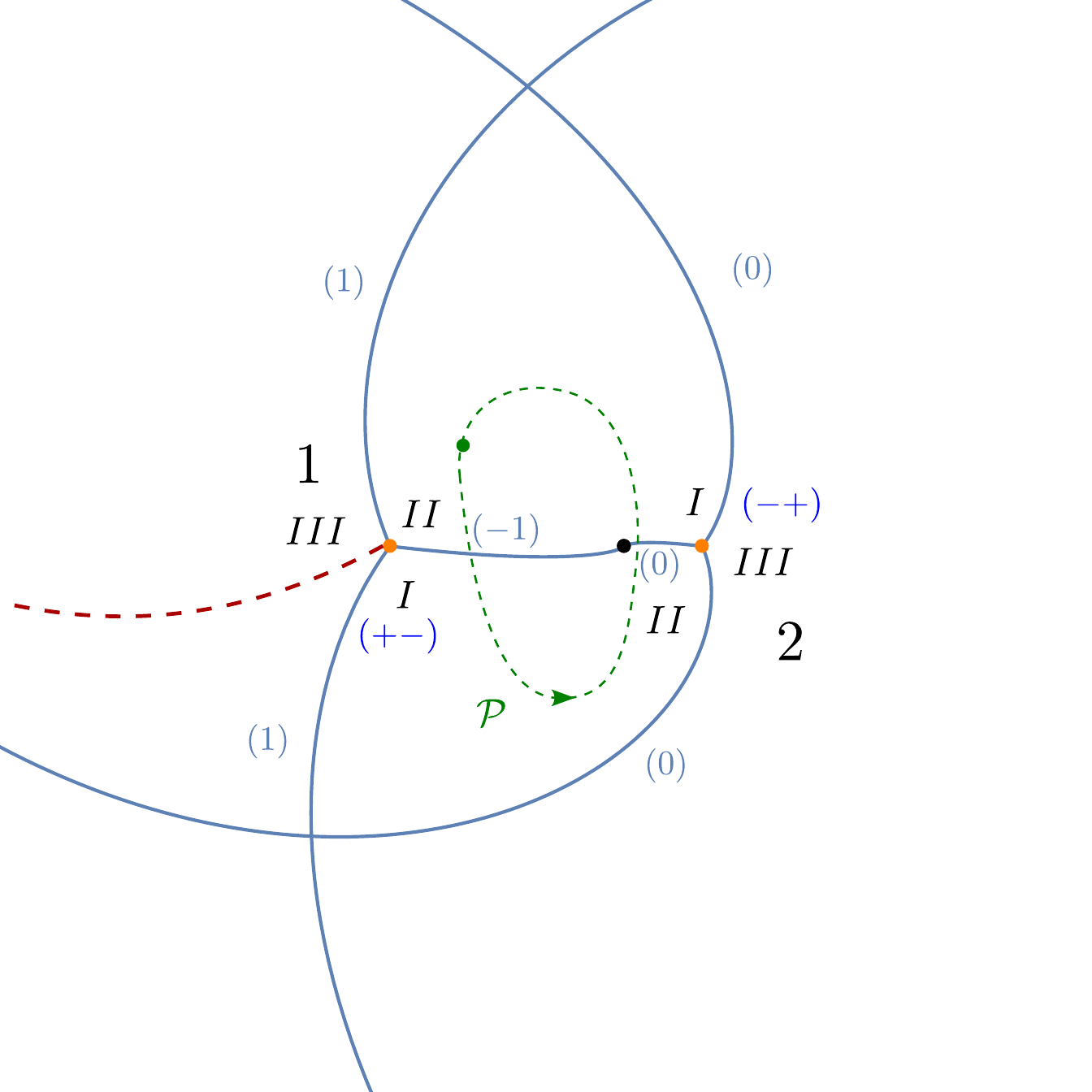}
\caption{The Stokes graph of \eqref{eq:q-Bessel} for $\vartheta = \frac{4}{9}\pi,\,\kappa=\frac{1}{2}$. Only primary Stokes lines (those that emanate directly from branch points) are shown. The label $(\ell)$ near each trajectory denotes the type of Stokes matrix $S^{(\ell)}$ associated to each Stokes line, as discussed in Section \ref{sec:stokes-data}. This can be obtained from a limit of the Stokes graph in Figure~\ref{fig:local-F0-stokes-primary} in a limit where both inner branch points collapse into the logarithmic puncture at the origin.}
\label{fig:q-Bessel-marked}
\end{center}
\end{figure}

\subsubsection{Monodromy}
We can compute the monodromy for a solution of the form \eqref{eq:BasisqAiry} around a loop surrounding the origin, as shown in Figure~\ref{fig:q-Bessel-marked}. 
Since the two branch points have opposite signature, the matrices between regions $(1,I)\to (2,II)$ and $(2,I)\to (1,II)$ are off-diagonal, and we parameterize them as follows 
\be
	\Uij_{\text{bottom}} 
	= \left(\begin{array}{cc}
	0 & i \CU_{1} \\
	i \CU_{1}^{-1} & 0  \\
	\end{array}\right)\,,
	\qquad
	\Uij_{\text{top}}
	= \left(\begin{array}{cc}
	0 & i \CU_{2} \\
	i \CU_{2}^{-1} & 0  \\
	\end{array}\right)\,.
\ee
We then compute the monodromy around the loop $\CP$ in Figure \ref{fig:q-Bessel-marked}:
\begin{equation}
    M_\CP=\left[S^{(-1)}\right]^{-1} \Uij_{\text{bottom}} \left[S^{(0)}\right]^{-1} \Uij_{\text{top}} =\left(
\begin{array}{cc}
 -\frac{1}{\CU_{1} \CU_{2}} & \quad 0 \\
 \frac{i \left(\xi  \CU_{1}^2+1\right)}{\xi  \CU_{1} \CU_{2}} & \quad -\CU_{1} \CU_{2} \\
\end{array}
\right),
\end{equation}
and the trace
\begin{shaded}
\be\label{eq:half-geometry-monodromy}
\begin{split}
	{\rm tr}\,M_{{\mathcal{P}}}  
	=- \frac{1}{\CU_{1} \CU_{2}} - \CU_{1} \CU_{2}
	\mathop{=}^{\eqref{eq:AiryKappaM}}	2\cos\left(\frac{2\pi i\mu}{\hbar} \right),
\end{split}
\ee
\end{shaded}
and we have the natural identification
\begin{equation}\label{eq:XtoMu}
\CU_{1}\CU_{2}=e^{\frac{2\pi i\mu}{\hbar}}
\end{equation}

\subsubsection{Comparison with String Theory, D-brane charges and degenerations}
Equation \eqref{eq:q-Bessel} is the quantum mirror curve to the resolution of the orbifold $\mathbb{C}\times\mathbb{C}^2/\mathbb{Z}_2$.
Recall from definition \ref{def:physical} that physical charges are represented by equivalence classes of anti-invariant cycles. A cycle $\gamma_*$ that winds counterclockwise around the puncture on sheet $(+,n_1)$ and clockwise around the puncture on sheet $(-,n_2)$ has therefore (quantum) period
\be\label{eq:half-geom-Z}
	\frac{2\pi }{\hbar} Z_{\gamma_*} =  -(V^{(+,n_1)}_{\gamma_*} - V^{(-,n_2)}_{\gamma_*}) = \frac{2\pi}{\hbar}\left(  -2i \mu +2\pi  \, (n_1-n_2)\right)\,.
\ee

In the context of geometric engineering, the WKB curve \eqref{eq:curve-half-geometry}
represents the moduli space of a 3d defect theory engineered by a toric brane \cite{Dimofte:2010tz}.
The 3d theory has an $SU(2)$ flavor symmetry, and $e^\mu$ is the associated fugacity. In type IIA string theory, $-2i\mu = Z_{D2}$ is identified with the central charge of the $D2$ brane supported on $\IP^1$, see e.g. \cite{Banerjee:2019apt}:
\be
	X_{D2} = e^{-\frac{2\pi }{\hbar} Z_{D2}}  = e^{\frac{2\pi i }{\hbar} \,{2\mu}}\,.
\ee
Given \eqref{eq:XtoMu}, and the physical identification of $\CU_{D2}$ with the $SU(2)$ flavor fugacity, we recover the character of $SU(2)$
\be\label{eq:trace-half-geom-D2}
	\frac{1}{\CU_1 \CU_2} + \CU_1 \CU_2 = 2\cos \left(\frac{2\pi \mu}{\hbar}\right) = X_{\frac{1}{2}\overline{D2}} + X_{\frac{1}{2}D2}   \,.
\ee
There are interesting limits of this geometry:
\begin{itemize}
\item Taking $\kappa\to 0$ reduces \eqref{eq:q-Bessel} to the $q$-Airy equation. The topology of the curve remains unchanged in this limit. In this case $\mu = \frac{\pi i}{2}$ and $Z_{D2} = \pi= \frac{1}{2}Z_{D0}$ as expected by orbifold symmetry.
\item Taking $\kappa\to 1$ takes one of the branch points to $x=0$, thereby changing the topology of the curve, which features a regular puncture of type \ref{item:square-root-puncture}. 
Since $\lim_{x\to 0}y_\pm\to 1$ in this case $\mu=0$, and the period \eqref{eq:half-geom-Z} vanishes with $n_1=n_2$.
\item Taking $\kappa\to -1$ takes the other branch point to $x=0$. Again the topology of the curve degenerates and there is a regular puncture of type \ref{item:square-root-puncture}. 
Since $\lim_{x\to 0}y_\pm\to -1$, the residues at punctures are given by $\mu=\pm \pi i$.
The period \eqref{eq:half-geom-Z} now vanishes with $n_2=n_1+1$.\footnote{Shifting the branch of $\log y_-$ to $\log y_-+2\pi i$ is consistent with having both sheets within the same (principal) branch of $\log y$, namely $\IR \times (-\pi i,\pi i]$.}
\end{itemize}

\subsection{A $q$-hypergeometric equation}
Consider the $q$-difference equation
\begin{equation}\label{eq:qDEConifold}
    \phi(qx)+\frac{x}{Q}\phi(q^{-1}x)=(1+q^{-1}x)\phi(x),
\end{equation}
closely related to the basic q-hypergeometric equation of type (1,1) (see for example Section 5.2 of \cite{Panfil:2018faz}). The linear system is
\begin{equation}
    \left(\begin{array}{c}
        \phi(x) \\
        \phi(qx)
    \end{array} \right)=A(x)\left( \begin{array}{c}
        \phi(q^{-1}x) \\
        \phi(x)
    \end{array} \right), \qquad A(x)=\left( \begin{array}{cc}
        0 & 1 \\
        -x/Q & 1+q^{-1}x
    \end{array} \right).
\end{equation}
Setting
\begin{equation}
    \phi(q^{-1}x)=g(x)\psi(x),\qquad g(qx)=\sqrt{\frac{x}{Q}}g(x),\quad\implies \quad g(x)=x^{\frac{1}{2}\frac{\log Q}{\log q}}e^{\frac{1}{4}\log x\left(\frac{\log x}{\log q}-1 \right)},
\end{equation}
we can recast this in the form
\be\label{eq:conifold-curve-q}
	\psi(qx)+\psi(q^{-1}x)
	=
	2Q^{-1}\left(q^{-1}\,x^{\frac{1}{2}}+x^{-\frac{1}{2}} \right)\psi(x)\,.
\ee
The corresponding WKB curve is
\be\label{eq:conifold-curve-class}
	y+y^{-1} = 2Q^{-1}\left(x^{\frac{1}{2}}+x^{-\frac{1}{2}} \right)\,.
\ee

Since $T(x)\in \mathbb{C}[[x^{\frac{1}{2}}]]$, it will be convenient to switch to a double covering of the $x$-plane by introducing a coordinate $w^2=x$. It then coincides with a special case $\kappa=0$ of the WKB curve for the q-modified Mathieu equation \eqref{eq:P1P1Mirror}, to be discussed in the next section.

The curve \eqref{eq:conifold-curve-class} has a manifest $\IZ_2$ symmetry (in addition to the symmetry $y\to y^{-1}$ that is common to all curves that we consider):
\be
	x\to x^{-1}\,.
\ee 
There are two square root branch points, located at
\be
	x_\pm=\frac{1-2 Q^2\pm\sqrt{1-4 Q^2}}{2 Q^2},
\ee
as well as logarithmic singularities at $x=0$ and at $x=\infty$.
On the $w$-plane the number of branch points doubles up, and the symmetry is now $\IZ_2\times \IZ_2$ generated by\footnote{The change of sign requires a corresponding $y\to -y$, which is left understood.}
\be
	w\to w^{-1}\,,\qquad w\to -w\,.
\ee 
The Stokes graph is shown in Figure \ref{fig:conifold-v2}, where branch points are labeled by $a=1,2,3,4$ and the regions near each branch point by $I,II,III$.\footnote{We choose to break the rotational $\IZ_2$ symmetry by a choice of trivialization that involves logarithmic cuts only along the negative real axis. However, since quantum periods symbols are independent of the choice of trivialization, they will respect the full $\IZ_2\times \IZ_2$ symmetry.}

\begin{figure}[h!]
\begin{center}
\includegraphics[width=0.75\textwidth]{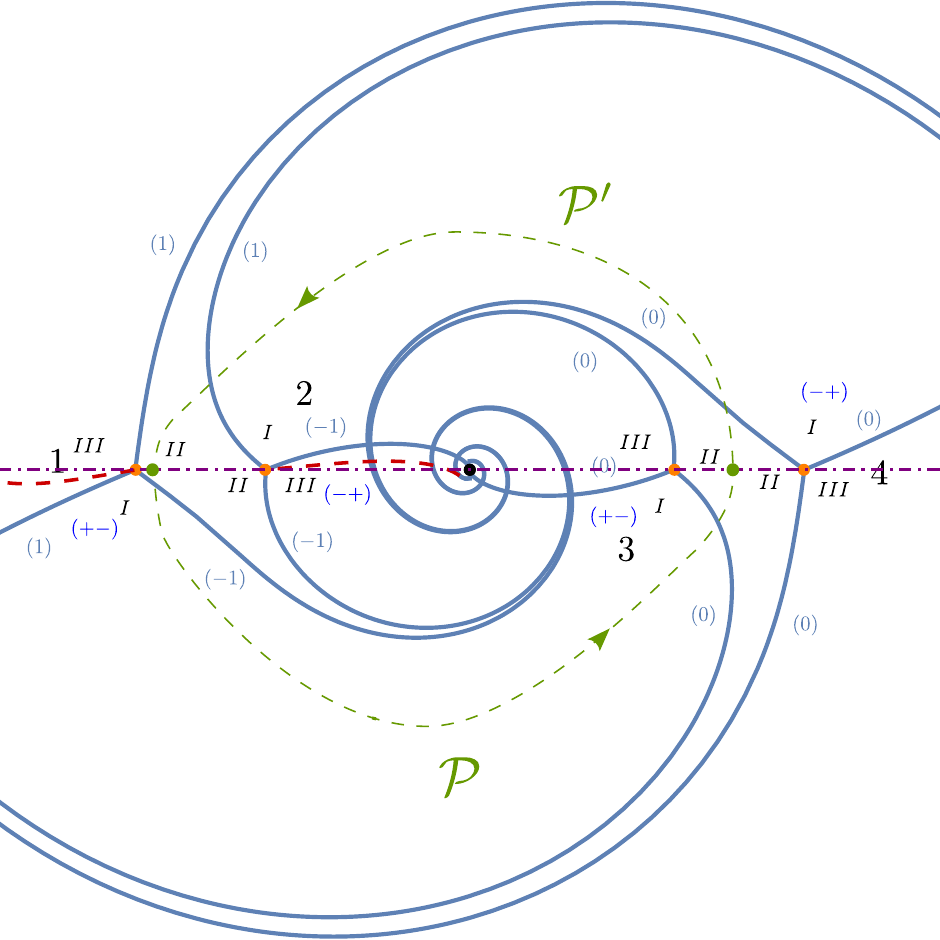}
\caption{The Stokes graph of \eqref{eq:conifold-curve-q} for $\vartheta = 0.6,\,\tau =200/97$, shown on the double-covering of the $x$-plane. Only primary Stokes lines (those that emanate directly from branch points) are shown. The label $(\ell)$ near each trajectory denotes the type of Stokes matrix $S^{(\ell)}$ associated to each Stokes line, as discussed in Section \ref{sec:stokes-data}. This can be obtained from a limit of the Stokes graph in Figure \ref{fig:local-F0-stokes-primary} by setting $\kappa=0$ and $\tau =0.97$.}
\label{fig:conifold-v2}
\end{center}
\end{figure}

The normalised solutions near each branch point are related by Stokes matrices as follows:
\be
\begin{split}
	&
	\Psi_{1}^{II} = \Psi_{1}^{I}\, S^{(-1)}(\xi_1)\,,
	\qquad
	\Psi_{1}^{III,\uparrow} = \Psi_{1}^{II}\, S^{(1)}(\xi_1)\,,
	\qquad
	\Psi_{1}^{I} = \Psi_{1}^{III,\downarrow}\, S^{(1)}(\xi_1)
	\\
	&
	\Psi_{2}^{II} = \Psi_{2}^{I}\, S^{(1)}(\xi_2)\,,
	\qquad
	\Psi_{2}^{III,\downarrow} = \Psi_{2}^{II}\, S^{(-1)}(\xi_2)\,,
	\qquad
	\Psi_{2}^{I} = \Psi_{2}^{III,\uparrow}\, S^{(-1)}(\xi_2)
	\\
\end{split}
\ee
\begin{equation}
    \Psi_3^{\star+I}=\Psi_3^{\star}S^{(0)},\qquad \Psi_4^{\star+I}=\Psi_4^{\star}S^{(0)},\qquad \star=I,II,III\mod{3} .
\end{equation}
Recall from \eqref{eq:log-sqrt-monodromy-constraint} that the Stokes matrices $S^{(\mp1)}$ for the Stokes line emanating from branch points $1$ and $2$ 
depend on the parameters
\begin{equation}\label{eq:xi-k-conifold}
    \xi_k=\left(\frac{x}{x_k} \right)^{\frac{2\pi i}{\hbar}}
\end{equation} 
Globally, nearby regions of pairs of distinct branch points are identified as follows
\be
	R_{1,II} = R_{2,II}
	\qquad
	R_{1,I} = R_{3,I}
	\qquad
	R_{2,I} = R_{4,I}
	\qquad
	R_{3,II} = R_{4,II}\,.
\ee
The Borel sum of formal WKB series within each region can be analytically continued from one branch point to the other without discontinuities or logarithmic shifts. 
Recall from Section \ref{sec:qFG} that the change of normalizations correspond to parallel transports is encoded by a transition matrix which is either anti-diagonal or diagonal, depending on whether or not there is an exchange of the dominant and subdominant solutions.
We thus have
\begin{equation}
    \Psi_1^{II}=\Psi_2^{II} \Uij_{21},\qquad \Psi_1^{I}=\Psi_3^{I} \Uij_{31},\qquad \Psi_4^{I}=\Psi_2^{I}X_{24},\qquad\Psi_3^{II}=\Psi_4^{II}\Uij_{43},
\end{equation}
with
\be\begin{split}
	\Uij_{21} =\left(\begin{array}{cc}
		0 & i \CX_A \\
		i \CX_A^{-1} & 0
	\end{array} \right)
	=\Uij_{43}
	\,,
	\qquad
	\Uij_{31} =\left(\begin{array}{cc}
		\CX_{B} &  0 \\
		0 & \CX_{B}^{-1}
	\end{array} \right)
	=\Uij_{42}\,,
\end{split}
\ee
where the orbifold symmetry $w\to -w$ has been taken into account. 

\subsubsection{Monodromy and half-monodromy}

We compute the monodromy of WKB solutions \eqref{eq:psipmnleading} around the path obtained by concatenation of $\mathcal{P}$ and $\mathcal{P'}$ as shown in Figure \ref{fig:conifold-v2} 
\be\label{eq:conifold-monodromy-matrix}
\begin{split}
M_{{\mathcal{P}\mathcal{P'}}} 
	&= 
	[S^{(-1)}(\xi_1)]^{-1} \Uij_{31}^{-1} S^{(0)} \Uij_{43}^{-1} [S^{(0)}]^{-1} \Uij_{42} S^{(1)}(\xi_2) \Uij_{21}
	\\
	&=\
	\left(
\begin{array}{cc}
	{\CX_B}^2 
	& 
	-i{\CX_A}^2 \left(1-\xi_2 {\CX_B}^2\right) 
	\\
	i(1-\xi_1^{-1}{{\CX_B}^2})
	 & 
	 \quad
	\CX_B^{-2} (1 + \CX_A^2(1-\xi _2 \CX_B^{2})(1-\xi _1^{-1} \CX_B^{2}))
\end{array}
\right)
\end{split}
\ee 
Its trace is
\begin{equation}
	\tr M_{\CP\CP'}=
	\CX_B^2+\frac{1}{\CX_B^2}+\frac{\CX_A^2}{\CX_B^2}
	-\xi _2 \CX_A^2-\xi_1^{-1}{\CX_A^2}+{\xi_1^{-1}\xi _2 \CX_A^2 \CX_B^2}
\end{equation}
Equation \eqref{eq:conifold-monodromy-matrix} displays a novel feature of this example: 
the trace of the monodromy is not constant, but rather it is a $q$-periodic function. 
This is perfectly natural from the point of view of character varieties of $q$-difference equations \cite{Ohyama2020,Joshi2023,Ramis2023}, and in our present setting follows because the origin is a logarithmic, rather than a regular, singularity. 
Working on the field of $q$-periodic functions one would keep all terms, but 
as discussed in Section \ref{sec:qMonodromy} when working over $\mathbb{C}$ we should take a regularised trace of the corresponding infinite-dimensional monodromy matrix:
\begin{shaded}
\begin{equation}\label{eq:conifold-mondromy}
	\tr \mathbb{M}_{\CP\CP'}=
	\CX_B^2+\frac{1}{\CX_B^2}+\frac{\CX_A^2}{\CX_B^2}\,.
\end{equation}
\end{shaded}

Motivated by the physical interpretation of this Stokes graph as an exponential network, we also consider the half-monodromy of the WKB solutions \eqref{eq:psipmnleading}. This corresponds to winding around the puncture at the origin in $\IC^*_x$, since $x\to  e^{2\pi i} x$ corresponds to $w\to e^{\pi i}\, w$. 
Let us fix a choice of half plane, and consider the path ${\mathcal{P}}$ shown in Figure \ref{fig:local-F0-stokes-primary}. The initial point is along the negative real axis, while the final point is along the positive real axis. This gives the transport matrix 
\be\label{eq:conifold-monodromy-matrix}
\begin{split}
M_{{\mathcal{P}}} 
	&= 
	[S^{(-1)}(\xi_1)]^{-1} \Uij_{31}^{-1} [S^{(0)}]^{-1} \Uij_{43}^{-1} 
	=\
	\left(
\begin{array}{cc}
 0 & -i \CX_A \CX_B \\
 -\frac{i}{\CX_A \CX_B} & \frac{\CX_A}{\CX_B}(1-\xi_1^{-1}{\CX_B^2}) \\
\end{array}
\right)
\,.
\end{split}
\ee 
This matrix depends on the choice of half-plane that we made, and any other choice would give a transport related to this one by conjugation $U_1^{-1}M_{{\mathcal{P}}}U_2$.
If the $\IZ_2$ symmetry $w\to -w$ was fully preserved we would have $U_1=U_2$ and the trace would be invariant under rotations.
However we are explicitly breaking this symmetry with the choice of trivialization, and therefore it matters which half-plane we choose. \footnote{
The choice that we made has the property that a rotation by $\pi$ of the half-plane would give a matrix with the same trace (but not rotations by other phases). Relations to other choice of half-planes can be obtained by careful analysis of how the basis of solutions changes due to the logarithmic cuts, and amounts to including certain factors of $\xi$ in the conjugation matrices.}
The trace is
\begin{equation}
	\tr M_\CP= \frac{\CX_A}{\CX_B}-\frac{\CX_A \CX_B}{\xi _1}\,.
\end{equation}
Working over $\IC$ this reduces to the first term only
\begin{shaded}
\begin{equation}\label{eq:conifold-half-mondromy}
	\tr \mathbb{M}_\CP= \frac{\CX_A}{\CX_B}\,.
\end{equation}
\end{shaded}

\subsubsection{Comparison with String Theory and D-brane charges}
The WKB curve of the original equation \eqref{eq:qDEConifold} coincides with the curve of the resolved conifold in quadratic framing, namely $y'+x' / (Q' y')=1+x'$ as studied e.g. in \cite{Banerjee:2019apt}. The WKB for the equation in involutive form \eqref{eq:conifold-curve-class} is related to this one by identifications $x'=x, y' = \frac{1}{2}Q x^{1/2} \,y$.
Taken in involutive form as in \eqref{eq:conifold-curve-class}, the curve can also be regarded as the orbifold limit of the mirror curve of local $\IF_0$, studied below.\footnote{This is \eqref{eq:P1P1Mirror} with $\kappa=0,\tau = 2 Q^{-1}$ and $x_{\IF_0} = w_{\rm{conifold}}$.} 

Following the ideas of Section \ref{sec:qMonodromy}, we write the infinite monodromy matrix for WKB solutions defined over $\mathbb{C}$:
\be
\begin{split}
\mathbb{M}_{{\mathcal{P}}}  
& =
\left(
\begin{array}{cc}
 0 & -i \CX_A \CX_B \\
 -\frac{i}{\CX_A \CX_B} & \frac{\CX_A}{\CX_B} \\
\end{array}
\right)
\delta_{m,n}
+ 
\left(
\begin{array}{cc}
 0 & 0 \\
0 &  - \xi_1^{-1} \CX_A\CX_B\\
\end{array}
\right)    \delta_{m,n-1}
\end{split}
\ee
The (regularised) trace of the infinite-dimensional matrix retains only the $\xi$-independent terms. This amounts to working equivariantly with respect to the $\mathbb{Z}$-cover $\tSigma\rightarrow\Sigma$:
\be\label{eq:conifold-monodromy-final}
\begin{split}
	{\tr}\, \mathbb{M}_{{\mathcal{P}}}  
	&=\frac{\CX_A}{\CX_B}\,.
\end{split}
\ee
To identify these coordinates with physical D-brane charges in type IIA string theory, we take advantage of the identification of the curve \eqref{eq:conifold-curve-q} with the orbifold limit of that of local $\IF_0$, given in \eqref{eq:P1P1Mirror}. The identification involves setting $\tau = 2 Q^{-1}$ and $\kappa=0$, and it will be convenient to recall from \cite[eq. (2.10)]{DML2022} 
that $X_{D2_f\-\overline{D2}_b} = e^{-\frac{4\pi i}{\hbar} \log\tau}$.
It then follows immediately from the identifications \eqref{eq:cluster-vs-voros-F0} and \eqref{eq:D-brane-charges-gamma-F0} that we have
\be
\begin{split}
	\frac{\CX_B^4}{\CX_A^2} = X_{\gamma_1} = X_{\gamma_3} = X_{\frac{1}{2}(\gamma_1+\gamma_3)} =  X_{\frac{1}{2}(D2_f\-\overline{D2}_b)}&= e^{-\frac{2\pi i}{\hbar} \log\tau}
	\\
	\frac{1}{\CX_A^2} = X_{\gamma_2} = X_{\gamma_4} = X_{\frac{1}{2}(\gamma_2+\gamma_4)} = X_{\frac{1}{2}(\overline{D0}\-D2_f\-\overline{D2}_b)} &= e^{2\pi^2-\frac{2\pi i}{\hbar} \log\tau} 
\end{split}
\ee
Therefore \eqref{eq:conifold-monodromy-final} translates to
\be
	\frac{\CX_A}{\CX_B} = (X_{\gamma_1+\gamma_2})^{1/4}
	= e^{\frac{\pi^2}{2}-\frac{\pi i}{\hbar} \log\tau} 
\ee
Noting that in this limit we have $X_{D2_f} = \CX_{D0}^{-1/2} = e^{\frac{2\pi^2}{\hbar}}$, it follows that $X_{D2_b} = e^{\frac{2\pi^2}{\hbar}+\frac{4\pi i}{\hbar} \log\tau}$ this may also be written as
\be
	\frac{\CX_A}{\CX_B} = (X_{D0\-\overline{D2}_b})^{1/4} \,.
\ee

\subsection{The $q$-Mathieu equation}\label{sec:local-F0}

All the examples we considered up to now are degenerations of the equation
\begin{equation}\label{eq:qP1P1}
\psi(qx)+\psi(q^{-1}x)=\left[-\tau(x+x^{-1})+\kappa \right]\psi(x),
\end{equation}
or in matrix form
\begin{equation}\label{eq:qConnF0}
    \Psi(qx)=A(x)\Psi(x),\qquad A(x)=\left( \begin{array}{cc}
        0 & 1 \\
        -1 & \kappa-\tau\left(qx+q^{-1}x^{-1} \right)
    \end{array} \right),
\end{equation}
whose differential limit is the modified Mathieu equation. Its WKB curve is
\be\label{eq:P1P1Mirror}
	y+y^{-1} + \tau(x+x^{-1}) - \kappa = 0.
\ee 
In suitable limits, this reduces to WKB curves studied so far:
\begin{itemize}
\item Rescaling $x \to -2 x/\tau$ and $\kappa\to 2\kappa$, and taking $\tau\to\infty$ reduces to
\be
	y+y^{-1} = 2( x + \kappa)
\ee
which is the curve \eqref{eq:curve-half-geometry} of the q-Airy$_\kappa$ equation.
\item With the substitutions $x\to x^{1/2}$, $\tau\to -2Q^{-1}$ and taking the limit $\kappa\to 0$, this reduces to the hypergeometric WKB curve \eqref{eq:conifold-curve-class}.
\end{itemize}

\begin{figure}[h!]
\begin{center}
\includegraphics[width=0.75\textwidth]{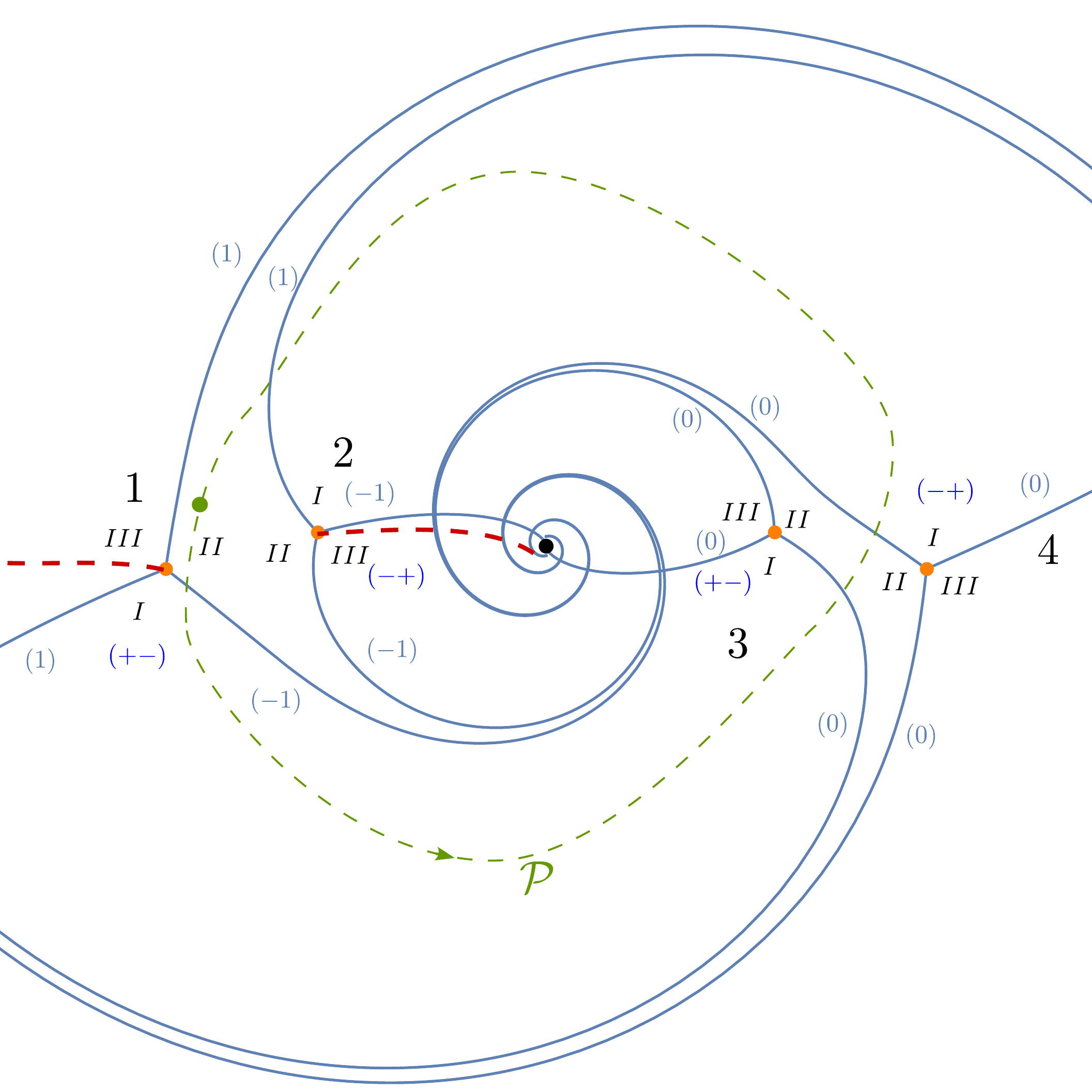}
\caption{
The Stokes graph of local $\mathbb{F}_0$ for $\vartheta = 0.6,\,\kappa=0.03 i,\,\tau=0.97$. Only primary Stokes lines (those that emanate directly from branch points) are shown. The label $(\ell)$ near each trajectory denotes the type of Stokes matrix $S^{(\ell)}$ associated to each Stokes line, as discussed in Section \ref{sec:stokes-data}. The assignment of these labels follows from a detailed analysis of the basis $\Psi$ based on global trivialization data. What emerges is that the two branch points encircled by logarithmic cuts have opposite signs for their local bases of solutions.
}
\label{fig:local-F0-stokes-primary}
\end{center}
\end{figure}

The curve has four quadratic branch points. We show the WKB Stokes graph and Stokes sectors for $\kappa=0.03i,\,\tau=0.97$ and $\vartheta=0.6$ in  Figure \ref{fig:local-F0-stokes-primary}.
As usual, square-root cuts lie along primary Stokes lines, see Section \ref{sec:stokes-data}. 
Labeling the branch points by $a=1,2,3,4$, and the regions near each branch point by $I,II,III$ as in Figure \ref{fig:local-F0-stokes-primary}, 
the normalised solutions near each branch point are related by Stokes matrices as follows:
\be
\begin{split}
	&
	\Psi_{1}^{II} = \Psi_{1}^{I}\, S^{(-1)}(\xi_1)\,,
	\qquad
	\Psi_{1}^{III,\uparrow} = \Psi_{1}^{II}\, S^{(1)}(\xi_1)\,,
	\qquad
	\Psi_{1}^{I} = \Psi_{1}^{III,\downarrow}\, S^{(1)}(\xi_1)
	\\
	&
	\Psi_{2}^{II} = \Psi_{2}^{I}\, S^{(1)}(\xi_2)\,,
	\qquad
	\Psi_{2}^{III,\downarrow} = \Psi_{2}^{II}\, S^{(-1)}(\xi_2)\,,
	\qquad
	\Psi_{2}^{I} = \Psi_{2}^{III,\uparrow}\, S^{(-1)}(\xi_2)
	\\
\end{split}
\ee
\begin{equation}
    \Psi_3^{\star+I}=\Psi_3^{\star}S^{(0)},\qquad \Psi_4^{\star+I}=\Psi_4^{\star}S^{(0)},\qquad \star=I,II,III\mod{3} .
\end{equation}
Globally, nearby regions of pairs of distinct branch points are identified as follows
\be
	R_{1,II} = R_{2,II}
	\qquad
	R_{1,I} = R_{3,I}
	\qquad
	R_{2,I} = R_{4,I}
	\qquad
	R_{3,II} = R_{4,II}\,.
\ee
The Borel sum of formal WKB series within each region can be analytically continued from one branch point to the other without discontinuities or logarithmic shifts. 
Recall from Section \ref{sec:qFG} that the change of normalisations corresponding to parallel transports is encoded by a transition matrix which is either anti-diagonal or diagonal, depending on whether or not there is an exchange of the dominant and subdominant solutions.
We thus have
\begin{equation}
    \Psi_1^{II}=\Psi_2^{II} \Uij_{21},\qquad \Psi_1^{I}=\Psi_3^{I} \Uij_{31},\qquad \Psi_4^{I}=\Psi_2^{I}X_{24},\qquad\Psi_3^{II}=\Psi_4^{II}\Uij_{43},
\end{equation}
with
\be\begin{split}
	\Uij_{21} =\left(\begin{array}{cc}
		0 & i \CU_{21} \\
		i \CU_{21}^{-1} & 0
	\end{array} \right)
	\,,
	\qquad
	\Uij_{31} =\left(\begin{array}{cc}
		\CU_{31} &  0 \\
		0 & \CU_{31}^{-1}
	\end{array} \right)
	\,,
	\\
	\Uij_{42} =\left(\begin{array}{cc}
		\CU_{42} & 0 \\
		0 & \CU_{42}^{-1} 
	\end{array} \right)
	\,,
	\qquad
	\Uij_{43} =\left(\begin{array}{cc}
		0 & i\CU_{43} \\
		i\CU_{43}^{-1} & 0
	\end{array} \right)
	\,.
\end{split}
\ee

\subsubsection{Monodromy}

The monodromy of the WKB solutions \eqref{eq:psipmnleading} around the puncture at the origin in $\IC^*_x$, (${\mathcal{P}}$ in Figure \ref{fig:local-F0-stokes-primary}) with base-point in region $II$ of branch point $1$, is computed by the following: 
\be\label{eq:localF0-monodromy-matrix}
\begin{split}
	M_{{\mathcal{P}}}
	&= 
	[S^{(-1)}(\xi_1)]^{-1} \Uij_{31}^{-1} S^{(0)} \Uij_{43}^{-1} [S^{(0)}]^{-1} \Uij_{42} S^{(1)}(\xi_2) \Uij_{21}
	\\
	&=
	\left(
	\begin{array}{cc}
	 \frac{\CU_{31} \CU_{42} \CU_{43}}{\CU_{21}} 
	 & \quad
	 -i\frac{\CU_{21} \CU_{31} \CU_{43}}{\CU_{42}}
	 (1-\xi_2\, \CU_{42}^2)
	 \\
	i \frac{\CU_{42} \CU_{43}}{\CU_{21}\CU_{31}}(1-\xi_1^{-1}\CU_{31}^2)
	 & \quad
	 \frac{\CU_{21}}{\CU_{31}\CU_{42}\CU_{43}}
	 \left( 1+ \CU_{43}^2(1-\xi_1^{-1} \CU_{31}^2)(1-\xi_2\CU_{42}^2) \right)
	 \\
\end{array}
\right).
\end{split}
\ee 
Recall from \eqref{eq:log-sqrt-monodromy-constraint} that the Stokes matrices $S^{(\mp1)}$ for the Stokes line emanating from branch points $1$ and $2$ 
depend on the parameters
\begin{equation}\label{eq:xi-k-F0}
    \xi_k=\left(\frac{x}{x_k} \right)^{\frac{2\pi i}{\hbar}}
\end{equation} 
where $x_k$ is the position of the $k$-th branch point, and $x$ is the position of the basepoint for the loop $\mathcal{P}$.
\footnote{In \eqref{eq:localF0-monodromy-matrix} all $\xi_k$ are actually evaluated at $x\,e^{2\pi i}$ because they are always evaluated at the same point as the solution $\Psi$.}
In our conventions
\be\label{eq:x1-x2-F0}
	x_1 = \frac{\kappa -2-\sqrt{(\kappa-2)^2-4 \tau ^2}}{2 \tau }\,,\qquad
	x_2 = \frac{\kappa -2+\sqrt{(\kappa-2)^2 -4 \tau ^2}}{2 \tau }
\ee

The trace of the monodromy around the origin is
\begin{equation}
\begin{split}
    \tr M_{\CP} &=\CU_{31} \CU_{42} \CU_{43} \CU_{21} 
	\frac{\xi_2}{\xi_1}
    +\frac{\CU_{43} \CU_{21}}{\CU_{31} \CU_{42}}+\frac{\CU_{21}}{\CU_{31} \CU_{42} \CU_{43}}+\frac{\CU_{31} \CU_{42} \CU_{43}}{\CU_{21}} \\
    &
    -\frac{ \CU_{42} \CU_{43} \CU_{21} }{\CU_{31}}\xi_2-\frac{\CU_{31} \CU_{43} \CU_{21}}{\CU_{42}}\frac{1}{\xi_1}
\end{split}
\end{equation}
Terms in the second line depend on a choice of basepoint $x$ for the loop $\mathcal{P}$ through \eqref{eq:xi-k-F0}, while those in the first line are independent of the basepoint.
Working on the field of $q$-periodic functions one would keep all terms, but when working over $\IC$ gauge invariance only holds for terms in the first line
\begin{shaded}
\be\label{eq:localF0-monodromy-final}
\begin{split}
    \tr \mathbb{M}_{\CP} &=\CU_{31} \CU_{42} \CU_{43} \CU_{21} 
	\left(\frac{x_1}{x_2} \right)^{\frac{2\pi i}{\hbar}}
    +\frac{\CU_{43} \CU_{21}}{\CU_{31} \CU_{42}}+\frac{\CU_{21}}{\CU_{31} \CU_{42} \CU_{43}}+\frac{\CU_{31} \CU_{42} \CU_{43}}{\CU_{21}} \\
\end{split}
\ee
\end{shaded}
with $x_1, x_2$ explicit functions of moduli $\kappa,\tau$ given in \eqref{eq:x1-x2-F0}.

\subsubsection{Exponential networks and cluster coordinates}
Like in the previous hypergeometric example, the trace of the monodromy is not constant, but rather it is a $q$-periodic function. 
The matrix form \eqref{eq:qConnF0} has nilpotent leading behaviour at the origin.

It is however interesting to relate our Voros symbols $\CX_{ij}$ with cluster coordinates arising in the framework of exponential networks and quivers. 
As discussed in Section \ref{sec:qMonodromy}, to do this we have to consider an infinite-dimensional space of solutions over $\mathbb{C}$, so that the monodromy matrix takes the form  
\be
\begin{split}
\mathbb{M}_{{\mathcal{P}}}  
& =
\left(
	\begin{array}{cc}
	 \frac{\CU_{31} \CU_{42} \CU_{43}}{\CU_{21}} 
	 & \quad
	 -\frac{i \CU_{21} \CU_{31} \CU_{43}}{\CU_{42}}
	 \\
	i \frac{\CU_{42} \CU_{43}}{\CU_{21}\CU_{31}}
	 & \quad
	 \frac{\CU_{21}}{\CU_{31}\CU_{42}\CU_{43}}
	 \left( 1+ \CU_{43}^2(1+\frac{\xi_2}{\xi_1} \CU_{31}^2\CU_{42}^2) \right)
\end{array}
\right)
\delta_{m,n}
\\
&
+
\left(
	\begin{array}{cc}
	0
	 & \quad
	i \xi_2  \, \CU_{21} \CU_{31} \CU_{43} \CU_{42}
	 \\
	0
	 & \quad
	- \xi_2 \frac{\CU_{21}\CU_{42}\CU_{43}}{\CU_{31}}
\end{array}
\right)   
\delta_{m,n+1}
+ 
\left(
	\begin{array}{cc}
	0
	 & \quad
	 0
	 \\
	- i \xi_1^{-1}\frac{i\CU_{42} \CU_{43}\CU_{31}}{\CU_{21}}
	 & \quad
	 -\xi_1^{-1}
	 \frac{\CU_{21}\CU_{31}\CU_{43}}{\CU_{42}}
\end{array}
\right)  
\delta_{m,n-1}
\end{split}
\ee
The (regularised) trace of the infinite-dimensional matrix retains only the $\xi$-independent terms. This amounts to working equivariantly with respect to the $\mathbb{Z}$-cover $\tSigma\rightarrow\Sigma$, and gives exactly \eqref{eq:localF0-monodromy-final}.

In the parametrisation coming from exponential networks, the lattice $\Gamma$ of physical cycles as introduced in Definition \ref{def:physical} was introduced in \cite{Banerjee:2020moh} from a basis of vanishing cycles on $\Sigma$ corresponding to the (mirror) of an exceptional collection for local $\mathbb{F}_0$. To each element of the basis is associated a saddle of the Stokes diagram,  shown in Figure \ref{fig:localF0-basis}, and a cluster variables $X_\gamma$ is labeled by lattice element $\gamma$. The monodromy in terms of cluster variables is given by~\eqref{eq:ExpNetTr} and~\eqref{eq:ExpNetGamma}: 
\be\label{eq:localF0-monodromy-gamma}
	\tr\,\mathbb{M}_{{\mathcal{P}}}=
	X_{\frac{1}{2}(\gamma_3+\gamma_4)}+X_{\frac{1}{2}(-\gamma_3-\gamma_4)}
	+X_{\frac{1}{2}(-\gamma_1+\gamma_2+\gamma_{D0})}+X_{\frac{1}{2}(\gamma_1-\gamma_2-\gamma_{D0})+(\gamma_2 + \gamma_4)}.
\ee
where 
\be\label{eq:total-charge-F0}
	\gamma_{D0} = \gamma_1+\gamma_2+\gamma_3+\gamma_4
\ee 
is the cycle corresponding to the (mirror) D0 brane charge.
One gets the following identification: 
\be\label{eq:cluster-vs-voros-F0}
\begin{split}
	X_{\gamma_1} & = 
	\frac{X_{\gamma_{D0}}}{\CU_{31}^2\CU_{42}^2\CU_{43}^2}\left(\frac{x_1}{x_2} \right)^{\frac{2\pi i}{\hbar}} \,,
	\\
	X_{\gamma_3} &= \frac{\CU_{31}^2\CU_{42}^2}{\CU_{21}^2}\,,	
\end{split}
\qquad
\begin{split}
	X_{\gamma_2} & = {\CU_{21}^2 }\left(\frac{x_1}{x_2} \right)^{-\frac{2\pi i}{\hbar}}\,,
	\\
	X_{\gamma_4} &= {\CU_{43}^2}\,,
\end{split}
\ee
so that indeed 
\be\label{eq:D0-in-F0}
	X_{\gamma_1+\gamma_2+\gamma_3+\gamma_4} = X_{\gamma_{D0}}\,.
\ee

An alternative, and more systematic way to obtain these identifications is as follows. The cluster coordinates $X_\gamma$ correspond to integrals of the quantum differential along cycles on $\Sigma$ that correspond to lifts of saddles connecting pairs of branch points on $\IC^*_x$.
The corresponding periods therefore corresponds to diagonal lifts of the transport between pairs of branch points. The diagonal lift is obtained by replacing Stokes matrices $S^{(\ell)}$ with just the branch cut matrix $\beta = \left(\begin{array}{cc}0&i\\i&0\end{array}\right)$.
\begin{itemize}
\item
To compute $X_{\gamma_4}$ we consider the path corresponding to the projection of $\gamma_4$ in Figure \ref{fig:localF0-basis}, which is a short path between branch points $3$ and $4$. The cycle closes up if the path is lifted to sheets $(+,n)$ and $(-,n+\Delta n)$ with $\Delta n=0$. Therefore the quantum period is
\be
	X_{\gamma_4} = \exp\left(\frac{1}{\hbar}\int_{x_3}^{x_4} (\log y_+ - \log y_-) \frac{dx}{x} + \dots\right) \equiv \CU_{43}^{2}
\ee
\item
For $\gamma_3$ we have a similar computation for a path connecting again $x_3$ and $x_4$, but passing between branch points $x_1$ and $x_2$
\be
	X_{\gamma_3} = \exp\left(\frac{1}{\hbar}\int_{x_4}^{x_3} (\log y_+ - \log y_-) \frac{dx}{x} + \dots\right) 
\ee
which should be identified with the diagonal transport
\be
\left(\begin{array}{cc}
	X_{\gamma_3}^{1/2} & 0\\
	0 & X_{\gamma_3}^{-1/2}
	\end{array}\right)
	=
	\Uij_{42}\cdot \beta\cdot \Uij_{21}\cdot \beta\cdot [\Uij_{31}]^{-1} \cdot \beta= 
	\left(\begin{array}{cc}
	\frac{\CU_{31}\CU_{42}}{\CU_{21}} & 0\\
	0 & \frac{\CU_{21}}{\CU_{31}\CU_{42}}
\end{array}\right)\,.
\ee

\item
For $\gamma_2$ the computation is again similar to that of $\gamma_4$ but with an important twist.
The cycle closes up at the branch points if the logarithmic shift of the lifts of $\log y_+$ and $\log y_-$ is zero $\Delta n=0$.
But in-between the branch points the integral features $\Delta n=-1$ because both branch points are surrounded by a logarithmic cut. Therefore
\be
	X_{\gamma_2} = \exp\left(\frac{1}{\hbar}\int_{x_2}^{x_1} (\log y_+ - \log y_- - 2\pi i ) \frac{dx}{x} + \dots\right) \equiv  \left(\frac{x_1}{x_2}\right)^{-\frac{2\pi i}{\hbar}} \CU_{21}^{2}\,,
\ee
where we used $\CU_{21}$ as the parallel transport without logarithmic shifts, consistently with our previous definitions.

\item
Finally, the most complicated case if that of $\gamma_1$. 
The relevant integration path is the projection of $\gamma_1$ in Figure \ref{fig:localF0-basis} is a path that connects branch points $1$ and $2$ while passing between $3$ and $4$. 
This combines the features of $\gamma_2$ (the logarithmic shift) with those of $\gamma_3$ (transport around $\IC^*$. Combining these features brings up a new phenomenon: the log-shift now picks up an additional overall factor $x_1/x_2\to x_1/x_2\, e^{2\pi i}$ due to the winding on $\IC^*_x$. Therefore
\be
\begin{split}
	X_{\gamma_1} 
	& = \exp\left(\frac{1}{\hbar}\int_{x_2}^{x_1} (\log y_+ - \log y_- + 2\pi i ) \frac{dx}{x} + \dots\right) 
	\\
	& =   \left(\frac{x_1}{x_2}\, e^{2\pi i}\right)^{\frac{2\pi i}{\hbar}} \,\frac{1}{\CU_{31}^2\CU_{43}^2\CU_{42}^2} 
	\\
	& = 
	\left(\frac{x_1}{x_2}\,\right)^{\frac{2\pi i}{\hbar}} \,\frac{X_{\gamma_{D0}}}{\CU_{31}^2\CU_{43}^2\CU_{42}^2}
\end{split}
\ee
where the monodromy without logarithmic contributions can be deduced by the diagonal transport matrix
\be
	\left(\begin{array}{cc}
	-X_{\gamma_1}^{1/2} & 0\\
	0 & -X_{\gamma_1}^{-1/2}
	\end{array}\right)
	=
	[\Uij_{31}]^{-1}\cdot \beta\cdot [\Uij_{43}]^{-1}\cdot \beta\cdot \Uij_{42} \cdot \beta= 
	\left(\begin{array}{cc}
	-\frac{1}{\CU_{31}\CU_{43}\CU_{42}} & 0\\
	0 & -\CU_{31}\CU_{43}\CU_{42}
	\end{array}\right)\,.
\ee
\end{itemize}

\begin{figure}[h!]
\begin{center}
\includegraphics[width=0.30\textwidth]{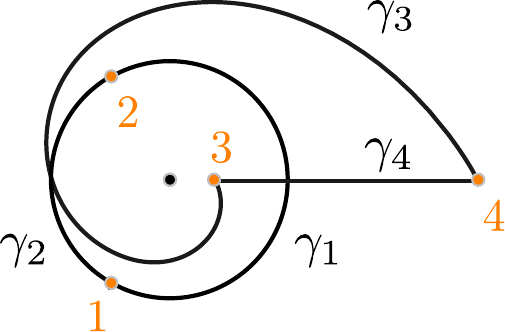}
\caption{Saddles for basis cycles for local $\IF_0$, at a point in moduli space studied in \cite{Banerjee:2020moh}. 
}
\label{fig:localF0-basis}
\end{center}
\end{figure}

\subsubsection{Comparison with String Theory, and D-brane charges}
Equation \eqref{eq:qP1P1} is the quantum mirror curve to local $\mathbb{P}^1\times\mathbb{P}^1$, geometrically engineering 5d pure $SU(2)$ Super Yang-Mills and related to the q-Painlev\'e $III_3$ equation \cite{Bonelli2017,Bershtein2017,Bershtein2018,DML2022}.
\begin{enumerate}
    \item Not all $\CU_{\gamma_i}$'s obtained from WKB approximation can be independent, since
    the $q$-difference equation has only three parameters: $\tau,\kappa,q$, so it would be impossible to have four independent cluster coordinates or Voros symbols.\footnote{As made explicit in \eqref{eq:D0-in-F0} their product (with appropriate signs) gives $X_{\gamma_{D0}} = e^{-\frac{4\pi^2}{\hbar}}$.
    However, this is not the issue since we are treating $q=e^\hbar$ as one of the variables. 
    In addition to this relation there is another one, which is not algebraic.} 
    A geometric consequence of this fact is that the space of periods of the curve \eqref{eq:P1P1Mirror} can't coincide with the moduli space of stability conditions for local $\mathbb{P}^1\times\mathbb{P}^1$, but only with a ``physical subspace''.\footnote{We thank T. Bridgeland for this comment.}
    
    \item The match between Voros symbols and Exponential networks allows us to make some further consideration of string theoretic nature. The identifications of D-brane charges with $\gamma_i$ from \cite{Banerjee:2020moh} adapted to our choice of moduli are\footnote{Compared to \cite{Banerjee:2020moh} we have shifted $D2_f\to D2_f\-{D0}$ and $D4\to D4\-D0$. 
    The shifts by $D0$ can be explained by the fact that (as explained in appendix) the two points in moduli space are related by a deformation that turns around a degeneration point for $\gamma_2,\gamma_1$, where a ``Seidel twist'' can occur. See \cite{Banerjee:2019apt} for a discussion of shifts of charges by Seidel twists in the context of exponential networks. This mapping of charges is compatible with~\eqref{eq:total-charge-F0}.}
\be\label{eq:D-brane-charges-gamma-F0}
	\gamma_1 : \ D0\-D4\,,
	\quad
	\gamma_2 : \ {D2}_f\-\overline{D4}\,,
	\quad
	\gamma_3 : \ \overline{D0}\-{D2}_b\-\overline{D2}_f\-\overline{D4}\,,
	\quad
	\gamma_4 : \ {D0}\-\overline{D2}_b\-{D4}\,.
\ee
Therefore the monodromy \eqref{eq:localF0-monodromy-gamma} can be rewritten as follows
\be\label{eq:F0-monodromy-D-branes}
	{\tr}\, \mathbb{M}_{{\mathcal{P}}}
	=
	X_{\frac{1}{2} \overline{D2}_f}
	+ X_{\frac{1}{2}D2_f}
	+ X_{\frac{1}{2} D2_f\- \overline{D4}}
	+ X_{D0\-\frac{1}{2} D2_f\-\overline{D2}_b\-{D4}}
\ee
Taking degeneration limits of local $\IF_0$ to conifold, to the half-geometry and to $\IC^3$ gives the monodromies found earlier, as summarized in Table \ref{tab:recap}.

\item
In the limit $\tau\to 0$ the WKB curve \eqref{eq:P1P1Mirror} reduces to the curve \eqref{eq:curve-half-geometry}, after a rescaling of $x$ by $\tau$.
In this limit branch points 2 and 3 fall into the puncture at $x=0$ and disappear. 
Due to this both $\CU_{21}, \CU_{43}\to 0$ but at the same time the `diagonal transports' (such as those used to identify $X_{\gamma_i}$) between branch points $1$ and $4$ remain finite. These are readily identified with $\CU_{21}\CU_{42}^{-1}$ for the path running above the origin, and with $\CU_{31}\CU_{43}$ for the path running below the origin. Therefore in \eqref{eq:localF0-monodromy-final} the terms $\CU_{31} \CU_{42} \CU_{43} \CU_{21} 
	\left(\frac{x_1}{x_2} \right)^{\frac{2\pi i}{\hbar}}
    +\frac{\CU_{43} \CU_{21}}{\CU_{31} \CU_{42}}$ vanish. On the other hand, the terms $\frac{\CU_{21}}{\CU_{31} \CU_{42} \CU_{43}}+\frac{\CU_{31} \CU_{42} \CU_{43}}{\CU_{21}}$ tend to \eqref{eq:trace-half-geom-D2} with the identification $\CU_1 = \CU_{31}\CU_{43}$ and $\CU_2 = \CU_{42} / \CU_{21}$. Identifying $D2$ in \eqref{eq:trace-half-geom-D2} with $D2_f$ also provides a consistency check on D-brane charge assignments in both models.

\item
In the 4d limit of the 5d $\CN=1$ Yang-Mills theory on $S^1\times \IR^4$, both $D0$ and $D2_b$ have central charges that grow infinite \cite{Banerjee:2020moh}. 
As a consequence the last term in \eqref{eq:localF0-monodromy-gamma} vanishes in the limit leaving only the first three terms. 
In terms of electric-magnetic charges $(m,e)$ of the low energy 4d $\CN=2$ Yang-Mills theory on the Coulomb branch, the $D2_f$ brane corresponds to the $W$-bosons and carries electric charge $(0,2)$, while the $D4$ brane corresponds to the magnetic monopole with charge $(1,0)$. The remaining terms in the 4d limit therefore have charge
\be
	\left({X_{\gamma_3+\gamma_4}}\right)^{1/2}
	+
	\left(\frac{1}{X_{\gamma_3+\gamma_4}}\right)^{1/2}
	+
	\left({X_{\gamma_2-\gamma_1+\gamma_{D0}}}\right)^{1/2}
	=
	\sqrt{X_{(0,-2)}}+ \sqrt{X_{(0,2)}}+\sqrt{X_{(-2,2)}} 
\ee
These agree with the prediction for the 4d theory computed in \cite[eq. (10.33)]{Gaiotto:2010be} with the identifications $X_{GMN}=X_{-\gamma_2}$ and $Y_{GMN}=X_{-\gamma_1+\gamma_{D0}}$. It also agrees with the limit from the q-Painlev\'e Hamiltonian to the Hamiltonian of the open relativistic Toda chain \cite{Bershtein2017}, which is related to the 4d theory in the work \cite{Cirafici:2017iju}.

\item Relation \eqref{eq:cluster-vs-voros-F0} shows that Voros symbols associated to neighboring branch points of type 1 in Section \ref{sec:qFG} are cluster variables, while those of type 2 are not. This is consistent with the triviality of the cross ratios of Wronskians \eqref{eq:qWronskiansCase2} in this latter case, and with the fact that only coordinates of type 1 are associated to a saddle of the Stokes diagram.
\end{enumerate}

\appendix

\section{Monodromy computation via exponential networks}\label{app:localF0}

In this appendix we provide an alternative derivation of the result \eqref{eq:localF0-monodromy-gamma} using the framework of nonabelianization in exponential networks \cite{Banerjee:2018syt}.
The exponential network for the mirror curve of local $\IF_0$ was studied in \cite{Banerjee:2020moh}, although in a different region of moduli space than the one considered here. 
For ease of comparison with \eqref{eq:localF0-monodromy-gamma} we briefly set up the computation from scratch for the same choice of moduli adopted in the present work. A connection to the analysis of \cite{Banerjee:2020moh} can be established naturally by moving between the two points in moduli space, as described in Section \ref{sec:local-F0}.

\begin{figure}[h!]
\begin{center}
\includegraphics[width=0.75\textwidth]{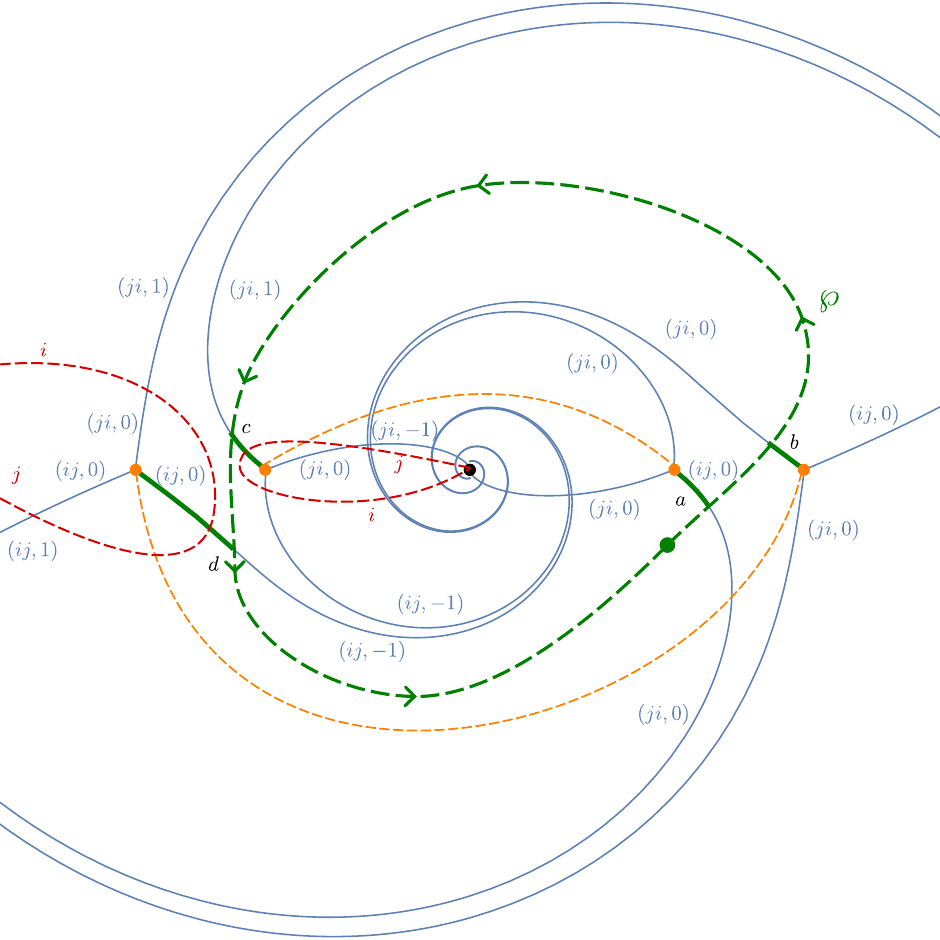}
\caption{Exponential networks for local $\IF_0$ with a choice of trivialization and soliton labels for trajectories.}
\label{fig:exp-net-computation}
\end{center}
\end{figure}

The exponential network is shown in Figure \ref{fig:exp-net-computation}, here we adopt a slightly different trivialization than in the main text (this will not affect the result).
Only primary trajectories are shown, and their labels correspond to the orientation that flows out from the branch points.
We follow the convention that a trajectory of type $(ij,n)$ corresponds to a solution of the ODE
\be
	(\log y_j - \log y_i + 2\pi i \, n) \frac{d\log x}{dt} = e^{i\vartheta}
\ee
for a network at phase $\vartheta$.
Sheets are labeled by $i,j$ and the match with labeling $+/-$ adopted in the main text can be carried out (although the choice of trivialization is different, as already stressed).

Logarithmic cuts are shown as red dashed lines beginning and ending at punctures (black dots), while square-root cuts are shown as orange dashed lines beginning and ending at branch points (orange dots). 

We consider a path ${\mathcal{P}}$ shown as a green dashed line, with basepoint denoted by a green dot.
The parallel transport along this path picks up contributions from detours at four Stokes lines, one from each of the branch points. Detour paths are highlighted by thick solid green segments of the network trajectories.
The path does not cross any branch cuts, and features diagonal matrices between detours, which we omit to lighten notation.
We shall work in conventions where a Stokes matrices are 2 by 2, with sheet $i$ representing the first basis vector in the space of solutions, and label $j$ the second basis vector.
The parallel transport is therefore
\be
	F({\mathcal{P}},\vartheta) = 
	\(\begin{array}{cc}
	1 & X_a \\
	0 & 1
	\end{array}\)
	\(\begin{array}{cc}
	1 & 0 \\
	X_b & 1
	\end{array}\)
	\(\begin{array}{cc}
	1 & 0 \\
	X_c & 1
	\end{array}\)
	\(\begin{array}{cc}
	1 & X_d \\
	0 & 1
	\end{array}\)
	=
	\(\begin{array}{cc}
	1 + X_a X_b + X_a X_c & \star \\
	\star & 1 + X_cX_d + X_b X_d
	\end{array}\)
\ee
Due to the fact that we suppressed the diagonal transport, it has been left understood that the terms $1$ in the diagonal elements are the lifts of ${\mathcal{P}}$ to sheets $i,j$ denoted ${\mathcal{P}}^{(i)},{\mathcal{P}}^{(j)}$ respectively.
The trace is therefore
\be
	{\rm Tr} \, F({\mathcal{P}},\vartheta) 
	=
	X_{{\mathcal{P}}^{(i)}} + X_{{\mathcal{P}}^{(i)} ab} + X_{{\mathcal{P}}^{(i)} ac}
	+ X_{{\mathcal{P}}^{(j)}} + X_{{\mathcal{P}}^{(j)} cd} + X_{{\mathcal{P}}^{(j)} bd}
\ee
Now observe that the soliton charges have the following types
\be
	a: \ (ij,0)\,,
	\qquad
	b: \ (ji,0)\,,
	\qquad
	c: \ (ji,1)\,,
	\qquad
	d: \ (ij,-1)\,.
\ee
In taking the trace over the infinite-dimensional vector space of vacua labeled both by polynomial and logarithmic indices, only paths that close up on the logarithmic branch should be retained.
The overall shift of logarithmic index is as follows for the six terms found in the polynomial trace
\be
	\underbrace{X_{{\mathcal{P}}^{(i)}}}_{0} + 
	\underbrace{X_{{\mathcal{P}}^{(i)} ab}}_{0} + 
	\underbrace{X_{{\mathcal{P}}^{(i)} ac}}_{1}+ 
	\underbrace{X_{{\mathcal{P}}^{(j)}}}_{0} + 
	\underbrace{X_{{\mathcal{P}}^{(j)} cd}}_{-1} + 
	\underbrace{X_{{\mathcal{P}}^{(j)} bd}}_{0}
\ee
Therefore the only terms that survive in the total trace are
\be\label{eq:ExpNetTr}
	X_{{\mathcal{P}}^{(i)}} + X_{{\mathcal{P}}^{(i)} ab} 
	+ X_{{\mathcal{P}}^{(j)}} + X_{{\mathcal{P}}^{(j)} cd} 
\ee
We next determine homology classes of the four paths.
Comparing from Figure \ref{fig:localF0-basis} we see that ${\mathcal{P}}^{(i)} - {\mathcal{P}}^{(j)}= \gamma_1+\gamma_2-D0$. The appearance of $D0$ can be understood by noting that both $\gamma_{1},\gamma_2$ are lifts of the paths in Figure \ref{fig:localF0-basis} to sheets $(i,n)$ and $(j,n+\Delta n)$ such that the lifts match at the endpoints, where $\Delta n=0$. However, since both cycles end on branch points encircled by a logarithmic cut, if $\Delta n=0$ at the endpoints it must be nonzero when crossing the logarithmic cuts. In the case at hand we have $\Delta n=1$ and therefore $\gamma_1+\gamma_2$ is a cycle that corresponds to lifts of $\mathcal{P}^{(i)}$ and $\mathcal{P}^{(j)}$ to different logarithmic sheets, giving 
\be\label{eq:gamma1-gamma2-D0-appendix}
	X_{\gamma_1+\gamma_2} = X_{\mathcal{P}^{(i,n)} - \mathcal{P}^{(j,n+1)}} = X_{\mathcal{P}^{(i,n)} - \mathcal{P}^{(j,n)}} \, e^{\frac{1}{\hbar}\oint (2\pi i)\,\frac{dx}{x}} = X_{\mathcal{P}^{(i,n)} - \mathcal{P}^{(j,n)}} \, e^{-4\pi^2} = X_{\mathcal{P}^{(i,n)} - \mathcal{P}^{(j,n)}} \, X_{\gamma_{D0}}
\ee
A simpler alternative is to use $\gamma_3,\gamma_4$ instead of $\gamma_1,\gamma_2$: since these cycles are attached to standard (non-logarithmic) branch points, we have simply
\be
	X_{-\gamma_3-\gamma_4} = X_{\mathcal{P}^{(i,n)} - \mathcal{P}^{(j,n)}} 
\ee
but using the relation \eqref{eq:D0-in-F0} recovers once again \eqref{eq:gamma1-gamma2-D0-appendix} since
\be
	X_{-\gamma_3-\gamma_4} = X_{\gamma_1+\gamma_2-\gamma_{D0}}\,.
\ee

By $\IZ_2$ symmetry of the covering we can furthermore deduce that
\be
	{\mathcal{P}}^{(i)} = -{\mathcal{P}}^{(j)} = -\frac{1}{2} (\gamma_3+\gamma_4)\,.
\ee
Then noting that 
\be
	ab = \gamma_4\,,\qquad cd = \gamma_2
\ee
where the identification follows from the deformation involved in going from the moduli in Figure~\ref{fig:localF0-basis} to the moduli considered here. (Branch points 1 and 2 collide and open up again along the real axis in the process.)
Therefore we can rewrite terms in \eqref{eq:ExpNetTr} as 
\be\label{eq:ExpNetGamma}
\begin{split}
	X_{{\mathcal{P}}^{(i)}} = X_{-\frac{1}{2} (\gamma_3+\gamma_4)}\,,
	& \quad
	X_{{\mathcal{P}}^{(i)} ab} =  X_{\frac{1}{2} (\gamma_1-\gamma_2-\gamma_{D0}) +(\gamma_2+\gamma_4)}\,,\quad
	\\
	X_{{\mathcal{P}}^{(j)}} = X_{\frac{1}{2} (\gamma_3+\gamma_4)}\,,\quad
	&
	X_{{\mathcal{P}}^{(j)} cd} = X_{\frac{1}{2} (-\gamma_1+\gamma_2+\gamma_{D0})}\,.
\end{split}
\ee
The sum of these terms reproduces \eqref{eq:localF0-monodromy-gamma}.

\bibliography{Biblio}
\bibliographystyle{alph}

\end{document}